\documentclass[a4paper,10pt]{article}
\usepackage{amsmath,amssymb,amsfonts,amsthm}
\usepackage{geometry}
\usepackage[english]{babel}
\usepackage{graphicx}
\usepackage{enumitem}
\usepackage{float}
\usepackage{aliascnt}

\usepackage{color} 
\usepackage[dvipsnames]{xcolor} 
\usepackage{csquotes} 
\usepackage{appendix} 
\usepackage{array} 
\usepackage{cancel}
\usepackage{soul}

\usepackage{tikz}
\usetikzlibrary{external}
\usepackage{pgfplots}
\usepgfplotslibrary{colormaps}
\pgfplotsset{
    compat=newest,
    colormap={mycolormap}{color=(lightgray) color=(white) color=(lightgray)}
}

\usepackage{tensor}
\usepackage{braket}
\usepackage{physics}
\usepackage{mathtools}
\usepackage{pdfpages}
\usepackage{bm}
\usepackage{hyperref}
\usepackage{xcolor} 
\usepackage{quiver} 
\usepackage{faktor} 
\usepackage{combelow} 
\allowdisplaybreaks

\usepackage{fancyhdr} 
\usepackage{emptypage}
\usepackage{mathtools,todonotes} 
\usepackage{tikz-cd}  
\usepackage[all,cmtip]{xy} 
\usepackage{mathrsfs}

\theoremstyle{definition} 
    \newtheorem{definition}{Definition}[section]
    \newtheorem{notation}[definition]{Notation}
    \newtheorem{theorem}[definition]{Theorem}
    \newtheorem{proposition}[definition]{Proposition}
    \newtheorem{lemma}[definition]{Lemma}
    \newtheorem{corollary}[definition]{Corollary}
    \newtheorem{claim}[definition]{Claim}

    \newtheorem{remark}[definition]{Remark}
    \newtheorem{example}[definition]{Example}

\numberwithin{equation}{section} 

\usepackage[
  bibstyle=alphabetic,
  citestyle=alphabetic,
  useprefix,
  giveninits=true,
  minbibnames=99,
  maxbibnames=99,
  sorting=nyt,
  sortcites=true,
  backend=biber
]{biblatex}
\renewbibmacro{in:}{} 

\DeclareSourcemap{
  \maps[datatype=bibtex]{
    \map[overwrite]{ 
      \step[fieldsource=doi, final]
      \step[fieldset=url, null]
      \step[fieldset=eprint, null]
    }
    \map[overwrite]{ 
      \step[fieldsource=eprint, final]
      \step[fieldset=pages, null]
      \step[fieldset=eid, null]
      \step[fieldset=journal, null]
    }  
  }
}

\bibliography{bibliography}  
\newcommand{\no}[1]{\mathopen{:}#1\mathclose{:}}
\newcommand{\Ghat}{\widehat{\mathcal{G}}}
\newcommand{\Ghatzero}{\widehat{\mathcal{G}}}
\newcommand{\Jmode}{\mathrm{J}}
\newcommand{\Kmode}{\mathrm{K}}
\newcommand{\Pmode}{\mathrm{P}}
\newcommand{\Amode}{\mathrm{A}}
\newcommand{\Fmode}{\mathrm{F}}
\newcommand{\Emode}{\mathrm{E}}
\newcommand{\Gmode}{\mathrm{G}}
\newcommand{\Hmode}{\mathrm{H}}
\newcommand{\gamode}{\gamma}
\newcommand{\Phmode}{\Phi}
\newcommand{\Thmode}{\Theta}
\newcommand{\amode}{\mathrm{a}}
\newcommand{\bmode}{\mathrm{b}}
\newcommand{\cmode}{\mathrm{c}}
\newcommand{\wmode}{\mathrm{w}}
\newcommand{\umode}{\mathrm{u}}
\newcommand{\vmode}{\mathrm{v}}
\newcommand{\Diag}{\Delta}
\newcommand{\ADiag}{\Delta}
\newcommand{\BDiag}{\Delta}

\newcommand{\Ccharge}{\mathrm{Z}}
\newcommand{\charge}[1]{\chi_{{#1}}}
\newcommand{\hrep}{\rho}
\newcommand{\rep}{\mathrm{P}}
\newcommand{\Mrep}{\mathcal{H}}

\newcommand{\Ewaki}[1]{\mathcal{E}_{#1}}
\newcommand{\cfsppolarization}{\mathcal{B}}
\newcommand{\fieldspace}{\mathcal{F}}

\newcommand{\boundary}{\Sigma}
\newcommand{\corner}{\Gamma}

\newcommand{\vpdv}[2]{\frac{\delta{#1}}{\delta{#2}}}

\newcommand{\torus}{{T^2}}

\newcommand{\circl}{{S^1}}

\newcommand{\statespace}{\mathcal{H}}
\newcommand{\zz}{Z}
\newcommand{\levi}{\varepsilon}
\newcommand{\su}{\mathfrak{su}(2)}
\newcommand{\eigfun}[1]{\phi_{#1}}
\newcommand{\Lapg}{\Delta_g}
\DeclareFontFamily{U}{mathx}{}
\DeclareFontShape{U}{mathx}{m}{n}{<-> mathx10}{}
\DeclareSymbolFont{mathx}{U}{mathx}{m}{n}
\DeclareMathAccent{\widehat}{0}{mathx}{"70}
\DeclareMathAccent{\widecheck}{0}{mathx}{"71}

\title{\textbf{Corner Quantization of 4D $BF$ Theory}}

\author{Giovanni Canepa$^\alpha$, Alberto S.\ Cattaneo$^\beta$, Filippo Fila--Robattino$^\beta$, Timon Leupp$^\beta$\footnote{
Emails: 
\href{mailto:cattaneo@math.uzh.ch}{cattaneo@math.uzh.ch},\quad
\href{mailto:giovanni.canepa.math@gmail.com}{giovanni.canepa.math@gmail.com},\quad
\href{mailto:filippo.filarobattino@math.uzh.ch}{filippo.filarobattino@math.uzh.ch},\quad
\href{mailto:timon.leupp@math.uzh.ch}{timon.leupp@math.uzh.ch} 
}}

\date{}

\begin{document}
\maketitle
\begin{center}
$\prescript{\alpha}{}{}$INFN, Sezione di Firenze \\ Via Sansone 1, 50019 Sesto Fiorentino (FI), Italy \\
\end{center}
\begin{center}
$\prescript{\beta}{}{}$Universität Zürich, Institut für Mathematik, \\Winterthurerstrasse 190, 8057 Zürich (ZH), Switzerland\\
\end{center}

\begingroup
\renewcommand{\thefootnote}{}  
\footnotetext{We acknowledge partial support of the SNF Grant No.\ 200021\_227719 and of the Simons Collaboration on Global Categorical Symmetries. This research was (partly) supported by the NCCR SwissMAP, funded by the Swiss National Science Foundation. This article is based upon work from COST Action 21109 CaLISTA, supported by COST (European Cooperation in Science and Technology)(www.cost.eu), MSCA-2021-SE-01-101086123 CaLIGOLA, and MSCA-DN CaLiForNIA -101119552. FFR acknowledges funding from the EU project Caligola HORIZON-MSCA-2021-SE-01, Project ID: 101086123.}
\endgroup

\begin{abstract}
This note studies the quantized corner structure of four-dimensional $BF$ theory, classifies the associated free and physical corner algebras and constructs possible representations. In the abelian case, for arbitrary closed oriented surfaces and in the presence or absence of a cosmological term, explicit presentations of the corner algebras are obtained in terms of generators and relations, identifying them as infinite-dimensional oscillator-type Lie algebras with an abelian summand. A construction of infinite families of simple modules via bosonic Fock space representations is provided. In the non-abelian case on the torus, the corner algebras are described as quotients constructed from the central extensions of double-loop algebras over certain non-semisimple Lie algebras. A construction of infinite families of simple Fock-type modules of the free corner algebra via an induced module procedure is also provided. The resulting modules descend only trivially to the physical quotient, revealing an obstruction in the present construction in the non-abelian setting.
\end{abstract}

\pagestyle{fancy}
\pagestyle{plain}
\tableofcontents

\section*{Introduction}
\addcontentsline{toc}{section}{\protect\numberline{}Introduction}

To describe a (quantum) field theory on a spacetime with boundary and corners, one expects to assign compatible geometric and algebraic data to each stratum of increasing codimension. At the classical level, the bulk theory determines, via the variational principle, a symplectic structure on the space of boundary fields, as described by the construction of Kijowski and Tulczyjew \cite{KT1979}. In the presence of corners, residual data are induced from integration by parts, which ensures certain compatibility conditions between the theories in different codimensions. Under suitable assumptions, these data can be described in terms of Dirac geometry \cite{BCCS26} and, in favorable cases, give rise to Poisson or, more generally, ${P}_\infty$-structures on the space of corner fields \cite{CanCatCorner23}. 

At the quantum level, these structures are paralleled by the assignment of algebraic objects to each stratum of spacetime. To the bulk, one associates a partition function; to the boundary, one associates a space of states; and to the corner, one associates a corner algebra. These assignments are expected to satisfy extensions of the Atiyah--Segal axioms \cite{Atiyah88,segal_definition_1988}. For example, the corner algebra $A_\corner$ associated to a closed codimension-2 manifold $\corner$ should act on the state spaces $\statespace_\boundary$ of compact codimension-1 manifolds $\boundary$ when $\partial \boundary = \corner$.  Hence, these state spaces turn into modules of the corner algebra. 
Furthermore, if the boundary is given by a decomposition of manifolds $\Sigma_1 \sqcup_\corner \Sigma_2$, gluing along $\Gamma$ should be implemented at the quantum level by a tensor product over the corner algebra. We therefore expect the state spaces to satisfy the gluing formula: $\mathcal{H}_\Sigma \cong \mathcal{H}_{\Sigma_1}\otimes_{A_\corner}\mathcal{H}_{\Sigma_2}$.

These features make the extension to corners very useful. For example, state spaces are immediately restricted by the representation theory of the respective corner algebra. Furthermore, the gluing formula allows one to determine the state spaces of codimension-1 manifolds with complicated topology from more elementary pieces. This is especially relevant when working with theories in higher dimensions.

The BV-BFV formalism and its higher-codimension extensions (often referred to as BF$^k$V, $k$ being the codimension) provide a framework that allows one to treat both the classical and quantum theories in a unified way, generalizing the Faddeev--Popov \cite{faddeevFeynmanDiagramsYangMills1967} and BRST methods \cite{becchiRenormalizationGaugeTheories1976,Tyutin:1975qk}. Classically, the BV formalism, introduced in \cite{BATALIN198127,osti_5425625}, embeds the space of fields together with gauge symmetries into a graded symplectic supermanifold equipped with a symplectic form of degree $-1$ and an action functional whose Hamiltonian vector field is cohomological and whose degree-0 cohomology is interpreted as the space of classical observables.\footnote{A review of this approach can be found in \cite{Mnev_2019}} The BFV formalism, introduced in \cite{BATALIN1983157,BATALIN77_smat}, provides the Hamiltonian counterpart, yielding a cohomological resolution of the reduced phase space on the boundary.\footnote{In \cite{SchaetzTH, Stasheff1997}, such formalism was employed for the reduction of general coisotropic submanifolds by means of a cohomological resolution.} Under the appropriate assumptions, the BV formalism can be relaxed to the case where the spacetime $M$ admits a boundary, in such a way that the failure of the BV structure in the bulk to satisfy certain axioms induces a compatible BFV structure on the boundary \cite{Cattaneo:2012qu,Cattaneo:2011syo}. This procedure can be iterated to strata of higher codimension, yielding, in favorable cases, a compatible BV-BF$^k$V structure. On codimension-2 strata, the BF$^2$V structure encodes the classical corner data in terms of higher Poisson structures. Upon quantization, these structures give rise to $A_\infty$-algebras acting on boundary state spaces, thereby realizing the expected algebraic structure at corners. In particular, the corner algebra can be understood as arising from the quantization of a $P_\infty$-algebra associated to the space of corner fields, typically determined by the BF$^2$V action and a choice of polarization. 

The first steps of this program have been developed in \cite{CanCatCorner23}, where the BF$^2$V structures of Chern–Simons theory and 4-dimensional $BF$ theory (with and without cosmological term) were described, and the corresponding corner algebras were identified implicitly as quantizations of Poisson structures on spaces of corner fields.
 
The main focus of this paper is $BF$ theory in four dimensions. Since its introduction in \cite{schwarz_partition_1979,Horowitz_1989, Blau:1989bq}, $BF$ theory has been the subject of intense research. Similarly to other topological field theories, it can be used to define invariants of knots and manifolds, studied in \cite{Cattaneo:2000mc,Cattaneo:1995tw,Cotta-Ramusino:1994nhr, Baez:1995ph}, for example. However, unlike many other topological theories, it can be defined consistently in any dimension. In dimension $d=2$, it is related to $(1+1)$-dimensional Yang–Mills theory. In dimension 
$d=3$, it is a special case of Chern--Simons theory. If the Lie algebra defining 
$BF$ is $\mathfrak{g}=\mathfrak{so}(1,2)$ (or 
$\mathfrak{so}(3)$), corresponds to $1+2$ (or Euclidean) gravity with cosmological constant $\Lambda$ in the coframe formulation.\footnote{In dimensions 3 and 4, one can add a so-called cosmological term to the $BF$ action~\eqref{eq:BFaction} assuming the Lie algebra has an invariant inner product. In this context, one speaks of the original theory as pure $BF$.}
In dimension $d=4$, $BF$ theory with $\mathfrak{g}=\mathfrak{so}(1,3)$ (or $\mathfrak{so}(4)$) is related to $1+3$ (or Euclidean) gravity with cosmological constant $\Lambda$ in a very non-trivial way \cite{Plebanski:1977zz}.
The precise connection of the respective $P_\infty$-structure describing the space of corner fields was investigated in \cite{CanCatCorner23}. One could, therefore, hope to infer state spaces for gravity from those constructed for 4-dimensional $BF$ theory.

In this work, we first classify the free and physical corner algebras derived in \cite{CanCatCorner23} associated with 4-dimensional abelian $BF$ theory, with and without cosmological term, on an arbitrary closed oriented surface, and give explicit presentations in terms of generators and relations. More precisely, we show that these algebras can be described in terms of infinite-dimensional oscillator algebras, together with an additional abelian summand, and we prove isomorphism theorems both before and after reduction by the constraints.
We further construct infinite families of simple modules obtained from bosonic Fock space representations, which can serve as physically admissible state spaces. These constructions apply, in particular, to the case of the torus and extend naturally to arbitrary closed oriented surfaces, where the abelian summand is related to the first cohomology group of the surface. In principle, the resulting data provide the necessary input to attempt a verification of the gluing formula.

The results are also in agreement with \cite[Section 4.1]{fliss_entanglement_2025}, where the authors investigate abelian $p$-form $BF$ theory in $d$ dimensions from the perspective of extended Hilbert spaces and edge/corner modes and derive a current algebra and respective modules. See Remarks \ref{rem:relation_to_current_algebra_1} and \ref{rem:relation_to_current_algebra_2} for more details on the connection between the current algebra and the physical corner algebra.

In the non-abelian case, we study the free and physical corner algebras of $BF$ theory on the torus. We classify the corresponding free corner algebras as central extensions of double-loop algebras over the isochronous (and non-semisimple) Galilean Lie algebra and describe the physical corner algebras as suitable algebra quotients. We use the induced module construction based on a suitable decomposition of the Lie algebra, leading to representations of the free corner algebra in terms of second-order differential operators on spaces of polynomials. These modules exhibit desirable structural properties, such as a grading and a vacuum vector, and can be viewed as analogs of highest-weight representations in this setting.
However, we find that the constraint ideal acts non-trivially on these modules, so that they only descend trivially to the physical corner algebra. This indicates an obstruction in the present construction in the non-abelian case, and it remains an open question whether the method can be adapted---possibly by modifying the decomposition or the choice of polarization---to construct genuine modules for the physical corner algebra.
Most of the results first appeared in the thesis of T.L. \cite{Leupp2025}.

\thispagestyle{plain}
\subsection*{Structure of the Paper}
\addcontentsline{toc}{subsection}{\protect\numberline{}Structure of the Paper}
Section~\ref{sec:corner} reviews the ideas behind the classical and quantum BV-BFV formalism, its extension to corners, and its connection to the classical geometric data associated to boundaries and corners. Furthermore, the examples of Chern--Simons theory and 4-dimensional $BF$ theory are discussed.
In Section~\ref{sec:abelianontor}, the case of abelian $BF$ theory on the torus is treated, and 
in Section~\ref{sec:abelian_BF_on_any_surface}, the abelian analysis is generalized to arbitrary closed oriented Riemannian surfaces.
Lastly, Section~\ref{sec:nonabtorus_corner_algebra} contains the study of the free and physical corner algebras in the non-abelian theory on the torus, and Section~\ref{sec:nonabtorus_corner_modules} contains the construction of families of modules for the free corner algebras.

\subsection*{Acknowledgments}
\addcontentsline{toc}{subsection}{\protect\numberline{}Acknowledgments}
T.L.\ would like to thank Denis Nesterov for helpful remarks regarding the representation theory of infinite-dimensional Lie algebras. We thank Stathis Vitouladitis for pointing out references to previous work on the corner algebra in the abelian case from a different perspective.

\pagestyle{fancy}
\section{Corner Structure in the BV-BFV Formalism}\label{sec:corner}
\subsection{The Classical BV-BFV Formalism}
In this section, we introduce the classical and quantum BV formalism. We start from the following definition.
\begin{definition}  \label{d:BFkVmanifold}
    A \textbf{BF$^k$V manifold}
    for $k\in\mathbb{N}$ is a quadruple $(\fieldspace,\omega,S, Q)$, where
    \begin{itemize}
        \item $\fieldspace$ is a $\mathbb{Z}$-graded manifold\footnote{Throughout this paper, we assume the Grassmann parity to be the mod-2 reduction of the $\mathbb{Z}$-grading.}
        \item $\omega\in\Omega^2(\fieldspace)_{k-1}$ is a symplectic form of degree $k-1$
        \item $S\in C^\infty(\fieldspace)_{k}$ is a function of degree $k$ called the BF$^k$V action satisfying the Classical Master Equation (CME) $\{S,S\}=0$, where $\{\cdot,\cdot\}$ is the Poisson bracket induced by $\omega$
        \item $Q=\{S,\cdot\}$ is a vector field of degree 1 which is cohomological, i.e.\ $Q^2=0$, by virtue of the CME
    \end{itemize}
\end{definition}
In the special case $k=0$, we will write BV = BF$^0$V. 
Furthermore, we will also sometimes call the BF$^k$V manifold a BF$^k$V structure instead.

Considering a field theory on a stratified manifold, one associates a BF$^k$V manifold to the geometric data of the theory on a closed codimension-$k$ stratum.
Allowing some of the properties of a BF$^k$V manifold to hold up to boundary, i.e.\ defining a lax BF$^k$V manifold, we can also connect two or more consecutive strata in order to form a (fully) extended classical BV-BFV formalism. This framework  assigns a lax BF$^k$V structure to every compact codimension-$k$ manifold $N$ together with some compatibility relation between the lax BF$^k$V and (strict) BF$^{k-1}$V structures of two manifolds $N$ and $N'$ whenever $\partial N=N'$.

For well behaved theories (like $BF$ or Chern--Simons theories), this notion usually allows one to proceed algorithmically from a Lagrangian field theory on the bulk, extended to a BV manifold, and produce a BF$^k$V manifold for each strata.
In order to fix the notation and avoid any confusion, we will denote compact codimension-1 manifolds (boundaries) with $\Sigma$ and closed codimension-2 manifolds (corners) with $\Gamma$. Furthermore, the BFV data naturally associated to $\Sigma$ will be denoted by $(\fieldspace_{\Sigma},\omega_{\Sigma},S_{\Sigma}, Q_{\Sigma})$ and correspondingly the BF$^2$V data associated  to $\Gamma$ will be denoted by $(\fieldspace_{\Gamma},\omega_{\Gamma},S_{\Gamma}, Q_{\Gamma})$.

From the BFV manifold associated to the boundary, we thus obtain a cohomological description of the reduced phase space $\mathcal{R}_\Sigma$, as remarked in the Introduction. In other words, it is possible to describe the space of functions on the physical space of fields on the boundary as 
 \begin{align}\label{e:RPSasCohomology}
     C^{\infty}(\mathcal{R}_{\Sigma}) \equiv H_{Q_{\Sigma}}^0(\fieldspace_{\Sigma}).
 \end{align}
 
From the BF$^2$V manifold associated to the corner, we instead describe the following classical picture cohomologically:
on a manifold with corners, a Lagrangian field theory associates a Dirac structure with support $L_{\Gamma}$ to corners, i.e.\ a vector subbundle $L_{\Gamma} \subset T C_{\Gamma} \times T^*C_{\Gamma}$ such that $\left\langle (v_1,\alpha_1),(v_2,\alpha_2) \right\rangle=\iota_{v_1}\alpha_2+\iota_{v_2}\alpha_1=0$ (isotropy) and $\left\langle \mathcal{L}_{v_1}\alpha_2, v_3\right\rangle + \left\langle \mathcal{L}_{v_2}\alpha_3, v_1\right\rangle + \left\langle \mathcal{L}_{v_3}\alpha_1, v_2\right\rangle=0$ (involutivity)
for sections $(v_i,\alpha_i)\in L_{\Gamma}$, where $C_{\Gamma}$ is a submanifold of the classical space of fields preserved by every $v\in L_{\Gamma}$.\footnote{More precisely, let $C_{\Gamma}$ be defined as the zero locus of a function $f \in C^\infty(\mathcal{F}_{\Gamma})$, then we require that $\mathcal{L}_v f =0$ for all $(v,\alpha) \in L_{\Gamma}$.} If $\Gamma$ is the boundary of a manifold $\Sigma$, the Dirac structure comes with a Hamiltonian space associated with $\Sigma$, representing the reduced phase space. 

This Dirac structure can be recovered from the BF$^2$V theory as follows. 
The $-1$-symplectic manifold $(\mathcal{F}_{\Gamma}, \omega_{\Gamma})$ is symplectomorphic to a  symplectic subbundle of $T^*[1]\mathcal{B}_{\Gamma}$ for a (non-unique) graded manifold $\mathcal{B}_{\Gamma}$. We call the choice of $\mathcal{B}_{\Gamma}$ a choice of polarization. After such a choice, every function 
$f\in C^{\infty}(\mathcal{F}_{\Gamma})$ can be reinterpreted as a multivector field on $\mathcal{B}_{\Gamma}$. In particular, the BF$^2$V action $S_{\Gamma}$ defines a 
multivector field 
\begin{align*}
    \pi_{\Gamma} = \sum_{i\in \mathbb{N}} \pi_i
\end{align*}
with $\pi_i$ a homogeneous $i$-multivector field of degree $2-i$. Thanks to the classical master equation of $S_{\Gamma}$, $\pi_\Gamma$ satisfies $\left[\pi_{\Gamma},\pi_{\Gamma}\right]=0$,\footnote{Here $[-,-]$ are the Schouten brackets on $\mathcal{B}_{\Gamma}$.} i.e.\  $\pi_\Gamma$ is a $P_{\infty}$-structure.

If we now choose the non-positive polarization, the classical part of the graph of the $P_{\infty}$-structure $\pi_\Gamma$ is a Dirac structure with support. In particular, the graph of $\pi_1$ defines the support $C_\Gamma$ and the graph of $\pi_2$ the actual Dirac structure $L_{\Gamma}$ (see  \cite{BCCS26} for more details). 
This leads to the following isomorphism of Poisson algebras
\begin{align}\label{e:cohomology_P_infty}
      C^{\infty} (C_{\Gamma}) \equiv H^0_{\pi_1}(\mathcal{B}_{\Gamma})
\end{align}
where the Poisson structure of $C^{\infty} (C_{\Gamma})$ is given by the one whose graph is $L_{\Gamma}$. Note that this coincides with the space of functions over the physical corner fields.
The BF$^2$V manifold and the associated $P_{\infty}$-structures were analyzed together with some examples in \cite{CanCatCorner23}.

\subsection{Quantization of Theories on Manifolds with Corners}
In \cite{Cattaneo_Mnev_Reshetikhin_2017}, the authors introduced a general perturbative quantization scheme in the BV-BFV formalism that works on manifolds with boundaries and is 
compatible with cutting and gluing.
The quantum BV-BFV framework has been successfully applied in a wide variety of field theories like: $BF$ theory, the Poisson sigma model \cite{Cattaneo_Mnev_Reshetikhin_2017},
Chern--Simons theory \cite{cmw:2020lle}, the relational symplectic groupoid \cite{CMW:2016zxn}, 2-dimensional Yang--Mills theory \cite{Iraso_Mnev_2019} and also Rozansky--Witten theory \cite{Moshayedi:2021pua}.

The formalism is expected to further extend to work on manifolds with corners such that it is compatible with cutting and gluing. The general theory of quantization with corners is still a work in progress, but a worked-out example discussing 2-dimensional Yang--Mills theory can be found in
\cite{Iraso_Mnev_2019}.
Since we are interested in the quantization of the corner data and its relation to the quantized boundary data, let us just focus on the corner and boundary and ignore the relation with the bulk.

\subsubsection{Quantization of the Reduced Phase Space}
On compact manifolds with a boundary, the quantization of the associated BFV manifold $\fieldspace_{\Sigma}$ with BFV action $S_\Sigma$ can be tackled using geometric quantization \cite{Cattaneo_Mnev_Reshetikhin_2017} as follows. We first choose a foliation by Lagrangian submanifolds (a.k.a.\ a polarization) and define a vector space $V_{\Sigma}$ as the space of polarized sections of a line bundle together with a Lie algebra morphism on a Poisson subalgebra $\mathcal{O}_{\Sigma}\subset C^\infty(\fieldspace_{\Sigma})$ of quantizable functions
\begin{align*}
    \mathfrak{q}\colon (\mathcal{O}_{\Sigma}, \{\cdot,\cdot\}) \to (\mathrm{End}(V_{\Sigma}), [\cdot,\cdot]).
\end{align*}
Then, one quantizes the BFV action $S_{\Sigma}$ to an operator $\Omega_{\Sigma} = \mathfrak{q}(S_{\Sigma})$ which must yield, if possible, a differential $[\Omega_{\Sigma},\cdot]$ in $\mathrm{End}(V_{\Sigma})$. The cohomology in degree 0 of the so constructed cochain complex $(V_{\Sigma},\Omega_{\Sigma})$ is then taken to be a quantization of the reduced phase space $\mathcal{R}_\boundary$, i.e.\ a quantization of the physical space of fields on the boundary
\begin{align}\label{def:physical_state_space_assoc_to_boundary}
    \mathcal{H}_{\Sigma} = H^0_{\Omega_{\Sigma}}(V_{\Sigma}).
\end{align}
We refer to this as the \textbf{state space} associated with $\Sigma$. This last equation is the quantum counterpart of Equation~\eqref{e:RPSasCohomology}.

\subsubsection{Extension to Corners}\label{sec:bv-bfbv_to_corners}
One expects the quantum BV-BFV framework to extend to compact manifolds with a corner by assigning 
to a closed codimension-2 manifold $\corner$ an $A_\infty$-algebra ${A}_\corner$, called the \textbf{corner algebra}, defined by the deformation quantization of the $P_\infty$-algebra ${P}_\corner$.\footnote{In many cases, one can choose a suitable polarization such that the structure trivializes after the first two brackets, see \cite[Section 3.3. \& 3.4.]{CanCatCorner23} which are reviewed in~\ref{sec:CSbf2v} and~\ref{sec:BFbf2v}.} ${P}_\corner$ is determined by the BF$^2$V structure induced on the corner by a choice of polarization. Furthermore, whenever $\partial \boundary = \corner$ for a compact codimension-1 manifold $\boundary$, it should assign an action of $A_\corner$ on the quantized BFV space $V_\Sigma$, compatible with $\Omega_\Sigma$ in a way such that the physical space of states $\mathcal{H}_{\Sigma}$ is turned into an $A_{\Gamma}$-module.

This is expected to be compatible with cutting and gluing in the following way: given a decomposition $\Sigma=\Sigma_1\sqcup_\Gamma\Sigma_2$, the state spaces satisfy
\begin{equation}\label{eq:gluing}
    \mathcal{H}_\Sigma \cong \mathcal{H}_{\Sigma_1}\otimes_{A_\corner}\mathcal{H}_{\Sigma_2}\,.
\end{equation}
An example in which this is worked out, in the general setting of the quantum BV-BFV formalism, for a theory is the paper \cite{Iraso_Mnev_2019}, mentioned previously.
\begin{example}\label{ex:torus}
In a 4-dimensional theory, an example of a corner is the 2-dimensional torus $\torus\coloneqq\circl\times\circl$.
In this example, the corner algebra $A_\torus$ acts on the state space $\mathcal{H}_{\bm{\torus}}$ associated to the 3-dimensional solid torus $\bm{\torus}$, a 3-dimensional manifold with boundary. One could then glue two solid tori together along their common boundary and obtain the state space on the resulting lens space 
using the gluing formula~\eqref{eq:gluing}.
\end{example}

In the present article, we adopt a simplified version where the corner algebra is not recovered by quantizing the $P_\infty$-algebra $P_\Gamma$ into $A_\Gamma$. Instead, 
we directly quantize the physical space of corner fields defined by the r.h.s.\ of Equation \eqref{e:cohomology_P_infty} to obtain the physical corner algebra, also denoted $A_\Gamma$. This can be achieved by first quantizing the Poisson structure induced by $\pi_2$ and then by recovering the physical corner algebra as the quotient of the result with respect to a quantization of the ideal encoding the constraints of the theory, codified by $\pi_1$ (cf.\ Section \ref{sec:quantization_scheme_of_physical_corner_algebra}).
The algebra does not act on the cohomological complex defined by the quantization of the BFV space, but directly on its cohomology $\mathcal{H}_{\Sigma}$, representing the quantization of the physical space.
Furthermore, also the gluing formula is expected to hold in this simplified setting.

\subsection{Example: Chern--Simons Theory}
In this example, we want to outline some of the main ideas that are used to construct the corner algebra and subsequently find representations thereof. To this end, we examine Chern--Simons theory with Lie algebra $\su$ (see \cite{Freed:1992vw} for an extensive discussion). 
First, we introduce the usual description on the bulk, assuming the spacetime has no boundary. We then specialize to the corner description, which was worked out in detail in \cite{CanCatCorner23} and quantize the associated affine Poisson algebra. Subsequently, we develop a convenient description of the corner algebra in terms of Fourier modes. As expected, modules of the quantized corner algebra are given by representations of the affine Lie algebra $\widehat{\mathfrak{sl}}(2)$, thereby providing possible state spaces on surfaces that have a circle as their boundary.

\subsubsection{\texorpdfstring{BF$^2$V}{bf2v} Structure} \label{sec:CSbf2v}
Let $M$ be a closed, oriented 3-manifold and fix an $SU(2)$-bundle $P$ over $M$. From low-dimensional homotopy theory, one knows that the bundle has to be trivial. Thus, the space of connection 1-forms $\mathcal{A}_P$ is simply isomorphic to $\Omega^1(M)\otimes\su$ by choosing a global section. Let us also fix an invariant, non-degenerate inner product on $\su$ denoted by $(\cdot,\cdot)$. The action functional of Chern--Simons theory is then defined to be
\begin{equation*}
    S_M[A]=\frac{1}{2}\int_M(A \overset{\wedge}{,}\dd A) +\frac{1}{3}(A \overset{\wedge}{,}[A \overset{\wedge}{,}A])\,,
\end{equation*}
where $A\in\Omega^1(M)\otimes\su$.  
Furthermore, the notations $(\cdot \overset{\wedge}{,}\cdot)$ and $[\cdot \overset{\wedge}{,}\cdot]$ denote the wedge product of the underlying forms combined with the respective operations on the Lie-algebra factor.
From now on, we will consider the wedge product and Lie-algebra pairing to be implicit unless it is ambiguous. 

We are interested in the corner algebra of this theory. To this end, we employ the BV-BFV formalism. We will closely follow \cite{CanCatCorner23} to examine the BF$^2$V structure.
Let $\corner \cong \circl$ be a potential corner of $M$. The corner data are then given by:
\begin{itemize}
    \item The 1-shifted symplectic space of corner fields: $\fieldspace_\corner = \left(C^\infty(\corner)[1]\otimes\su\right)\oplus\left(\Omega^1(\corner)\otimes\su\right)$ and $\omega_\corner = \int_{\corner}\delta c \,\delta A$, where the coordinate functions are labeled $(c,A)$ respectively.\footnote{To make this statement precise, one needs to choose a favorite setting of infinite-dimensional analysis. For example, one can employ the variational bicomplex formalism, where one can rigorously define coordinates for all jets of the fields and work with the associated Cartan calculus. However, this is out of the scope of this work. For a rigorous treatment of the variational bicomplex, see \cite{blohmann2024lft}.}
    \item The corner action:
    $S_\corner = \int_\corner \frac{1}{2}c\,\dd_{A}c$.
\end{itemize}
The symbol $ C^\infty(\corner)[1]$ denotes the 1-shifted sheaf of functions and $\dd_{A}$ denotes the covariant exterior derivative with respect to $A$.\footnote{The shift essentially just means the functions have odd parity, meaning they anticommute. For more information on graded geometry in field theory see for example \cite[Section 4.2.2.]{Mnev_2019}.}
To quantize, one possible polarization is $\fieldspace_\corner \cong T^*[1]\cfsppolarization$, the 1-shifted cotangent bundle of $\cfsppolarization \coloneqq \Omega^1(\corner)\otimes\su$. The 1-shifted cotangent bundle carries a canonical symplectic structure which, in particular, defines a Poisson bracket on $C^\infty(T^*[1]\cfsppolarization)$. This Poisson algebra of functions can itself be canonically identified with the algebra of multivector fields on $\cfsppolarization$ \cite[Section 2.2.]{CanCatCorner23}.
The multivector field $\pi$ for the corner action then defines a $P_\infty$-structure on $\cfsppolarization$ by virtue of the master equation. In the present case, one obtains the Poisson bivector field 
\begin{align*}
\pi_2 = \int_\corner \left( \frac{1}{2} \vpdv{}{A}\dd\vpdv{}{A} + \frac{1}{2}A\left[\vpdv{}{A},\vpdv{}{A}\right]\right)    
\end{align*}
on $\cfsppolarization$. In this context, $A$ is regarded as a coordinate on the latter.
The action of the Poisson bivector field on linear functionals on $\cfsppolarization$ defines the bracket
\begin{equation*}
    \left\{ \int_\corner fA,\int_\corner gA\right\}_2=\int_\corner(f\dd g+[f,g]A)\,.
\end{equation*}
Using the Leibniz rule, one can extend the bracket to locally polynomial functionals turning $C^\infty(\cfsppolarization)$ into an affine linear Poisson algebra denoted $P_\corner$, modulo subtleties regarding non-polynomial functions like power series.

\subsubsection{Quantization}
We now proceed to quantize this affine Poisson structure (APS) on $\cfsppolarization$ and find modules for the resulting algebra.
Recall that a linear Poisson structure (LPS) is equivalently described by the Kirillov–Kostant Poisson structure on the vector space $\mathcal{G}^*$, the dual (or in this case, some restricted version) of a given Lie algebra $\mathcal{G}$. This always holds in finite dimensions but also works for the case at hand. On the other hand, an APS on $\mathcal{G}^*$ defines a 1-dimensional central extension of $\mathcal{G}$ denoted by $\Ghat$. Finally, the original APS induces an LPS on $\Ghat^*$ \cite{Bhaskara_1990}.

In our example, $\mathcal{G}=C^\infty(\corner)\otimes\su$ and the 2-cocycle defining the central extension is given by $c(f,g) = \int_\corner f\dd g \,\Ccharge$, for a basis $\Ccharge$ of $\mathbb{R}$. However, instead of working with the algebraic dual, we will take $\Ghat^*$ to be the restricted dual consisting of linear forms defined by the pairing of Lie-algebra-valued forms and to be equipped with the induced Poisson bracket.

From the theory of deformation quantization, one expects the universal enveloping algebra (UEA) $\mathcal{U}(\Ghat)$ to provide a quantization of the Poisson algebra $C^\infty(\Ghat^*)=C^\infty(\cfsppolarization)\equiv P_\Gamma$ \cite{Gutt_1983}.\footnote{Technically,
we choose the subalgebra of polynomial functions to quantize. However, in infinite dimensions, the space of polynomials on $\Ghat_\Lambda^*$ is strictly bigger than the space of symmetric tensors of $\Ghat_\Lambda$. Therefore, we will actually quantize the latter. This is okay because, under some assumptions, one can recover those functions by finding a suitable topology on the quantization and taking the completion with respect to that topology. \cite[Section 2.2]{esposito_15}\label{fn:quanti}} In other words, the corner algebra of Chern--Simons theory (in this polarization) is $A_\corner = \mathcal{U}(\Ghat)$. 

Since the representation theory of a Lie algebra and that of its UEA are equivalent, we are left with the task of finding representations of $\Ghat$.\footnote{In this article, we will only discuss complex representations.}

\subsubsection{The Corner Algebra of Chern--Simons Theory}
The (complexified) Lie algebra describing the quantized corner structure of $\su$-Chern--Simons theory on a circle is given by the vector space:
\begin{equation*}
    \Ghat = C^\infty(\circl)\otimes \su_\mathbb{C} \oplus \mathbb{C}
\end{equation*}
with Lie bracket
\begin{equation*}
    [  f \oplus r, g \oplus s]_{\Ghat} = [f,g]\oplus \frac{1}{2\pi i}\int_{\circl} f \dd g\, \Ccharge\,.
\end{equation*}
Note that a rescaling of the 2-cocycle associated to the central extension produces isomorphic Lie algebras, because the extension only depends on the second Lie algebra cohomology group $H^2(\mathcal{G
},\mathbb{C})$. Fix a basis $\{t_\mu\}_{1\leq\mu\leq3}$ of $\su$, such that $(t_\mu, t_\nu) = \delta_{\mu,\nu}$ and $[t_\mu, t_\nu] = \levi_{\mu \nu}^\lambda t_\lambda$, where $\levi_{\mu\nu}^\lambda$ denotes the Levi-Civita symbol. 
Instead of the entire algebra, we will consider the Lie subalgebra that consists of finite Fourier expansions
    \begin{align*}
        f &= \sum_{\substack{1 \leq\mu \leq 3 \\ m\in \mathbb{Z}}} f_{\mu m} \left( t_\mu\otimes e^{im\theta}\right) \quad \text{for } f_{\mu m}\in \mathbb{C}\,,
    \end{align*}
where $\theta$ is the coordinate
on the circle with orientation given by the volume form $\dd\theta$ and with only finitely many non-zero coefficients in the sum.
We refer to this Lie subalgebra as the  \textbf{finite Fourier mode algebra} and, by abuse of notation, denote it by the same symbol. Now that the vector space $\Ghat$ has countable dimension, we can choose a suitable vector space basis $\{\Jmode_{\mu m}\coloneqq t_\mu\otimes e^{im\theta}\vert m\in\mathbb{Z}, 1\leq \mu \leq 3\}.$ A quick computation yields the Lie bracket relations:
\begin{align*}
    [\Jmode_{\mu m},\Jmode_{\nu n}] &= [t_\mu,t_\nu]\otimes e^{i(m+n)\theta}+ \frac{1}{2\pi i}\int_{\circl} ((t_\mu,
    t_\nu)in e^{i(m+n)\theta}\dd \theta)\Ccharge\\
    &=\levi_{\mu \nu}^\lambda \Jmode_{\lambda m+n} + n\delta_{\mu,\nu}\delta_{m,-n}\Ccharge\,.
\end{align*}
This is precisely the structure of the affine Lie algebra $\widehat{\mathfrak{sl}}(2)$ without the derivation element. Its representation theory has been widely studied. 
For a starting point, consider \cite[Section 9.2]{kac_infinite_1983}, where the representation theory of Kac--Moody algebras is developed. Furthermore, in \cite{Wakimoto_1986}, an explicit example of a family of representations in terms of differential operators on an infinite space of polynomials parametrized by two complex numbers is given, which we will recall in the following.

Define the matrices 
\begin{align*}
    H= \begin{bmatrix}
        1 & 0\\ 0 &-1
    \end{bmatrix},\quad X= \begin{bmatrix}
        0 & 1\\ 0 &0
    \end{bmatrix} \quad\text{and}\quad Y= \begin{bmatrix}
        0 & 0\\ 1 &0
    \end{bmatrix},
\end{align*}
which form a basis of $\mathfrak{sl}(2)$. In other words, 
$ H = -2i\Jmode_3,  X=\Jmode_1-i\Jmode_2, Y=-\Jmode_1-i\Jmode_2$ in terms of the previous generators.
Consequently, for some formal variable $t$, $H(n)\coloneqq t^n\otimes H, X(n)\coloneqq t^n\otimes X, Y(n)\coloneqq t^n\otimes Y\text{ for all } n\in\mathbb{Z}\text{ together with } \Ccharge$
form a basis of $\widehat{\mathfrak{sl}}(2)$ (excluding the derivation element). 
Define the differential operators
\begin{align*}
    a_j\coloneqq Y_+(j)x_j-Y_-(j)\pdv{}{x_j}, \quad a_j^*\coloneqq Y_+(j)\pdv{}{x_j}+Y_-(j){x_j} \quad \forall j\in\mathbb{Z} 
\end{align*}
on the space of polynomials $\mathbb{C}[\{x_j\}_{j\in\mathbb{Z}}]$ in infinitely many variables, where
$Y_+(j) =\begin{cases}
    1 \quad \text{if }j\geq0\\
    0 \quad \text{if }j<0
\end{cases}$
and 
$Y_-(j)=1-Y_+(j)$. Furthermore, introduce the normal products
\begin{equation*}
\begin{aligned}
     \no{a_ia_j^*}&= a_ia_j^*+Y_-(j)\delta_{i,j} \quad&&\forall i,j\in\mathbb{Z}\\
     \no{a_ia_j^*a_k^*} &=a_ia_j^*a_k^*+Y_-(i)(\delta_{i,j}a_k^*+\delta_{i,k}a_j^*)\quad &&\forall i,j,k\in\mathbb{Z}
\end{aligned}
\end{equation*}
and define $E(m)\coloneqq \sum_{j\in \mathbb{Z}} \no{a_{j+m}a_j^*}$ for all $m\in\mathbb{Z}$. Finally, introduce the normal products
\begin{equation*}
\begin{aligned}
    \no{E(m)a_j^*} &= E(m)a_j^*+ Y_-(m)a_{j-m}^* \quad &&\forall m,j\in\mathbb{Z},\\
     \no{E(m)a_j^*a_k^*} &= E(m)a_j^*a_k^*+Y_-(m)(a_{j-m}^*a_k^*+a_j^*a_{k-m}^*) \quad &&\forall m,j,k\in\mathbb{Z}
\end{aligned}
\end{equation*}
and define the differential operators 
\begin{align*}
    b_0\coloneqq0,\quad b_j\coloneqq jy_j, \quad b_{-j}\coloneqq -\pdv{}{y_j} \quad \forall j\in\mathbb{N}
\end{align*} 
on the space of polynomials $\mathbb{C}[\{y_j\}_{j\in\mathbb{N}}]$.

\begin{theorem}[{\cite[Theorem 1]{Wakimoto_1986}}]\label{thm:waki_free_field_rep_of_sl2}
Let $\mu,\nu\in\mathbb{C}$.
In the case when $\nu\neq0$, the assignment
\begin{align*}
    \Ccharge &\longmapsto -\left(2+\frac{\nu^2}{2}\right) , \\
    X(n) &\longmapsto a_{-n},\\
    H(n) &\longmapsto 2 E(-n)-\nu b_{-n}+(1-\mu)\delta_{n,0}\\
    Y(n) &\longmapsto \bigg(\mu-1-\bigg(\frac{\nu^2}{2}\bigg)n\bigg)a_n^*-\sum_{j\in\mathbb{N}}\no{ E(j)a_{j+n}^*}+\nu\sum_{j\in\mathbb{Z}}b_ja_{j+n}^*
\end{align*}
constitutes a representation of $\widehat{\mathfrak{sl}}(2)$ in terms of differential operators on the space of polynomials $\mathbb{C}[\{x_j,y_k\}_{j\in\mathbb{Z},k\in\mathbb{N}}]$. In the case when $\nu=0$, the same assignment constitutes a representation on the space of polynomials $\mathbb{C}[\{x_j\}_{j\in\mathbb{Z}}]$. 
\end{theorem}

Consequently, we obtain modules for the physical corner algebra $A_{\circl}\cong \mathcal{U}(\widehat{\mathfrak{sl}}(2))$ parametrized by the choice of complex numbers $\mu,\nu$; we know the action of $A_{\circl}$ on possible state spaces for surfaces that have a circle as their boundary.

\subsection{Example: \texorpdfstring{$BF$}{BF} Theory}
In this section, we introduce the BF$^2$V description and define the free and physical corner algebras of 4-dimensional $BF$ theory, following the work of \cite{CanCatCorner23}. We also discuss how to heuristically quantize a Poisson submanifold defined by constraint functionals.

\subsubsection{\texorpdfstring{BF$^2$V}{bf2v} Structure} \label{sec:BFbf2v}
Let $M$ be a closed, oriented 4-manifold, $G$ a Lie group with Lie algebra $\mathfrak{g}$ that has a non-degenerate invariant inner product.\footnote{In pure $BF$ theory, this is not necessary, as one can define $B$ to take values in the dual bundle instead.} Fix a principal $G$-bundle $P$ and denote $\mathcal{A}_P$ the space of principal connections. The action functional of $BF$ theory with a cosmological term is then defined by
\begin{equation}\label{eq:BFaction}
    S_M[A,B]=\int_M (B \overset{\wedge}{,} F_A ) + \frac{\Lambda}{2}(B\overset{\wedge}{,}B)\,,
\end{equation}
where $B\in\Omega^2(M,\mathrm{ad}P)$ is an adjoint-valued 2-form, $F_A$ is the curvature of the connection 1-form $A\in\mathcal{A}_P$ and $\Lambda \in \mathbb{R}$ is the cosmological constant. Pure $BF$ theory is the special case when $\Lambda=0$. As in the previous example, we will consider the wedge product and Lie-algebra pairing implicitly from now on. 

To describe the corner algebra of this theory, we employ the BV-BFV formalism. We will closely follow \cite{CanCatCorner23} to define the BF$^2$V structure and subsequently the $P_\infty$-structure.
Their results for $BF$ theory can be summarized as follows. Let $\corner$ be a closed, oriented 2-manifold, and we assume that the pullback bundle $\mathcal{A}_{P\vert_\corner}$ is trivial for simplicity, meaning isomorphic to $\Omega^1(\corner)\otimes\mathfrak{g}$. The corner data is then given by:
\begin{itemize}
    \item The 1-shifted symplectic space of corner fields: \begin{equation*}
    \fieldspace_{\corner} \coloneqq \left(\Omega^1(\corner) \oplus\Omega^2(\corner) \oplus \Omega^2[-1](\corner) \oplus\Omega^1[1](\corner)\oplus \Omega^0[1](\corner) \oplus \Omega^0[2](\corner)\right)\otimes \mathfrak{g}\,,
    \end{equation*} 
    and 
    \begin{equation*}
        \omega_\corner = \int_{\corner} \delta B \,\delta c+\delta \tau \,\delta A+\delta \phi \,\delta B^+\,,
    \end{equation*}
    where 
    the coordinate functions are labeled $(A,B,B^+,c,\tau,\phi)$ respectively. In the BV-BFV formalism for $BF$ theory, one has to include ghosts and ghosts-for-ghosts to take care of the gauge symmetry, explaining the additional fields (see \cite{CanCatCorner23} for details).
    \item The corner action: 
    \begin{align*}
S_{\corner}\coloneqq\int_{\Gamma} \bigg(&\frac{1}{2}B[c,c]+\tau(\dd_{A_0}c+ [a,c]) + \phi \left(F_{A_0}+ \dd_{A_0}a + \frac{1}{2}[a,a]+[c,B^+]\right)\\&+\Lambda\left(\frac{1}{2}\tau\tau+ B\phi\right)\bigg)\,,
    \end{align*}
    where we split $a \coloneqq A - A_0$ and $\dd_{A_0}$ is the exterior derivative twisted by the reference connection $A_0$. This splitting is mostly relevant for a non-trivial bundle, but to match notation with \cite{CanCatCorner23}, we keep it that way.
\end{itemize}
To find a suitable quantization, one should choose a polarization. One possible choice is to realize 
\begin{equation}\label{eq:choice_of_polarization}
    \fieldspace_\corner = T^*[1]\cfsppolarization, \quad \cfsppolarization \coloneqq \fieldspace_\corner\vert_{c=\phi=\tau=0}.
\end{equation}
The corresponding multivector field of the corner action is $\pi = \pi_1 + \pi_2$, where
\begin{align*}
    \pi_1 &= \int_\corner (F_A+\Lambda B)\vpdv{}{B^+}\,,\\
    \pi_2 &= \int_\corner \left(\frac{1}{2} B\left[\vpdv{}{B},\vpdv{}{B}\right] + \vpdv{}{a}\dd_{A_0}\vpdv{}{B}+a\left[\vpdv{}{a},\vpdv{}{B}\right]+B^+ \left[\vpdv{}{B^+},\vpdv{}{B}\right]+\frac{1}{2}\Lambda\vpdv{}{a}\vpdv{}{a} \right)\, .
\end{align*}
In other words, it defines a differential graded affine Poisson algebra. 

Ultimately, we are interested in the degree-0 cohomology induced by the differential $\pi_1$, because it describes the Poisson algebra of functions on the corner whose graph is the Dirac structure, as discussed in Section \ref{sec:corner} (also cf.\ Equation \eqref{e:cohomology_P_infty}). The degree-0 cohomology of $\pi_1$ is a Poisson algebra and is given by
\begin{equation*}
    H^0_{\pi_1}(\cfsppolarization)\equiv C^\infty(\cfsppolarization_{\mathrm{phys}})\cong \faktor{C^\infty(\cfsppolarization_0)}{I_{F_A+\Lambda B}}\,,
\end{equation*}
where $\cfsppolarization_{\mathrm{phys}}\coloneqq \{(A,B)\in \cfsppolarization_0 \vert F_A + \Lambda B = 0\}$ are the physical corner fields, $\cfsppolarization_0 \coloneqq\Omega^1(\corner)\otimes\mathfrak{g}\oplus \Omega^2(\corner)\otimes \mathfrak{g}$ is the ghost-degree zero part of the (unconstrained) corner fields and $I_{F_A +\Lambda B}$ is a Poisson ideal generated by $\int_\corner f(F_A+\Lambda B)$ for $f\in\Omega^0(\corner)\otimes \mathfrak{g}$. We will refer to the functionals that are in the ideal as constraints.
The space $\cfsppolarization_0$ is a Poisson submanifold with brackets defined on linear functionals on $\cfsppolarization_0$ by
\begin{alignat*}{2}
    &\left\{ \int_\corner \alpha a, \int_\corner \beta a \right\}_2 &&= \Lambda \int_\corner\alpha \beta\,,\\
    &\left\{ \int_\corner \alpha a, \int_\corner fB \right\}_2 &&= \int_\corner ( \alpha \dd_{A_0}f + [\alpha,f]a) \,,\\
    &\left\{ \int_\corner f B, \int_\corner gB \right\}_2 &&=\int_\corner [f,g]B\, .
\end{alignat*}

\subsubsection{Quantization} \label{sec:quantization_scheme_of_physical_corner_algebra}
To quantize the Poisson algebra $C^\infty(\cfsppolarization_{\mathrm{phys}})$, we need to understand how a Poisson submanifold is quantized with respect to the ambient Poisson manifold. This is a difficult question, but has been answered in \cite[Section 7.5.]{Cattaneo:2007er} and 
\cite{CATTANEO2007521} if there are no anomalies. However, this approach involves heavy technical machinery and is difficult to describe explicitly. Instead, we will take a less sophisticated approach. We will first quantize the degree-0 part $C^\infty(\mathcal{B}_0)$ and quotient by appropriate constraints afterwards. In other words, we will do the following:
\[\begin{tikzcd}
	{H^0_{\pi_1}(\cfsppolarization)\cong\faktor{C^\infty(\mathcal{B}_0)}{I_{F_A+\Lambda B}}} && {A_\corner \coloneqq \faktor{\mathcal{U}(\Ghat_{\Lambda})}{\mathcal{I}_{F_A(+\Lambda B)}}}
	\arrow["{\text{heuristic quant.}}", squiggly, from=1-1, to=1-3]
\end{tikzcd}\]
where
\[\begin{tikzcd}
	{C^\infty(\mathcal{B}_0)} && {\mathcal{U}(\Ghat_{\Lambda})} \\
	{I_{F_A+\Lambda B}} && {\mathcal{I}_{F_A+\Lambda B}}
	\arrow["{\text{quant.}}", squiggly, from=1-1, to=1-3]
	\arrow["{\text{Poisson ideal}}",hook', from=2-1, to=1-1]
	\arrow["{\text{quant.}}", squiggly, from=2-1, to=2-3]
	\arrow["{\text{ideal}}"',hook', from=2-3, to=1-3]
\end{tikzcd}\]
Note that since $\cfsppolarization_0$ is a vector space and the Poisson structure is affine, we can equivalently describe it by the dual
$\Ghat^*_\Lambda$ of a Lie algebra $\mathcal{G}_\Lambda$ with central extension $\Ghat_\Lambda$ defined by the 2-cocycle
$$c_\Lambda(\alpha\oplus f  , \beta \oplus g) = \int_{\corner} (\alpha \dd_{A_0} g  -  \beta \dd_{A_0} f + \Lambda  \alpha \beta)\Ccharge .$$ From the theory of deformation quantization, one expects the UEA $\mathcal{U}(\Ghat_\Lambda)$ to provide a quantization of $C^\infty(\cfsppolarization_0)$ (see Footnote~\ref{fn:quanti}). In Section \ref{sec:corner algebra of BF}, the Lie algebra $\Ghat_\Lambda$ is presented explicitly, as it was worked out in {\cite[Section 7.2.5]{CanCatCorner23}}. The Poisson ideal of constraints is expected to induce an ideal in the UEA which is denoted by $\mathcal{I}_{F_A+\Lambda B}$.\footnote{In general, the ideal is not necessarily two-sided, but in our case, this turns out to be true.}
This procedure will be examined in more depth in Section \ref{sec:constr}.

\subsubsection{The Corner Algebra of \texorpdfstring{$BF$}{bf} Theory}\label{sec:corner algebra of BF}
The Lie algebra describing the
unconstrained quantized corner structure
of 4-dimensional $BF$ theory on $\corner$ is given by the vector space
\begin{equation}\label{eq:bf_full_general_space}
    \Ghat_\Lambda = \Omega^0(\corner)\otimes\mathfrak{g} \oplus \Omega^1(\corner)\otimes\mathfrak{g} \oplus \mathbb{R}
\end{equation}
with bracket
\begin{equation}\label{eq:bf_full_general_bracket}
    [f\oplus\alpha  \oplus r,  g\oplus\beta \oplus s]_{\Ghat_\Lambda} = [f,g]\oplus (\mathrm{ad}_f\beta - \mathrm{ad}_g \alpha)\oplus \frac{-1 }{(2\pi)^2} \int_{\corner} (\alpha \dd_{A_0} g  -  \beta \dd_{A_0} f + \Lambda  \alpha \beta)\Ccharge \, ,
\end{equation}
where $\mathrm{ad}_f \alpha \coloneqq [f,\alpha]$ stands for the 1-form produced by the pointwise commutator. We will refer to this algebra as the \textbf{free corner algebra}.

However, one still needs to take the quotient with respect to the constraints to obtain the physically relevant algebra. In the following section, we show that the Poisson ideal $I_{F_A +\Lambda B}$  in $C^\infty(\cfsppolarization_0)$ induces a two-sided ideal $\mathcal{I}_{F_A +\Lambda B}$ in the UEA $\mathcal{U}(\Ghat_\Lambda)$. The physically relevant algebra describing the quantization of the physical space of corner fields is therefore expected to be the following heuristic quotient quantizing $H^0_{\pi_1}(\mathcal{B})$, as in Equation \eqref{e:cohomology_P_infty}:\footnote{Since we take the UEA without completion (cf.\ Footnote \ref{fn:quanti}), the quantized ideal of constraints is not a subspace of the UEA, so the quotient is heuristic. In principle, one would have to quantize the entire space of degree-0 fields $C^\infty(\cfsppolarization_0)$ for the quotient to be well defined.}
\begin{definition}
    The \textbf{physical corner algebra} of 4-dimensional $BF$ theory associated with a closed oriented surface $\corner$ is given by 
    \begin{equation}\label{eq:quotient_corner_algebra}
    A_\corner\equiv\mathcal{U}(\Ghat_\Lambda)\vert_{F_A+\Lambda B= 0} \coloneqq \faktor{\mathcal{U}(\Ghat_\Lambda) }{ \mathcal{I}_{F_A +\Lambda B}}\,
\end{equation}
and similarly for $\Lambda = 0$.
\end{definition}
The physical corner algebra should act on the state spaces (cf.\ Equation \ref{def:physical_state_space_assoc_to_boundary}) of compact 3-manifolds that have $\corner$ as their boundary. Therefore, we are interested in finding modules of $A_\corner$. Physically admissible representations should be those modules $M$ of the free corner algebra $\mathcal{U}(\Ghat_\Lambda)$ that descend to the quotient \eqref{eq:quotient_corner_algebra}. In other words, the ideal of constraints must act trivially under the module action. If the constraints do not act trivially on $M$, they might generate a proper submodule $\mathcal{I}_{F_A +\Lambda B} M\subsetneq M$. By definition, the (non-zero) quotient module ${M}/{\mathcal{I}_{F_A + \Lambda B} M}$ will be an $A_\corner$-module.

\subsubsection{Quantization of the Constraints}\label{sec:constr}
Let us first prove that the set $I_{F_A+\Lambda B}$ generated 
by $\int_\corner f(F_A+\Lambda B)$ actually forms a Poisson ideal in $C^\infty(\cfsppolarization_0)$. The Poisson brackets of the linear functionals in $\cfsppolarization_0$ with the constraint functionals are the following:
\begin{alignat*}{2}
    &\left\{ \int_\corner f(F_A+\Lambda B) , \int_\corner \alpha a \right\}_2 &&=0\,,\\
    &\left\{\int_\corner f(F_A+\Lambda B) , \int_\corner gB \right\}_2 &&= \int_\corner [f,g](F_A +\Lambda B) \, ,
\end{alignat*}
which shows the validity of our claim.

To understand the corresponding quantization, we make the following observation: 
$\int_\corner f(F_A+\Lambda B)\in (\Ghat_\Lambda^*)^* \oplus \left(\Ghat_\Lambda^* \otimes\Ghat_\Lambda^* \right)^*$ is a symmetric multilinear functional (a polynomial) in the fields. Assuming some appropriate completion of the tensor product, the functional defines an element of $ \Ghat_\Lambda \oplus \mathrm{Sym}^2(\Ghat_\Lambda)$.
By the linear symmetrization isomorphism
     $\mathrm{Sym}^\bullet(\Ghat_\Lambda) \cong \mathcal{U}(\Ghat_\Lambda),$
we can expect this functional to describe an element, denoted $\widehat{f}$, of (the completion of) the UEA of $\Ghat_\Lambda$.
Hence, we (heuristically) define the \textbf{quantized ideal of constraints} as
the two-sided ideal generated by $\widehat{f}$ in the corresponding UEA, i.e.\ $\mathcal{I}_{ F_A+\Lambda B} \coloneqq \langle \widehat{f}\rangle$.

In the abelian case, i.e.\ $\mathfrak{g}=\mathbb{R}$, the ideals of constraints $\mathcal{I}_{\dd A}$ \& $\mathcal{I}_{\dd A+\Lambda B}$ are subsets of the centers of the Lie algebras $\Ghat$ \& $\Ghat_\Lambda$ respectively. This special case makes the selection of physically admissible representations much easier and mathematically rigorous.

\section{Abelian \texorpdfstring{$BF$}{BF} Theory on the Torus}\label{sec:abelianontor}
In this section, we classify the free and physical corner algebras of four-dimensional abelian $BF$ theory on a torus and provide possible modules. First, to make the Lie algebra more manageable, we reduce to the Lie subalgebra generated by finite Fourier expansions and work out the bracket relations on the generators of that subalgebra. The Lie algebras for a zero and non-zero cosmological constant turn out to be isomorphic to a known infinite-dimensional Lie algebra. Afterwards, we impose the constraints to determine the physical corner algebras. Finally, we provide an example of a module that realizes a possible action on the state space for a compact, oriented $3$-manifold $\boundary$ such that $\partial \boundary \cong \corner$, for example, a solid torus.

\subsection{The Lie Algebras \texorpdfstring{$\Ghat $ \& $ \Ghat_\Lambda$}{ghat}}
The Lie algebra describing the
unconstrained quantized corner structure of 4-dimensional abelian $BF$ theory on a torus $\torus$ is given by the vector space \begin{equation*}
    \Ghat_\Lambda = \Omega^0(\torus) \oplus \Omega^1(\torus) \oplus \mathbb{R}
\end{equation*}
with Lie bracket \begin{equation}\label{eq:abelianfullbrac}
    [ f\oplus\alpha \oplus r,g \oplus \beta\oplus s]_{\Ghat_\Lambda} = - \frac{1}{(2\pi)^2} \int_{\torus} (\alpha\dd g - \beta \dd f + \Lambda \alpha \beta ) \Ccharge\, ,
\end{equation}
where $\Lambda \in \mathbb{R}$ is the cosmological constant and $\Ccharge$ is the central charge. This algebra is obtained from~\eqref{eq:bf_full_general_space} and~\eqref{eq:bf_full_general_bracket} by setting $\corner = T^2$ and $\mathfrak{g}=\mathbb{R}$. For $\Lambda=0$, we denote the Lie algebra by $\Ghat$. 

Following what we have done for Chern--Simons theory in Section~\ref{sec:CSbf2v}, we will again consider the Lie subalgebra of the complexification $(\Ghat_\Lambda)_\mathbb{C}$ defined by finite Fourier modes. These elements are of the form
\begin{alignat*}{3}
    &\alpha &&= \sum_{m,n\in \mathbb{Z}} \alpha_{mn}^{(\theta)}e^{im\theta}e^{in\varphi}\dd\theta + \alpha_{mn}^{(\varphi)}e^{im\theta}e^{in\varphi}\dd\varphi &&\quad \text{for }\alpha_{mn}^{(\theta)},\alpha_{mn}^{(\varphi)}\in\mathbb{C}\,, \\
    &f &&= \sum_{m,n\in \mathbb{Z}} f_{ mn}  e^{im\theta}e^{in\varphi} &&\quad \text{for }f_{mn}\in\mathbb{C}\,,
\end{alignat*}
with only finitely many non-zero coefficients and where ($\theta,\varphi$) are coordinates on the torus with orientation given by the volume form $\dd\theta\wedge\dd\varphi$. The space of differential 1-forms $\Omega^1(\torus)$ is a $C^\infty(\torus)$-module. It has a basis given by the global 1-forms $\{\dd\theta,\dd \varphi\}$.\footnote{The coordinate functions $\theta$ and $\varphi$ are not defined globally, but the induced 1-forms $\dd \theta$ and $\dd \varphi$ are.}
With a slight abuse of notation, we denote the Lie subalgebra by the same symbol and drop the complexification symbol. Now that the vector space $\Ghat_\Lambda$ has countable dimension, we can choose a suitable basis $$\{\Emode_{ mn} \coloneqq e^{im\theta}e^{in\varphi}, \Phmode_{ mn}\coloneqq e^{im\theta}e^{in\varphi}\dd\varphi, \Thmode_{mn}\coloneqq  e^{im\theta}e^{in\varphi}\dd\theta \mid m,n\in\mathbb{Z}\}\cup\{\Ccharge\}.$$

From now on, the summation convention will be adopted for Fourier mode indices unless there is a chance of confusion. Furthermore, we drop the subscript of the Lie bracket in~\eqref{eq:abelianfullbrac}. A quick computation yields the following lemma.
\begin{lemma}
The Lie brackets on generators of the finite Fourier mode algebra are given by:
    \begin{alignat}{2} 
            &[\Emode_{ kl},\Phmode_{ mn}] & &= i m\delta_{k,-m}\delta_{l,-n}\Ccharge\,,\label{eq:abEandPhcoupling}\\
            &[\Emode_{ kl},\Thmode_{ mn}] & &= -i n\delta_{k,-m}\delta_{l,-n}\Ccharge\,, \label{eq:abEandThcoupling}\\
            &[\Phmode_{ kl},\Thmode_{ mn}] & &= \Lambda \delta_{k,-m}\delta_{l,-n}\Ccharge\,\label{eq:abPhandThcoupling}
        \end{alignat}
    and all other brackets vanish.
\end{lemma}
\begin{proof}
The proof of the bracket relations is straightforward. Evaluating the bracket on the basis elements and remembering the chosen orientation on $\torus$ yields:
    \begin{align*}
        [\Emode_{ kl}, \Phmode_{ mn}] 
        &=\frac{1}{(2\pi)^2} \int_{\torus} \Phmode_{ mn}\dd\Emode_{ kl}  \Ccharge\\
        &= -\frac{1}{(2\pi)^2}\left(\int_{\torus} ike^{i(m+k)\theta}e^{i(n+l)\varphi}\dd\theta\wedge \dd\varphi \right) \Ccharge= i m \delta_{k,-m}\delta_{l,-n}\Ccharge,\\
        [\Emode_{ kl}, \Thmode_{ mn}]  &=\frac{ 1}{(2\pi)^2}  \int_{\torus} \Thmode_{ mn}\dd\Emode_{ kl} \Ccharge\\
         &= \frac{1}{(2\pi)^2}\left(\int_{\torus} ile^{i(m+k)\theta}e^{i(n+l)\varphi}\dd\theta\wedge \dd\varphi\right)\Ccharge = -i n\delta_{k,-m}\delta_{l,-n}\Ccharge,\\
        [\Phmode_{kl}, \Thmode_{mn}]  &=  -\frac{  1}{(2\pi)^2} \int_{\torus}  \Lambda\Phmode_{ kl} \Thmode_{ mn}\Ccharge\\
        &=  \frac{ 1}{(2\pi)^2}\left( \int_{\torus} \Lambda e^{i(k+m)\theta}e^{i(l+n)\varphi}\dd\theta\wedge \dd\varphi \right)\Ccharge =  \Lambda \delta_{k,-m}\delta_{l,-n}\Ccharge
         \end{align*}
\end{proof}
Next, we classify these Lie algebras for zero and non-zero cosmological constant. 

\subsection{Classification of \texorpdfstring{$\Ghat$ \& $\Ghat_\Lambda$}{Ghat}}
Before we state and prove a classification theorem for the free corner algebras $\Ghat$ \& $\Ghat_\Lambda$, we introduce the oscillator algebra and the countably infinite-dimensional abelian Lie algebra.
\begin{definition}
    The \textbf{oscillator (or Heisenberg) algebra $\mathcal{A}$} is the complex Lie algebra over the vector space $\langle \Ccharge,a_m,a_m^\dag\mid m\in\mathbb{Z}\rangle$ with the commutation relations
    \begin{align*}
        [a_m,a_n^\dag]&= \delta_{m,n}\Ccharge \quad \forall m,n\in\mathbb{Z}
    \end{align*}
    and zero otherwise.
\end{definition}
\begin{remark}
    In the usual definition of the oscillator algebra, one includes an additional central element called the zero mode (cf.\ \cite[Section 2.2]{KacRaiBombay87}). Furthermore, one can also define a finite-dimensional oscillator algebra. For $n\in\mathbb{N}$, $\mathcal{A}(n)\coloneqq\langle \Ccharge,a_m,a_m^\dag\mid 1\leq m\leq n\rangle$, with the same commutation relations as above, is the oscillator algebra of dimension $2n+1$.
\end{remark}
\begin{definition}
    The \textbf{countably infinite-dimensional abelian Lie algebra $\mathfrak{a}$} is the countably infinite-dimensional vector space equipped with the trivial Lie bracket.
\end{definition}
\begin{theorem}\label{thm:abisheis} 
There is an isomorphism of Lie algebras:
\begin{equation}
    \Ghat \cong \Ghat_\Lambda \cong \mathcal{A}\oplus\mathfrak{a}\,.\nonumber
\end{equation}
\end{theorem}
\begin{proof}
We distinguish between two cases.
\textbf{Case $\Lambda=0$:} Observe that both the $\Thmode$-generators and $\Phmode$-generators couple in the same way (up to scaling) to the $\Emode$-generators. We can therefore define two new families of generators by normalizing appropriately and taking the sum and the difference of the original generators. The new family arising from the sum behaves the same as the original up to a scaling factor and is adorned with a check. The second family arising from the difference now has trivial bracket with all other generators and is therefore part of the center. These generators will be adorned with a hat.
This observation is made precise with the following definitions:
    \begin{align*}
        \Fmode_{kl}^\pm &\coloneqq \frac{1}{2} \left(-\frac{1}{l}\Thmode_{kl}\pm \frac{1}{k}\Phmode_{kl} \right) 
        &  \widecheck{\Phmode}_{k} &\coloneqq \frac{1}{k}\Phmode_{k0} 
        & \widecheck{\Thmode}_{l} &\coloneqq -\frac{1}{l}\Thmode_{0l} &  k, l &\neq0 \\
        \widehat{\Phmode}_{l} &\coloneqq \Phmode_{0l} 
        & \widehat{\Thmode}_{k} &\coloneqq \Thmode_{k0}
        &\widehat{\Emode} &\coloneqq \Emode_{00} & k, l &\in\mathbb{Z}
    \end{align*}
Computing the bracket relations of the new generators using Equations~(\ref{eq:abEandPhcoupling})-(\ref{eq:abPhandThcoupling}), one obtains:
\begin{align*}
    &[\Emode_{kl},\Fmode_{mn}^+] =i \delta_{k,-m}\delta_{l,-n}\Ccharge\,, \\
    &[\Emode_{0l},\widecheck{\Thmode}_{n}] =i \delta_{l,-n}\Ccharge\,, \\
    &[\Emode_{k0},\widecheck{\Phmode}_{m}] =i \delta_{k,-m}\Ccharge\,,
\end{align*}
and zero otherwise. To make the connection to the oscillator algebra manifest, we can denote the non-central generators:
\begin{alignat*}{3}
        &\amode_l ^\dag &&\coloneqq \widecheck{\Thmode}_{l}\quad && l\neq 0\,,\\ 
        &\bmode_k^\dag &&\coloneqq \widecheck{\Phmode}_{k} \quad && k\neq 0\,, \\
        &\cmode_{kl}^\dag &&\coloneqq \Fmode_{kl}^+ \quad && k\neq 0,l\neq 0\,, \\
        &\amode_l &&\coloneqq \Emode_{0-l}\quad && l\neq 0\,, \\ 
        &\bmode_k &&\coloneqq \Emode_{-k0} \quad && k\neq 0\,, \\
        &\cmode_{kl} &&\coloneqq \Emode_{-k-l} \quad && k\neq 0,l\neq 0\,.
    \end{alignat*}
The final relations are the following:
\begin{equation*}
    [\amode_l,\amode^\dag_l] =[\bmode_k,\bmode^\dag_k] =[\cmode_{kl},\cmode^\dag_{kl}] = i\Ccharge \,,
\end{equation*}
and zero otherwise.\footnote{Note that we could absorb the $i$-factor by a redefinition of generators, however, for later convenience, we leave it as is.} The abelian summand is spanned by the central elements excluding $\Ccharge$, i.e.\ {$\mathfrak{a} \coloneqq \langle\Fmode^-_{kl},\widehat{\Phmode}_n,\widehat{\Thmode}_m,\widehat{\Emode}\mid  k,l,m,n\in\mathbb{Z}\text{ and }k,l\neq 0\rangle$}. Hence, the resulting Lie algebra is isomorphic to the direct sum $\mathcal{A}\oplus\mathfrak{a}$.

\textbf{Case $\Lambda\neq 0$:} Now the $\Thmode$-generators and $\Phmode$-generators couple non-trivially in (\ref{eq:abPhandThcoupling}). However, we can do a similar transformation to separate the generators into central- and oscillator-type elements. In terms of the generators for the $\Lambda=0$ case, one can come up with:
\begin{alignat*}{3}
    &\umode^\dag_l&&\coloneqq \amode^\dag_l &&\quad l\neq0 \,,\\
    &\vmode^\dag_k&&\coloneqq \bmode^\dag_k &&\quad k\neq0 \,,\\
    &\wmode^\dag_{kl}&&\coloneqq \cmode^\dag_{kl} &&\quad k,l\neq0 \,,\\ 
    &\umode_l&&\coloneqq \frac{1}{2}\left(\frac{l}{i\Lambda}\widehat{\Phmode}_{-l}+\amode_l\right) &&\quad l\neq0 \,,\\
    &\vmode_k&&\coloneqq \frac{1}{2}\left(\frac{k}{i\Lambda}\widehat{\Thmode}_{-k}+\bmode_k\right) &&\quad k\neq0 \,,\\
    &\wmode_{kl}&&\coloneqq \frac{1}{2}\left(\frac{2kl}{i\Lambda}\Fmode_{-k-l}^- +\cmode_{kl}\right)&&\quad k,l\neq0 \,,\\ 
    &\widehat{\umode}_l&&\coloneqq \frac{1}{2}\left(\frac{l}{i\Lambda}\widehat{\Phmode}_{-l}-\amode_l\right)  &&\quad l\neq0\,,\\
    &\widehat{\vmode}_k&&\coloneqq \frac{1}{2}\left(\frac{k}{i\Lambda}\widehat{\Thmode}_{-k}-\bmode_k\right) &&\quad k\neq0\,,\\
    &\widehat{\wmode}_{kl}&&\coloneqq \frac{1}{2}\left(\frac{2kl}{i\Lambda}\Fmode_{-k-l}^- -\cmode_{kl}\right)&&\quad k,l\neq0 \,,\\
    &\bar{\Phmode} &&\coloneqq \frac{1}{i\Lambda}\widehat{\Phmode}_0\,, \\
    &\bar{\Thmode} &&\coloneqq  \widehat{\Thmode}_0\,,
\end{alignat*}
The final relations are the following:
\begin{equation*}
    [\umode_l,\umode^\dag_l] =[\vmode_k,\vmode^\dag_k] =[\wmode_{kl},\wmode^\dag_{kl}] = [\bar{\Thmode},\bar{\Phmode}]= i\Ccharge \,,
\end{equation*}
and zero otherwise. The abelian summand is spanned by the central elements $\widehat{\wmode}_{kl},\widehat{\umode}_l,\widehat{\vmode}_k$ and $\widehat{\Emode}$, the same as before. Again, the resulting Lie algebra is isomorphic to the direct sum $\mathcal{A}\oplus\mathfrak{a}$ and, as such, is also isomorphic to the Lie algebra for $\Lambda=0$.
\end{proof}
By Dixmier's Lemma, any central element, say  $\widehat{\Emode}$, in a countably infinite-dimensional Lie algebra acts by a multiple of the identity on an irreducible representation. In this case, we denote the proportionality factor by $\charge{\widehat{\Emode}}\in\mathbb{C}$ and call it the charge of $\widehat{\Emode}$.

Next, we want to understand the appropriate reduction necessary to impose the constraints.

\subsection{Constraints \texorpdfstring{$\dd A = 0$ \& $\dd A + \Lambda B= 0$}{dA + Lambda B= 0}}
The constraints simplify a lot in the abelian theory. In particular, the constraints~\eqref{eq:abelian_constr} and~\eqref{eq:abelian_constrLambda} are linear and central in the Poisson algebras and thus define central elements in the Lie algebras. The corresponding ideals, $\mathcal{I}_{\dd A}\subset Z(\Ghat)$ and $\mathcal{I}_{\dd A +\Lambda B}\subset Z(\Ghat_\Lambda)$, quantize the constraints (recall the heuristics in Section~\ref{sec:constr} and see the following Proof~\ref{prf:abelian_corneralg_iso}). 
One can form the quotient Lie algebras $\Ghat\vert_{\dd A = 0}\coloneqq \faktor{\Ghat}{\mathcal{I}_{\dd A}}$ and ${\Ghat_\Lambda}\vert_{\dd A +\Lambda B= 0}\coloneqq \faktor{\Ghat_\Lambda}{\mathcal{I}_{\dd A+\Lambda B}}$ in a straightforward way. They are characterized by the
following proposition:
\begin{proposition} The physical corner algebra of abelian $BF$ theory in four dimensions  on the torus,  with and without cosmological term, is characterized by the isomorphism of Lie algebras
\begin{equation*}
    \Ghat\vert_{\dd A = 0} \cong \mathcal{A}\oplus\mathbb{C}^3, \quad{\Ghat_\Lambda}\vert_{\dd A +\Lambda B= 0}\cong \mathcal{A}\, .
\end{equation*}
\end{proposition}

\begin{proof}\label{prf:abelian_corneralg_iso}
\textbf{Case $\Lambda = 0$:}
The constraint ideal in the classical Poisson algebra is generated by the functionals:
\begin{equation}\label{eq:abelian_constr}
    \int_\torus f \dd A =-\int_\torus (\dd f)A \,,
\end{equation}
for any smooth function $f\in C^\infty(M)$. But these functionals correspond precisely to the elements $\widehat{f} \coloneqq \dd f\in\Ghat$. Furthermore, any closed 1-form $\alpha$, and in particular exact ones, has trivial bracket with the other functionals. This follows directly from the definition of the bracket in (\ref{eq:abelianfullbrac}), which only depends on $\dd \alpha$ (via Stokes' Theorem). Therefore, the set of constraints lies in the center $\mathcal{I}_{\dd A}\coloneqq \langle\dd f \mid f\in C^\infty(M)\rangle\subset Z(\Ghat)$. We can examine the constraints explicitly using the Lie subalgebra of Fourier modes that we have examined in the previous section.

In the usual basis, the constraints imply the following for all $k,l\in\mathbb{Z}\colon$
\begin{align*}
    0  &\overset{!}{=}\dd \Emode_{kl}= il\Phmode_{kl}+ik\Thmode_{kl}= \begin{cases}
        -2ikl \Fmode_{kl}^-  & \quad\text{for } k \neq 0, l \neq 0 \\
        il\widehat{\Phmode}_l  & \quad\text{for } k = 0, l \neq 0 \\
        ik\widehat{\Thmode}_k   & \quad\text{for } k \neq 0, l = 0 \\
        0   & \quad\text{for } k = l = 0 \, .
  \end{cases}
\end{align*}
Taking the quotient with respect to the ideal spanned by these elements, one is left with:
\begin{equation*}
    \Ghat\vert_{\dd A=0} \cong \mathcal{A}\oplus\mathbb{C}^3,
\end{equation*}
which describes the quantization of the degree zero cohomology of corner fields in abelian $BF$ for $\Lambda = 0$.
In other words, we are left with the Lie algebra spanned by the generators $ \amode_l^\dag,\amode_l,$$ \bmode_k^\dag, \bmode_k, \cmode_{kl}^\dag, \cmode_{kl},$$ \Ccharge,\widehat{\Thmode}_0,\widehat{\Phmode}_0,\widehat{\Emode}$, where the latter three generate the abelian summand.
\begin{remark}
    It makes sense that the center (ignoring the extension) is 3-dimensional. Before the quotient, the center consisted of constant functions and closed 1-forms. By taking the quotient, we essentially obtain the corresponding cohomology groups of the torus, i.e.\ $\mathbb{R}^3\cong H^0(\torus,\mathbb{R})\oplus H^1(\torus,\mathbb{R})$ ignoring the complexification.
\end{remark}
To summarize, restricting to the submanifold $\dd A=0$ manifests itself in the quantization as selecting the representations of $\Ghat$ where the charges $\charge{\Fmode_{kl}^-}=\charge{\widehat{\Phmode}_l}=\charge{\widehat{\Thmode}_k} = 0$ for $k \neq 0, l \neq 0$ are set to zero. 

\textbf{Case $\Lambda\neq0$:}
Similarly to the $\Lambda = 0$ case, the constraints amount to setting most charges in $\mathfrak{a}$ to zero.
The constraint ideal is generated by the functionals:
\begin{equation}\label{eq:abelian_constrLambda}
    \int_\torus f(\dd A + \Lambda B)=-\int_\torus (\dd f)A +\int_\torus\Lambda fB 
\end{equation}
And therefore, we just need to set the combination $\widehat{f} \coloneqq \dd f -\Lambda f\in\Ghat_\Lambda$ to zero. The ideal of constraints $\mathcal{I}_{\dd A +\Lambda B}\coloneqq \langle\dd f -\Lambda f\in\Ghat_\Lambda\mid f\in C^\infty(M)\rangle$ again lies in the center for a similar reason.
In the usual basis, the constraints imply the following for all $k,l\in\mathbb{Z}\colon$
\begin{align*}
    0 \overset{!}{=}\dd \Emode_{kl} - \Lambda\Emode_{kl} 
    = ik\Thmode_{kl}+il\Phmode_{kl}-\Lambda \Emode_{kl}
 &= \begin{cases}
        -2ikl \Fmode_{kl}^- - \Lambda\cmode_{-k-l}  & \text{for } k \neq 0, l \neq 0 \,\\
        il\widehat{\Phmode}_l- \Lambda\amode_{-l} & \text{for } k = 0, l \neq 0 \,\\
        ik\widehat{\Thmode}_k - \Lambda\bmode_{-k} & \text{for } k \neq 0, l = 0\,\\
        -\Lambda \widehat{\Emode}  & \text{for } k = l = 0\, \\
  \end{cases}\\
  &= \begin{cases}
        2\Lambda \widehat{\wmode}_{-k-l} & \text{for } k \neq 0, l \neq 0\,\\
        2\Lambda \widehat{\umode}_{-l}   & \text{for } k = 0, l \neq 0\, \\
        2\Lambda \widehat{\vmode}_{-k}  & \text{for } k \neq 0, l = 0 \\
        -\Lambda \widehat{\Emode}   & \text{for } k = l = 0 \,.
    \end{cases}
\end{align*}
Taking the quotient with respect to the ideal generated by these elements, one is left with:
\begin{equation*}
    \Ghat_\Lambda\vert_{\dd A + \Lambda B=0} \cong \mathcal{A}\,,
\end{equation*}
which describes the quantization of the degree zero cohomology of corner fields in abelian $BF$ for $\Lambda \neq 0$. In other words, we are left with the oscillator Lie algebra spanned by the generators
   $ \umode_l^\dag, \umode_l, \vmode_k^\dag, \vmode_k, \wmode_{kl}^\dag, \wmode_{kl},  \bar{\Thmode},\bar{\Phmode},\Ccharge\,.$
   
To summarize, restricting to the submanifold $\dd A+\Lambda B=0$ manifests itself in the quantization as selecting the representations of $\Ghat_\Lambda$ where the charges $\charge{\widehat{\wmode}_{kl}}=\charge{\widehat{\umode}_l}=\charge{\widehat{\vmode}_k} = 0$ for $k \neq 0, l \neq 0$ are set to zero. 
\end{proof}
\begin{remark}\label{rem:relation_to_current_algebra_1}
    As mentioned in the Introduction, the current algebra obtained in \cite{fliss_entanglement_2025} agrees with the physical corner algebra defined herein for $\Lambda=0$.
\end{remark}
\subsection{Modules of \texorpdfstring{$\Ghat\vert_{\dd A }$ \& $\Ghat_\Lambda\vert_{\dd A + \Lambda B=0}$}{g}} \label{sec:modules_of_abelian_corner_algebra}
From the previous section, we know that the physical corner algebras of abelian $BF$ theory in four dimensions associated with a torus are:
\begin{equation*}
    A_\torus \cong \mathcal{U}(\mathcal{A}\oplus \mathbb{C}^3),\quad  A_\torus \cong \mathcal{U}(\mathcal{A}), 
\end{equation*}
for $\Lambda = 0$ and $\Lambda \neq 0$ respectively.

To obtain modules of $A_\torus$, one can simply choose a representation of $\mathcal{A}$; for example, the bosonic Fock space representation as defined in \cite[Section 2.2]{KacRaiBombay87}. More concretely, let $V=\mathbb{C}[\{v_{m}\}_{m\in\mathbb{Z}}]$ be the space of polynomials in countably infinitely many variables and define the action of the generators of $\mathcal{A}$ on $V$ by the assignment
\begin{equation*}
\begin{aligned}
    \Ccharge &\longmapsto \charge{\Ccharge},\\
    \amode^\dag_{m} &\longmapsto v_{m} &\quad \text{ for } m \in \mathbb{Z}, \\
    \amode_{m} &\longmapsto \charge{\Ccharge}\pdv{}{v_{m}} &\quad \text{ for } m \in \mathbb{Z},
\end{aligned}
\end{equation*}
for $\charge{\Ccharge}\in\mathbb{C}$.
One can then extend the module action to $\mathcal{A}\oplus\mathfrak{a}$ by choosing complex numbers $\charge{a}\in\mathbb{C}$ for each generator $a$ in $\mathfrak{a}$. Finally, one obtains a family of representations of $\Ghat\vert_{\dd A }$ \& $\Ghat_\Lambda\vert_{\dd A + \Lambda B=0}$ and equivalently the physical corner algebras, parametrized by the choice of complex numbers $\charge{a}$.

These modules describe physically admissible state spaces, i.e.\ possible quantizations of the reduced phase spaces associated to compact 3-manifolds that bound a torus in 4-dimensional abelian $BF$ theory.\footnote{If one is interested in not just the finite Fourier mode algebra but the total algebra, one could take the induced module of this family to obtain non-trivial modules.} It would now be interesting to investigate the details of the gluing in Example~\ref{ex:torus} and how the state space on the solid torus and lens space decomposes in terms of the corner algebra modules.

\section{Abelian \texorpdfstring{$BF$}{BF} Theory on a Closed Oriented Riemannian Surface}\label{sec:abelian_BF_on_any_surface}
In this section, we generalize the results from Section \ref{sec:abelianontor} and classify the free and physical corner algebras of four-dimensional abelian $BF$ on an arbitrary closed oriented surface $\corner$ and provide possible modules. First, we identify a suitable countably infinite-dimensional Lie subalgebra defined using a Riemannian metric, eigenfunctions of the Laplacian and the Hodge decomposition theorem. This subalgebra is the analogue of the finite Fourier mode algebra in this setting. We then prove a classification statement for these Lie subalgebras and consider their modules.

\subsection{The Lie Algebras \texorpdfstring{$\Ghat $ \& $ \Ghat_\Lambda$}{ghat}}
The Lie algebra describing {the
unconstrained quantized corner structure of} 4-dimensional abelian $BF$ on any closed oriented surface $\corner$ is given by the vector space
\begin{equation*}
    \Ghat_\Lambda = \Omega^0(\corner) \oplus \Omega^1(\corner) \oplus \mathbb{R}\,,
\end{equation*}
with Lie bracket \begin{equation*}
    [ f\oplus\alpha \oplus r,g \oplus \beta\oplus s]_{\Ghat_\Lambda} = -  \int_{\corner} (\alpha\dd g - \beta \dd f + \Lambda \alpha \beta ) \Ccharge\, ,
\end{equation*}
where $\Lambda \in \mathbb{R}$ is the cosmological constant and $\Ccharge$ is the central charge. For $\Lambda=0$, we denote the Lie algebra by $\Ghat$.
We now construct the finite Fourier mode subalgebra analog in this generality and provide a suitable vector space basis.

Choose a Riemannian metric $g$ on $\corner$ and denote its associated volume form by $\mathrm{vol}_g$ and the Laplace--Beltrami operator by $\Lapg$. There exists a set of functions $\{\eigfun{lm}\in C^\infty(\corner) \mid l\in\mathbb{N},1\leq m\leq R(l) \text{ for }  R(l)\in\mathbb{N}\}$ satisfying
\begin{enumerate}
    \item $\Lapg\eigfun{lm}=\lambda_{l}\eigfun{lm}$, with $0=\lambda_0<\lambda_1<\cdots \rightarrow +\infty$
    \item $\dim(\mathrm{Eig}({\Lapg,\lambda_l}))= R(l)$
    \item $\{\eigfun{lm}\}_{lm}$ is an orthonormal basis of $L^2(\corner)$, i.e.\ $\int_\corner \eigfun{lm}\eigfun{l'm'}\mathrm{vol}_g=\delta_{l,l'}\delta_{m,m'}$
\end{enumerate}
by Theorem 18.\ in \cite{berard_spectral_1986}.  
In particular,  $R(0)=\dim(H^0(\corner,\mathbb{R})) $. We can regard $\{\eigfun{lm}\}_{lm}$ also as eigenfunctions of the Hodge Laplacian by absorbing the resulting minus sign into the eigenvalues. These signs will not matter, and we shall use $\Lapg$ from now on to denote the Hodge Laplacian.

By Hodge's Theorem on compact Riemannian manifolds without boundary, we can decompose any 1-form $\omega\in\Omega^1(\corner)$ as 
\begin{equation*}
    \omega = \dd f + \delta \eta + \gamode,
\end{equation*}
where $f \in C^\infty(\corner),\eta \in\Omega^2(\corner),\gamode\in\Omega^1(\corner)$ and $\gamode$ is harmonic, i.e.\ $\Lapg \gamode = (\dd \delta + \delta \dd ) \gamode = 0$. Using the Hodge star operator $\star\equiv\star_g$, $\delta = -\star \dd \star$ in 2 dimensions and $\star^2\vert_{\Omega^0(\corner)}=1$, we can equivalently write
\begin{equation*}
    \omega = \dd f - \star \dd h + \gamode,
\end{equation*}
for a smooth function $h\in C^\infty(\corner)$ such that $\eta =\star h$.

Now, define the finite Fourier mode $\mathbb{R}$-subalgebra:
\begin{align*}
    C^\infty_{\mathrm{FFM}}(\corner)\coloneqq \Bigg\{f = \sum_{l\in\mathbb{N}}\sum_{m=1}^{R(l)}f_{lm}\eigfun{lm}\Bigg \vert \text{ finitely many }f_{lm}\neq0 \Bigg\}
\end{align*}
and the finite Fourier mode vector subspace 
\begin{align*}
    \Omega^1_{\mathrm{FFM}}(\corner) \coloneqq \Bigg\{ \omega& = \sum_{l\in\mathbb{N}\setminus\{0\}}\sum_{m=1}^{R(l)}\omega^{(\Gmode)}_{lm}\dd \eigfun{lm} + \omega^{(\Hmode)}_{lm}\star\dd \eigfun{lm} \\&+ \sum_{r=1}^{\dim(H^1(\corner,\mathbb{R}))}\omega^{(\gamode)}_r\gamode_r  \Bigg \vert \text{ finitely many }\omega^{(\Gmode)}_{lm},\omega^{(\Hmode)}_{lm}\neq0\Bigg\},
\end{align*}
where $\{\gamode_r \mid 1\leq r\leq\dim(H^1(\corner,\mathbb{R}))\}$ is a basis of $H^1(\corner,\mathbb{R})$ in terms of harmonic 1-forms.
Finally, the Lie subalgebra of interest is just $\Ghat_{\Lambda,\mathrm{FFM}}\coloneqq C^\infty_{\mathrm{FFM}}(\corner)\oplus\Omega^1_{\mathrm{FFM}}(\corner)\oplus\mathbb{R}$ with the induced Lie bracket. We refer to this Lie subalgebra as the \textbf{finite Fourier mode algebra} and drop the FFM from the notation. Now that the vector space $\Ghat_\Lambda$ has countable dimension, we can choose a suitable basis $\{\Amode_{lm} \coloneqq \eigfun{lm}\mid l\in\mathbb{N},1\leq m\leq R(l)\} \cup \{\Gmode_{lm}\coloneqq \dd\eigfun{lm}, \Hmode_{lm}\coloneqq  \star\dd\eigfun{lm}\mid l\in\mathbb{N}\setminus \{0\},1\leq m\leq R(l) \}\cup\{\gamode_r \mid 1\leq r\leq\dim(H^1(\corner,\mathbb{R}))\}\cup \{\Ccharge\}$.

\begin{lemma}
    The Lie brackets on generators of the finite Fourier mode algebra are given by:
    \begin{align*} 
            [\Amode_{lm},\Hmode_{l'm'}] & = -\lambda_{l'}\delta_{l,l'}\delta_{m,m'}\Ccharge\,,\\
            [\Gmode_{ lm},\Hmode_{ l'm'}] & = -\Lambda\lambda_{l'}\delta_{l,l'}\delta_{m,m'}\Ccharge\,, \\
            [\gamode_r,\gamode_{r'}] & = \Lambda\delta_{r,r'}\Ccharge\,,
        \end{align*}
    and all other brackets vanish.
\end{lemma}

\begin{proof}
The proof of the Lie bracket relations is a straightforward computation. The non-trivial brackets are:
\begin{align*}
    [\Amode_{lm},\Gmode_{l'm'}] & = \int_\corner\Gmode_{l'm'}\dd\Amode_{lm}\Ccharge=\int_\corner \dd\eigfun{l'm'}\wedge\dd\eigfun{lm}\Ccharge=\int_\corner \dd(\eigfun{l'm'}\wedge\dd\eigfun{lm}\Ccharge)=0,\\
    [\Amode_{lm},\Hmode_{l'm'}] & = \int_\corner\Hmode_{l'm'}\dd\Amode_{lm}\Ccharge=\int_\corner \star\dd\eigfun{l'm'}\wedge\dd\eigfun{lm}\Ccharge=-\int_\corner \delta(\eigfun{l'm'}\mathrm{vol}_g)\wedge\dd\eigfun{lm}\Ccharge\\
    &=-\int_\corner \Lapg(\eigfun{l'm'}\mathrm{vol}_g)\wedge\eigfun{lm}\Ccharge=-\lambda_{l'}\int_\corner \eigfun{l'm'}\eigfun{lm}\mathrm{vol}_g\Ccharge=-\lambda_{l'}\delta_{l,l'}\delta_{m,m'}\Ccharge,\\
    [\Gmode_{lm},\Hmode_{l'm'}] & = -\Lambda\int_\corner\Gmode_{l'm'}\Hmode_{lm}\Ccharge=-\Lambda\int_\corner\dd\eigfun{l'm'}\wedge\star\dd\eigfun{lm}\Ccharge=-\Lambda\lambda_{l'}\delta_{l,l'}\delta_{m,m'}\Ccharge,\\
    [\Gmode_{lm},\gamode_r] & = -\Lambda\int_\corner\Gmode_{lm}\gamode_r\Ccharge=-\Lambda\int_\corner\dd\eigfun{lm}\wedge\gamode_r\Ccharge=-\Lambda\int_\corner\dd(\eigfun{lm}\wedge\gamode_r)\Ccharge=0,\\
    [\Hmode_{lm},\gamode_r] &= 0,\\
    [\gamode_r,\gamode_{r'}] &= -\Lambda\int_\corner \gamode_r\wedge\gamode_{r'}\Ccharge\eqqcolon -\Lambda b_{rr'}\,,
\end{align*}
where we define a skew-symmetric $\dim(H^1(\corner,\mathbb{R}))\times \dim(H^1(\corner,\mathbb{R}))$ matrix $b$ and all the other brackets vanish.
\end{proof}

\begin{remark}\label{rem:hamronic_subalg}
The generators $\gamode_r$ of harmonic 1-forms and the central charge $\zz$ span finite-dimensional Lie subalgebras which we will denote by $\mathfrak{h}$ \& $\mathfrak{h}_\Lambda$, for zero and non-zero cosmological constant respectively. Note that as vector spaces, $\mathfrak{h}$ \& $\mathfrak{h}_\Lambda$ are isomorphic to $H^1(\corner,\mathbb{R})\oplus\mathbb{R}$.
\end{remark}

\begin{lemma}\label{lem:hamronic_subalg_classific_lemma}
We have the following isomorphism of Lie algebras:
\begin{equation*}
    \mathfrak{h}\cong \mathbb{R}^{2n+1},\quad \mathfrak{h}_\Lambda\cong\mathcal{A}(n),\quad \text{for } n= \frac{1}{2}\dim(H^1(\corner,\mathbb{R})).
\end{equation*}
\end{lemma}
\begin{proof}
\textbf{Case $\Lambda=0$:} In pure $BF$ theory, the generators $\gamode_r$ are central for all $1\leq r \leq \dim(H^1(\corner,\mathbb{R}))$ because the bracket vanishes. The result follows immediately from Remark \ref{rem:hamronic_subalg}.

\textbf{Case $\Lambda\neq0$:}
    In this case, the Lie bracket is non-trivial. Note that $\dim(H^1(\corner,\mathbb{R}))$ is even because the Poincaré pairing of cohomology classes of forms is non-degenerate on closed oriented manifolds and by Hodge's Theorem, $\gamode_r$ is a harmonic representative for each cohomology class. Thus, the matrix $b$ has to be non-degenerate which is only the case in even dimensions.
    Another way to see this is to realize that $\corner$ is homeomorphic to a disjoint union of genus-$g$ surfaces, which have an even first Betti number.

    Set $s = \frac{1}{2}\dim(H^1(\corner,\mathbb{R}))$. For any even-dimensional skew-symmetric matrix $b$, there exists a basis $\{\Tilde{\gamode}_k,\Tilde{\gamode}_k^\dag \mid 1\leq k\leq s\}$ of $\mathfrak{h}$ such that $b$ takes the following block diagonal form:
    \begin{equation}
        b=\bigoplus_{i=1}^s \mu_i\varepsilon\,,
    \end{equation}
    where $\mu_i\geq0$ and $\varepsilon = \begin{bmatrix}
        0 & 1 \\
        -1 & 0 
    \end{bmatrix}$. In this basis, $[\Tilde{\gamode}_k, \Tilde{\gamode}_l^\dag]= -\Lambda \mu_k \delta_{k,l}\Ccharge$ and otherwise the Lie bracket vanishes. By absorbing factors, we obtain the desired isomorphism.
\end{proof}
\begin{remark}
    Implicitly, these isomorphisms should be taken with respect to the complexifications, as the oscillator algebra is complex. 
\end{remark}

\begin{example}
    These results are in good agreement with the ones of Section \ref{sec:abelianontor}. Indeed, let
    us calculate the harmonic Lie subalgebra $\mathfrak{h}$ on the Torus  $\torus=\circl\times\circl$ with coordinates $\theta,\varphi$ and metric $g=\dd \theta\otimes \dd \theta+\dd \varphi\otimes \dd \varphi$.
    Choose the harmonic representatives $\{\Tilde{\gamode}\coloneqq\dd \theta, \Tilde{\gamode}^\dag\coloneqq\dd \varphi\}$, i.e.\ a basis of $H^1(\torus,\mathbb{R})$. Note that we are abusing notation and the 1-forms are not actually exact. They are, however, closed and coclosed.

    In this basis $b=\begin{bmatrix} 0 & \mathrm{Vol}_g(\torus) \\ -\mathrm{Vol}_g(\torus) & 0 \end{bmatrix}$,
    where $\mathrm{Vol}_g(\torus)$ is the volume of the Torus with this metric. By rescaling the basis, we obtain the following Lie bracket relations
    \begin{equation*}
        [\Tilde{\gamode}, \Tilde{\gamode}^\dag]= \Ccharge,
    \end{equation*}
    and zero otherwise. In other words, $\mathfrak{h}  \cong \mathbb{R}^3$ and $\mathfrak{h}_\Lambda \cong \mathcal{A}(1)$, which agree with what was shown in Theorem \ref{thm:abisheis}.
\end{example}

\subsection{Classification of \texorpdfstring{$\Ghat$ \& $\Ghat_\Lambda$}{Ghat}}
We have the following classification theorem for the finite Fourier mode subalgebra that directly generalizes Theorem \ref{thm:abisheis}. Until the end of this section, we explicitly show the dependence of the Lie algebra on the underlying corner manifold $\corner$.
\begin{proposition}
There is an isomorphism of Lie algebras:
\begin{equation*}
    \Ghat(\corner) \cong\Ghat_\Lambda(\corner)\cong \mathcal{A}\oplus\mathfrak{a}
\end{equation*}
\end{proposition}

\begin{proof}
\textbf{Case $\Lambda=0$:}
We redefine the generators
\begin{alignat*}{3}
    &\cmode_{lm} &&\coloneqq \Amode_{lm}\quad &&\text{ for } l\neq0, 1\leq m \leq R(l)\,,\\
    &\cmode_{lm}^\dag &&\coloneqq \frac{-1}{\lambda_l}\Hmode_{lm}\quad &&\text{ for } l\neq0, 1\leq m \leq R(l)\,,\\
    &\widehat{\Amode}_m &&\coloneqq \Amode_{0m}&&\text{ for } l=0, 1\leq m \leq \dim(H^0(\corner,\mathbb{R}))\,, \\
    &\widehat{\Gmode}_{lm} &&\coloneqq \Gmode_{lm}\quad &&\text{ for } l\neq0, 1\leq m \leq R(l)\,,\\
    &\widehat{\gamode}_{r} &&\coloneqq \gamode_{r}\quad &&\text{ for } 1\leq r \leq \dim(H^1(\corner,\mathbb{R}))\,.
\end{alignat*}
In this basis, we obtain the Lie bracket relations
\begin{equation*}
    [\cmode_{lm},\cmode_{lm}^\dag]= \Ccharge
\end{equation*}
and zero otherwise. The abelian summand is spanned by the central elements excluding $\Ccharge$, i.e.\ $\mathfrak{a} \coloneqq {\langle\widehat{\Amode}_m,\widehat{\Gmode}_{ln},\widehat{\gamode}_{r}\mid 1\leq m \leq \dim(H^0(\corner,\mathbb{R})), l\neq0, 1\leq n \leq R(l)\text{ and }1\leq r \leq \dim(H^1(\corner,\mathbb{R}))\rangle}$. The resulting Lie algebra is isomorphic to the direct sum $\mathcal{A}\oplus\mathfrak{a}$.

\textbf{Case $\Lambda\neq0$:}
We redefine the generators
\begin{alignat*}{3}
    &\wmode_{lm} &&\coloneqq \frac{1}{2}\Big(\Amode_{lm}+\frac{1}{\Lambda}\Gmode_{lm}\Big)\quad &&\text{ for } l\neq0, 1\leq m \leq R(l)\,,\\
    &\wmode_{lm}^\dag &&\coloneqq \frac{-1}{\lambda_l}\Hmode_{lm}\quad &&\text{ for } l\neq0, 1\leq m \leq R(l)\,,\\
    &\widehat{\Amode}_m &&\coloneqq \Amode_{0m}&&\text{ for } l=0, 1\leq m \leq \dim(H^0(\corner,\mathbb{R}))\,,  \\
    &\widehat{\wmode}_{lm} &&\coloneqq \frac{1}{2}\Big(\Amode_{lm}-\frac{1}{\Lambda}\Gmode_{lm}\Big)\quad &&\text{ for } l\neq0, 1\leq m \leq R(l)\,.
\end{alignat*}
Furthermore, fix a basis $\{\Tilde{\gamode}_k,\Tilde{\gamode}_k^\dag \mid 1\leq k\leq s\}$ of $\mathfrak{h}$ just like in the proof of Lemma \ref{lem:hamronic_subalg_classific_lemma}, but with absorbed prefactors. In this basis, we obtain the Lie bracket relations
\begin{align*}
    [\wmode_{lm},\wmode_{lm}^\dag]&= \Ccharge,\\
    [\Tilde{\gamode}_k,\Tilde{\gamode}_k^\dag]&= \Ccharge,
\end{align*}
and zero otherwise. The abelian summand is spanned by the central elements excluding $\Ccharge$, i.e.\ $\mathfrak{a} \coloneqq {\langle\widehat{\Amode}_m,\widehat{\wmode}_{ln}\mid 1\leq m \leq \dim(H^0(\corner,\mathbb{R})),  l\neq0, 1\leq n \leq R(l)\rangle}$. The resulting Lie algebra is isomorphic to the direct sum $\mathcal{A}\oplus\mathfrak{a}$.
\end{proof}

\subsection{Constraints \texorpdfstring{$\dd A = 0$ \& $\dd A + \Lambda B= 0$}{dA + Lambda B= 0}}
Finally, we have the following classification theorem for the physical corner algebra.
\begin{proposition}\label{pro:ghat_on_surface}
    Let $\corner$ be a closed, oriented 2-dimensional surface. The physical corner algebra of abelian $BF$ theory in four dimensions on $\corner$, with and without cosmological term, is characterized by the isomorphism of Lie algebras
\begin{equation*}
    \Ghat\vert_{\dd A = 0}(\corner) \cong \mathcal{A}\oplus H^0(\corner,\mathbb{R})\oplus H^1(\corner,\mathbb{R}), \quad{\Ghat_\Lambda}\vert_{\dd A +\Lambda B= 0}(\corner)\cong \mathcal{A},
\end{equation*}
where we consider the cohomology groups as abelian Lie algebras.
\end{proposition}
\begin{proof}
The proof is essentially identical to Proof \ref{prf:abelian_corneralg_iso}.
\textbf{Case $\Lambda=0$:}
The constraints imply the following for $l \in \mathbb{N},1\leq m \leq R(l)\colon$
\begin{align*}
    0  &\overset{!}{=}\dd \eigfun{lm}= \begin{cases}
        \widehat{\Gmode}_{lm}  & \text{for } l \neq 0,1\leq m \leq R(l) \\
        0  & \text{for } l = 0\, .
  \end{cases}
\end{align*}
Therefore, only the abelian part is affected, and it becomes finite dimensional in the quotient Lie algebra. In particular
\begin{equation*}
    {\langle\widehat{\Amode}_m,\widehat{\gamode}_{r}\mid1\leq m \leq \dim(H^0(\corner,\mathbb{R})), 1\leq r \leq \dim(H^1(\corner,\mathbb{R}))\rangle} \cong H^0(\corner,\mathbb{R})\oplus H^1(\corner,\mathbb{R})
\end{equation*}
as real abelian Lie algebras.

\textbf{Case $\Lambda\neq0$:}
The constraints imply the following for $l\in\mathbb{N},1\leq m \leq R(l)\colon$
\begin{align*}
    0  &\overset{!}{=}\dd \eigfun{lm}-\Lambda \eigfun{lm}= \Gmode_{lm}-\Lambda \Amode_{lm}= \begin{cases}
        -2\Lambda \widehat{\wmode}_{lm}  & \text{for } l \neq 0,1\leq m \leq R(l)\\
        -\Lambda \widehat{\Amode}_m  & \text{for } l = 0, 1\leq m \leq \dim(H^0(\corner,\mathbb{R}))\, .
  \end{cases}
\end{align*}
Therefore, only the abelian part is affected, and it vanishes in the quotient Lie algebra.
\end{proof}

\begin{corollary}\label{cor:ghat_on_genus_surface}
    If $\corner \cong \Sigma_g$, a genus $g\in\mathbb{N}$ surface, we obtain the isomorphism of Lie algebras
    \begin{equation*}
        \Ghat\vert_{\dd A = 0}(\Sigma_g) \cong \mathcal{A}\oplus \mathbb{R}^{2g+1}, \quad{\Ghat_\Lambda}\vert_{\dd A +\Lambda B= 0}(\corner)\cong \mathcal{A}.
    \end{equation*}
\end{corollary}

\begin{remark}
    Note that Corollary \ref{cor:ghat_on_genus_surface} is consistent with our results from Section \ref{sec:abelianontor} up to complexification.
\end{remark}
\begin{remark}\label{rem:relation_to_current_algebra_2} In \cite[Section 4.1]{fliss_entanglement_2025}, the authors construct a centrally extended current algebra associated to a closed $(d-2)$-dimensional manifold $\Gamma$ from the charges that generate global symmetries in the presence of corners. More concretely, the current algebra is obtained via a canonical quantization procedure with appropriate boundary conditions and subsequent reformulation using Hodge decomposition. In the special case of $p=1$ and $d=4$, the current algebra \cite[Equation (4.35)]{fliss_entanglement_2025} is isomorphic to the physical corner algebra in Proposition \ref{pro:ghat_on_surface} for $\Lambda=0$.

The constraints in the setting of  \cite{fliss_entanglement_2025} are taken care of by fixing the redundancy in the generators of the gauge transformations before quantization. One specifically requires the forms generating the transformations to be coclosed, see \cite[Equation (4.14)]{fliss_entanglement_2025}. Subsequently, in the defining brackets in \cite[Equations (4.30) \& (4.31)]{fliss_entanglement_2025}, the condition of the forms being coclosed gets rid of the exact part in the Harmonic decomposition, which is also exactly what happens in the present setting when one passes to the quotient, cf.\ the proof of Proposition \ref{pro:ghat_on_surface}.

The authors subsequently describe the Verma modules, characters and the associated edge-mode partition function. 
Note that these higher-dimensional current algebras also appear in the setting of higher-form symmetries: \cite{Vitouladitis_2026,hofman_generalised_2025}.
\end{remark}

\subsection{Modules of \texorpdfstring{$\Ghat\vert_{\dd A }$ \& $\Ghat_\Lambda\vert_{\dd A + \Lambda B=0}$}{g}} 
To summarize, the physical corner algebras of abelian $BF$ theory in four dimensions associated with a closed oriented surface $\corner$ are:
\begin{equation*}
    A_\corner \cong \mathcal{U}(\mathcal{A}\oplus H^0(\corner,\mathbb{R})\oplus H^1(\corner,\mathbb{R})),\quad  A_\corner \cong \mathcal{U}(\mathcal{A}), 
\end{equation*}
for $\Lambda = 0$ and $\Lambda \neq 0$ respectively. As discussed in Section \ref{sec:modules_of_abelian_corner_algebra}, physically admissible state spaces of boundary components that border the corner $\corner$ are given by representations of the oscillator algebra, like the bosonic Fock space, together with a choice of $(\dim(H^0(\corner,\mathbb{R}))+\dim(H^1(\corner,\mathbb{R})))$-many or no charges respectively. 

\section{{Corner Algebras for Non-Abelian \texorpdfstring{$BF$}{BF} Theory on the Torus}}\label{sec:nonabtorus_corner_algebra}
In this section, we classify the free and physical corner algebras of four-dimensional non-abelian $BF$ on a torus with Lie algebra $\mathfrak{g} = \su$ for $\Lambda= 0$ and $\Lambda\neq 0$. 
More concretely, we proceed similarly to Section \ref{sec:abelianontor} and start by considering the Lie subalgebra generated by finite Fourier expansions and by working out the bracket relations in Section \ref{sec:non_ab_lie_alg_def}. We show in Section \ref{sec:classification_of_non_ab_lie_alg} that the Lie algebras for zero and non-zero cosmological constant are isomorphic to certain central extensions of the double-loop algebra over the 9-dimensional isochronous Galilean Lie algebra $\mathfrak{igal}(3)$. Afterwards, in Section \ref{sec:constr_discussion_non_ab}, we discuss the heuristic two-sided constraint ideals and define the physical quotient algebra. 

\subsection{The Lie Algebras \texorpdfstring{$\Ghat(\su)$ \& $\Ghat_\Lambda(\su)$}{gsu}}\label{sec:non_ab_lie_alg_def}
The Lie algebra describing the
unconstrained quantized corner structure of 4-dimensional non-abelian $BF$ on a torus with $\mathfrak{g}=\su$ is given by the vector space \begin{equation*}
    \Ghat_\Lambda(\su) = \Omega^0(\torus)\otimes \su \oplus \Omega^1(\torus)\otimes \su \oplus \mathbb{R}
\end{equation*}
with Lie bracket
\begin{equation}\label{eq:non_ab_bracket_full}
    [f\oplus\alpha \oplus r,   g\oplus\beta \oplus s]_{\Ghat_\Lambda(\su)} = [f,g]\oplus (\mathrm{ad}_f\beta - \mathrm{ad}_g \alpha)\oplus \frac{-1 }{(2\pi)^2} \int_{\torus} (\alpha \dd g - \beta \dd f + \Lambda \alpha \beta)\Ccharge\,,
\end{equation}
where we choose the trivial reference connection $A_0=0$.\footnote{A non-zero reference connection $A_0$ does not change the Lie algebra structure; the extra term can be reabsorbed by a redefinition of the generators.} Furthermore, we drop the subscript of the Lie bracket in~\eqref{eq:non_ab_bracket_full}.
We want to have an explicit description in terms of generators. Therefore, let us again choose the basis $\{t_\mu\}_{1\leq\mu\leq3}$ of $\su$, such that $(t_\mu, t_\nu) = \delta_{\mu,\nu}$ and $[t_\mu, t_\nu] = \levi_{\mu \nu}^\lambda t_\lambda$, where $(\cdot,\cdot)$ is an invariant, non-degenerate inner product on $\su$ and $\levi$ is the Levi-Civita symbol. 

Instead of the entire algebra, we will consider the Lie subalgebra of $(\Ghat_\Lambda(\su))_\mathbb{C}$ that consists
of finite Fourier modes. These elements are of the form
    \begin{alignat*}{2}
        f &= \sum_{\substack{1 \leq\mu \leq 3 \\ m,n\in \mathbb{Z}}} {f}_{\mu mn} \left( t_\mu\otimes e^{im\theta}e^{in\varphi}\right) &&\quad \text{for }f_{\mu mn}\in\mathbb{C}\,,\\
        \alpha &= \sum_{\substack{1 \leq\mu \leq 3 \\ m,n\in \mathbb{Z}}} \alpha_{\mu mn}^{(\theta)}\left( t_\mu\otimes e^{im\theta}e^{in\varphi}\right)\dd\theta + \alpha_{\mu mn}^{(\varphi)}\left( t_\mu\otimes e^{im\theta}e^{in\varphi}\right)\dd\varphi &&\quad \text{for }\alpha_{\mu mn}^{(\theta)},\alpha_{\nu mn}^{(\varphi)}\in\mathbb{C}\,,
    \end{alignat*}
where ($\theta,\varphi$) are coordinates on the torus with orientation given by the volume form $\dd\theta\wedge\dd\varphi$ and with only finitely many non-zero coefficients. 
With a slight abuse of notation, we denote the Lie subalgebra by the same symbol and drop the complexification symbol. Now that the vector space $\Ghat_\Lambda(\su)$ has countable dimension, we can choose a suitable basis$\colon$
\begin{alignat*}{2}
     &\Jmode_{\mu mn}&&\coloneqq t_\mu\otimes e^{im\theta}e^{in\varphi}\,,\\
     &\Kmode_{\mu mn}&&\coloneqq \left( t_\mu\otimes e^{im\theta}e^{in\varphi}\right)\dd\varphi\, ,\\
     &\Pmode_{\mu mn}&&\coloneqq \left( t_\mu\otimes e^{im\theta}e^{in\varphi}\right)\dd\theta\,,
\end{alignat*}
where $m,n\in \mathbb{Z}, 1\leq \mu\leq 3$. The notation will become clear later on when we make the connection to the isochronous Galilean Lie algebra. 

A quick computation yields the following lemma.
\begin{lemma}
    In the finite Fourier mode basis, the bracket relations are:
    \begin{equation}\label{eq:nonabcomrel}
    \begin{aligned}
        [\Jmode_{\mu kl},\Jmode_{\nu mn}] &=  \levi_{\mu \nu}^\lambda \Jmode_{\lambda (k+m)(l+n)}\,,\\
        [\Jmode_{\mu kl},\Kmode_{\nu mn}] &= \levi_{\mu \nu}^\lambda \Kmode_{\lambda (k+m)(l+n)} +i m\delta_{\mu,\nu}\delta_{k,-m}\delta_{l,-n}\Ccharge\,,\\
        [\Jmode_{\mu kl},\Pmode_{\nu mn}] &= \levi_{\mu \nu}^\lambda \Pmode_{\lambda (k+m)(l+n)} -in\delta_{\mu,\nu}\delta_{k,-m}\delta_{l,-n}\Ccharge\,, \\
        [\Kmode_{\mu kl},\Pmode_{\nu mn}] &= \Lambda \delta_{\mu,\nu}\delta_{k,-m}\delta_{l,-n}\Ccharge\,,
    \end{aligned}
    \end{equation}
    and all other brackets vanish.
\end{lemma}
\begin{proof}
We will only compute the non-zero brackets explicitly. Evaluating the bracket on these basis elements yields:
    \begin{align*}
        [\Jmode_{\mu kl}, \Jmode_{\nu mn}]  
         &= [t_\mu,t_\nu]\otimes e^{i(k+m)\theta}e^{i(l+n)\varphi} 
         = \levi_{\mu \nu}^\lambda \Jmode_{\lambda (k+m)(l+n)}\,,\\
        [\Jmode_{\mu kl}, \Kmode_{\nu mn}]
        &=\mathrm{ad}_{\Jmode_{\mu kl}}\Kmode_{\nu mn} +\frac{ 1}{(2\pi)^2} \int_{\torus} \Kmode_{\nu mn}\dd\Jmode_{\mu kl}\Ccharge\\ &= \left([t_\mu,t_\nu]\otimes e^{i(k+m)\theta}e^{i(l+n)\varphi}\right)\dd\varphi  + \frac{1}{(2\pi)^2}\left(\int_{\torus} ik(t_\nu,t_\mu)e^{i(m+k)\theta}e^{i(n+l)\varphi}\dd\varphi\wedge\dd\theta \right)\Ccharge\\&= \levi_{\mu \nu}^\lambda \Kmode_{\lambda (k+m)(l+n)}+i m\delta_{\mu,\nu}\delta_{k,-m}\delta_{l,-n}\Ccharge\,, \\
        [\Jmode_{\mu kl}, \Pmode_{\nu mn}]   &=\mathrm{ad}_{\Jmode_{\mu kl}}\Pmode_{\nu mn} +\frac{ 1}{(2\pi)^2} \int_{\torus} \Pmode_{\nu mn}\dd\Jmode_{\mu kl}\Ccharge
         \\&= \left([t_\mu,t_\nu]\otimes e^{i(k+m)\theta}e^{i(l+n)\varphi}\right)\dd\theta+\frac{ 1}{(2\pi)^2}\left(\int_{\torus} il(t_\nu,t_\mu)e^{i(m+k)\theta}e^{i(n+l)\varphi}\dd\theta\wedge\dd\varphi\right)\Ccharge \\
         &= \levi_{\mu \nu}^\lambda \Pmode_{\lambda (k+m)(l+n)}-i n\delta_{\mu,\nu}\delta_{k,-m}\delta_{l,-n}\Ccharge\,, \\
        [\Kmode_{\mu kl}, \Pmode_{\nu mn}]&=  -\frac{ 1}{(2\pi)^2} \int_{\torus} \Lambda\Kmode_{\mu kl} \Pmode_{\nu mn}\Ccharge
        =  -\frac{ 1}{(2\pi)^2} \left(\int_{\torus} \Lambda(t_\mu,t_\nu)e^{i(k+m)\theta}e^{i(l+n)\varphi}\dd\varphi\wedge\dd\theta\right)\Ccharge \\
         &= \Lambda \delta_{\mu,\nu}\delta_{k,-m}\delta_{l,-n}\Ccharge \, .
    \end{align*}
\end{proof} 
\subsection{Classification of \texorpdfstring{$\Ghat(\su)$ \& $\Ghat_\Lambda(\su)$}{Ghat}}\label{sec:classification_of_non_ab_lie_alg}
To provide a classification of $\Ghat(\su)$ and $\Ghat_\Lambda(\su)$, one can examine the zeroth-level subalgebras. The zeroth-level Lie algebra for zero cosmological constant is defined as the 9-dimensional subalgebra spanned by generators $\{\mathrm{X}_{\mu 00}\}_{1\leq\mu\leq3}$ for $\mathrm{X}\in\{\Jmode,\Kmode,\Pmode\}$ and is denoted $\Ghatzero(\su)_0$. The zeroth-level Lie algebra for non-zero cosmological constant is defined as the 10-dimensional subalgebra spanned by generators $\Ccharge$ and $\{\mathrm{X}_{\mu 00}\}_{1\leq\mu\leq3}$ for $\mathrm{X}\in\{\Jmode,\Kmode,\Pmode\}$ and is denoted $\Ghatzero_\Lambda(\su)_0$. It is easily verified that these actually constitute Lie subalgebras of $\Ghat(\su)$ and $\Ghat_\Lambda(\su)$ respectively.

Before we state the theorem classifying $\Ghat(\su)_0$ and $\Ghat_\Lambda(\su)_0$, we define the isochronous Galilean Lie algebra and its extension.
\begin{definition}
    The \textbf{isochronous Galilean Lie algebra $\mathfrak{igal}{(3)}$} is the 9-dimensional real vector space $\langle\Jmode_{\mu } ,\Kmode_{\mu } ,\Pmode_{\mu }\mid 1\leq\mu\leq3\rangle$ with a Lie bracket defined by
    \begin{alignat*}{2}
        &[\Jmode_{\mu },\Jmode_{\nu }] & &=  \levi_{\mu \nu}^\lambda \Jmode_{\lambda }\,,\\
        &[\Jmode_{\mu },\Kmode_{\nu }] & &= \levi_{\mu \nu}^\lambda \Kmode_{\lambda}\, ,\\
        &[\Jmode_{\mu },\Pmode_{\nu }] & &= \levi_{\mu \nu}^\lambda \Pmode_{\lambda}\, ,
    \end{alignat*}
    where $\levi$ denotes the Levi-Civita symbol. 
\end{definition}
The algebra $\mathfrak{igal}{(3)}$ is the Lie subalgebra of the Galilean Lie algebra in 3 dimensions $\mathfrak{gal}(3)$ (cf.\ \cite[(2.21a)-(2.21i)]{levyleblond71}) excluding the time generator $H$, hence the prefix isochronous. In 3 dimensions, the Galilean Lie algebra has a unique central extension which is denoted $\widehat{\mathfrak{gal}}(3)$ and is characterized by $[\Kmode_\mu,\Pmode_\nu]=\delta_{\mu,\nu}m \mathrm{I}$ \cite[(3.26)]{levyleblond71}. The isochronous part then forms a Lie subalgebra of $\widehat{\mathfrak{gal}}(3)$ again.
\begin{definition}
    The \textbf{extended isochronous Galilean Lie algebra $\mathfrak{i}\widehat{\mathfrak{gal}}(3)$} is the 10-dimensional real vector space $\langle\Jmode_{\mu } ,\Kmode_{\mu } ,\Pmode_{\mu },\mathrm{I} \mid 1\leq\mu\leq3\rangle$ with a Lie bracket defined by
    \begin{alignat*}{2}
        &[\Jmode_{\mu },\Jmode_{\nu }] & &=  \levi_{\mu \nu}^\lambda \Jmode_{\lambda }\,,\\
        &[\Jmode_{\mu },\Kmode_{\nu }] & &= \levi_{\mu \nu}^\lambda \Kmode_{\lambda}\, ,\\
        &[\Jmode_{\mu },\Pmode_{\nu }] & &= \levi_{\mu \nu}^\lambda \Pmode_{\lambda}\, ,\\
        &[\Kmode_{\mu },\Pmode_{\nu }] & &=   \delta_{\mu ,\nu}m\mathrm{I}\,,
    \end{alignat*}
    where $\levi$ denotes the Levi-Civita symbol.
\end{definition}
Now, we are ready to identify the zeroth-level Lie algebras.
\begin{theorem} There are isomorphisms of Lie algebras:
\begin{equation*}
    \Ghatzero(\su)_0 \cong \mathfrak{igal}(3), \quad \Ghatzero_\Lambda(\su)_0 \cong \mathfrak{i}\widehat{\mathfrak{gal}}(3).
\end{equation*}
\end{theorem}
\begin{proof}
The proof follows immediately by writing out the bracket relations of the zeroth-level subalgebras and identifying them as the ones of the respective Galilean Lie subalgebras. 
To that end, denote the zeroth-level generators by  $\mathrm{X}_{\mu} \coloneqq \mathrm{X}_{\mu 00}$ for $\mathrm{X}\in\{\Jmode,\Kmode,\Pmode\}$ and, for $\Lambda \neq 0$, identify $\Lambda\Ccharge= m\mathrm{I}$. Therefore, the cosmological term has the interpretation of a ``mass" from the extended-Galilei-algebra point of view.
\end{proof}
\begin{remark}
    The appearance of the Galilei algebra is not surprising and does not indicate a breaking of Lorentz invariance in the underlying theory, as there is no spacetime structure to begin with anyway. Instead, it arises because the first two summands of the bracket in Equation~\eqref{eq:non_ab_bracket_full} are reminiscent of a semi-direct structure. Furthermore, $\mathfrak{igal}(3)$ itself is isomorphic to a semi-direct product, namely $\mathfrak{igal}(3)\cong \su\ltimes(\mathbb{R}^3\oplus\mathbb{R}^3)$.
\end{remark}

We can also reverse the process and reconstruct the original corner Lie algebra from the zeroth-level Lie subalgebra.
Take the double-loop algebra over $\Ghat(\su)_0$ and centrally extend by the 2-cocycle defined in Section~\ref{sec:BFbf2v}.\footnote{Given any Lie algebra $\mathfrak{g}$, we can construct its double loop algebra $\mathfrak{g}[z,z^{-1},w,w^{-1}]\coloneqq \mathfrak{g}\otimes \mathbb{C}[z,z^{-1},w,w^{-1}]$, where the latter factor is the ring of Laurent polynomials in two variables valued in $\mathbb{C}$. Note that only a sum of finitely many powers of the variables are allowed in the ring.}\footnote{Note that taking $\Ghat_\Lambda(\su)_0$ does not work, because it already contains the central charge.}
Thus, we have the following corollary:
\begin{corollary}
    The Lie algebras $\Ghat(\su)$ and $\Ghat_\Lambda(\su)$ are isomorphic to  central extensions of the double-loop algebra $\mathfrak{igal}(3)[z,z^{-1},w,w^{-1}]$.
\end{corollary}
\begin{remark}\label{rem:cocycle_nontrivial_universal}
    The 2-cocycles are certainly non-trivial. However, it would also be interesting to know how many other distinct central extensions exist which are local in the fields.
    Additionally, note that there is a universal central extension for $\mathfrak{igal}(3)[z,z^{-1},w,w^{-1}]$ because it is perfect, i.e.\ $\mathfrak{g}=[\mathfrak{g},\mathfrak{g}]$ (see \cite[Proposition 1.3]{Kallen_1973}). 
\end{remark}

\subsection{Constraints \texorpdfstring{$F_A = 0$ \& $F_A +\Lambda B = 0$}{FA = 0}}\label{sec:constr_discussion_non_ab}
In the non-abelian theory, the constraints are not linear and thus not represented by elements of the Lie algebra. Instead, they lie in the UEA of the latter, i.e.\ $\mathcal{I}_{F_A}\subset\mathcal{U}(\Ghat(\su))$ and $\mathcal{I}_{F_A +\Lambda B}\subset\mathcal{U}(\Ghat_\Lambda(\su)$) (recall the heuristics in Section~\ref{sec:constr}). Furthermore, they are no longer central but instead constitute two-sided ideals in the UEA. Heuristically, we can then form the physical quotient UEAs

\begin{align*}
    \mathcal{U}(\Ghat(\su))\vert_{F_A= 0} &\coloneqq \faktor{\mathcal{U}(\Ghat(\su)) }{ \mathcal{I}_{F_A}} & \text{ and} & & \mathcal{U}(\Ghat_\Lambda(\su))\vert_{F_A+\Lambda B= 0} &\coloneqq \faktor{\mathcal{U}(\Ghat_\Lambda(\su)) }{ \mathcal{I}_{F_A +\Lambda B}}.
\end{align*} 
In the following, we just treat the $\Lambda\neq0$ case since it works almost identically for zero cosmological constant.
The Poisson ideal of constraints $I_{F_A+\Lambda B}$ in $C^\infty(\cfsppolarization_0)$ is generated by the functionals:
\begin{align*}
    \frac{1}{(2\pi)^2}\int_\torus f(F_A +\Lambda B),
\end{align*}
where $f\in\Omega^0(\torus)\otimes\su$ and the factor is purely for convenience.
We can decompose this multilinear form into an infinite sum of products of linear functionals.
\begin{proposition} The constraint functionals, expressed in terms of an infinite sum of linear functionals, are given by
\begin{equation*} 
    \widehat{f}=\sum_{r,s}\sum_\lambda {f}_{\lambda r s }\left( -is \Kmode_{\lambda r s} -ir\Pmode_{\lambda r s} +\Lambda \Jmode_{\lambda r s}+\sum_{m,n}\sum_{\mu,\nu}\varepsilon_{\mu\nu}^\lambda \Pmode_{\mu (r+m)(s+n)}\Kmode_{\nu -m -n}\right).
\end{equation*}
\end{proposition}
\begin{proof} Denote $t_{\mu mn}\coloneqq t_\mu\otimes e^{im\theta}e^{in\varphi}$:
\begin{align*}
    \widehat{f}(A,B)&= \frac{1}{(2\pi)^2}\int_\torus (f\overset{\wedge}{,}\dd A + \frac{1}{2}[A,A] +\Lambda B)\\
    &= \frac{{f}_{\lambda r s }}{(2\pi)^2}\int_\torus \Bigl(\Jmode_{\lambda r s} \overset{\wedge}{,}(imA_{\mu mn}^{(\varphi)}-inA_{\mu m n}^{(\theta)}+ \Lambda B_{\mu mn})t_{\mu m n}\dd \theta \wedge \dd \varphi  \\ &  \quad + \varepsilon_{\mu \nu}^\rho A_{\mu m n}^{(\theta)}A_{\nu kl}^{(\varphi)} t_{\rho (m+k)(n+l)})\dd \theta \wedge \dd \varphi\Bigr)\\
    &=  {f}_{\lambda r s }(imA_{\lambda mn}^{(\varphi)}-inA_{\lambda m n}^{(\theta)}+ \Lambda B_{\lambda mn})\frac{1}{(2\pi)^2}\int_\torus e^{i(r+m)\theta}e^{i(s+n)\varphi}\dd \theta \wedge \dd \varphi  \\ & \quad + {f}_{\lambda r s }\varepsilon_{\mu \nu}^\lambda A_{\mu m n}^{(\theta)}A_{\nu kl}^{(\varphi)}\frac{1}{(2\pi)^2}\int_\torus e^{i(r+m+k)\theta}e^{i(s+n+l)\varphi}\dd \theta \wedge \dd \varphi\\
    &= {f}_{\lambda r s }( -irA_{\lambda -r-s}^{(\varphi)}+isA_{\lambda -r-s}^{(\theta)}+ \Lambda B_{\lambda -r-s} + \varepsilon_{\mu \nu}^\lambda A_{\mu m n}^{(\theta)}A_{\nu (-r-m)(-s-n)}^{(\varphi)})\\
    &=\sum_{r,s}\sum_\lambda {f}_{\lambda r s }\left( -is \Kmode_{\lambda r s} -ir\Pmode_{\lambda r s} +\Lambda \Jmode_{\lambda r s}+\sum_{m,n}\sum_{\mu,\nu}\varepsilon_{\mu\nu}^\lambda \Pmode_{\mu (r+m)(s+n)}\Kmode_{\nu -m -n}\right)(A,B)\,.
\end{align*}
The first equality follows by definition, and the second by the expansion of the function and the forms in terms of the finite Fourier basis and the Lie algebra basis. The third equality is obtained by contracting the inner product. One then integrates over the volume form to obtain the fourth equality. Finally, the linear functionals can be restored using the definition in terms of the non-degenerate pairing.
\end{proof}

Next, we introduce a preferred basis of $\mathcal{U}(\Ghat_\Lambda(\su))$. By the Poincaré--Birkhoff--Witt (PBW) Theorem,\footnote{For more details on the PBW Theorem, see \cite[Section 17.3]{Humphreys_1973}.} the monomials
\begin{equation*}
    \{\Jmode_{\mu_0 k_0 l_0}\cdots \Jmode_{\mu_a k_a l_a}\Kmode_{\nu_0 m_0 n_0}\cdots \Kmode_{\nu_b m_b n_b}\Pmode_{\rho_0 r_0 s_0}\cdots \Pmode_{\rho_c r_c s_c}\}\,,
\end{equation*}
where $a,b,c\in\mathbb{N}$, $k_0 \leq \cdots \leq k_a,  \cdots , s_0 \leq \cdots \leq s_c$ and $\mu_0 \leq \cdots \leq\mu_a, \cdots,\rho_0 \leq \dotsc\leq\rho_c $, along with 1, form a basis of the UEA. We are actually considering a quotient of the UEA by the ideal generated by ${1-\Ccharge}$ since we are only interested in modules where the central charge acts by 1. 
The expression $\widehat{f}$ descends to a well-defined equivalence class in the UEA. This follows because the generators in the quadratic term of $\widehat{f}$ commute when summing over the totally-antisymmetric structure constants. We will again denote the resulting equivalence class by $\widehat{f}$. These elements represent the quantization of the constraint functionals.
\begin{remark}
    Since the sum is infinite, the quantized constraints should formally be part of some appropriate completion of the UEA. Instead of resolving this issue, we will continue assuming that these infinite sums can be made sense of and accept the resulting statements as heuristics.
\end{remark}
 
Next, we want to examine whether the ideal generated by the set of constraints is a proper ideal in the UEA. To that end, we have the following proposition
\begin{proposition}\label{prp:quant_constr_relation}
    The set of constraints  $\{\widehat{f}\in \mathcal{U}(\Ghat_\Lambda(\su))\}$ satisfies the bracket relations
    \begin{alignat*}{2}
        &[\widehat{f}, \Jmode_{\mu kl}] &&= \sum_{r,s}f_{\lambda rs}\sum_{\rho} \varepsilon_{\lambda\mu}^\rho \widehat{f}_{\rho (r+k)(s+l)}\,,\\
        &[\widehat{f}, \Kmode_{\mu kl}] &&= 0 \,,\\
        &[\widehat{f}, \Pmode_{\mu kl}] &&= 0\,,
    \end{alignat*}
     where $\widehat{f}_{\lambda rs}\coloneqq -is \Kmode_{\lambda r s} -ir\Pmode_{\lambda r s} +\Lambda \Jmode_{\lambda r s}+\sum_{m,n}\sum_{\mu,\nu}\varepsilon_{\mu\nu}^\lambda \Pmode_{\mu (r+m)(s+n)}\Kmode_{\nu -m -n}$, for all $1\leq\mu\leq3$ and $k,l\in\mathbb{Z}$. By the derivation property of the bracket, the set generates a proper two-sided ideal $\mathcal{I}_{F_A +\Lambda B}$ in $\mathcal{U}(\Ghat_\Lambda(\su))$.
\end{proposition}
\begin{proof}
    {This is the proof of Proposition 3.9 in \cite{Leupp2025}.}
\end{proof}
\begin{remark}
As a consistency check, one can examine how the ideal of constraints behaves when restricting to the zeroth-level algebra. In particular, the restricted constraints $\widehat{f}_\lambda \coloneqq \Lambda\Jmode_{\lambda}+\varepsilon_{\mu \nu}^\lambda \Pmode_\mu \Kmode_\nu$ satisfy the bracket relations:
\begin{alignat*}{2}
        &[\widehat{f}_{\mu }, \Jmode_{\nu }] &&= \sum_\lambda \varepsilon_{\mu\nu}^\lambda \widehat{f}_{\lambda}\,,\\
        &[\widehat{f}_{\mu }, \Kmode_{\nu }] &&= 0 \,,\\
        &[\widehat{f}_{\mu }, \Pmode_{\nu}] &&= 0\,,
    \end{alignat*}
and lie in $\mathcal{U}(\mathfrak{igal}(3))$ and $\mathcal{U}(\mathfrak{i}\widehat{\mathfrak{gal}}(3))$ for $\Lambda=0$ and $\Lambda\neq0$ respectively. One can then consider their squared sum, producing central elements in the UEA. These are precisely the Casimir elements (2.22b) and (3.27c) found in \cite{levyleblond71} up to a proportionality factor. The second Casimirs can be interpreted as a sort of intrinsic angular momentum of a particle transforming under the extension of the Galilean algebra \cite[Section V.3.a.]{levyleblond71}.
\end{remark}

We can conclude that the physical corner algebras of non-abelian $BF$ theory with Lie algebra $\su$ in four dimensions associated with a torus are given by

\begin{equation*}
    \begin{aligned}
    A_\torus \cong 
    \mathcal{U}(\mathfrak{igal}(3)[z,z^{-1},w,w^{-1}]\oplus\mathbb{R})\Bigg/\Bigg< &-is \Kmode_{\lambda r s} -ir\Pmode_{\lambda r s} \\&+\sum_{m,n}\sum_{\mu,\nu}\varepsilon_{\mu\nu}^\lambda \Pmode_{\mu (r+m)(s+n)}\Kmode_{\nu -m -n}\vert s,r\in \mathbb{Z},1\leq\lambda\leq3\Bigg>
    \end{aligned}
\end{equation*}
and 
\begin{equation*}
    \begin{aligned}
    A_\torus \cong 
    \mathcal{U}(\mathfrak{igal}(3)[z,z^{-1},w,w^{-1}]\oplus\mathbb{R})\Bigg/\Bigg< &-is \Kmode_{\lambda r s} -ir\Pmode_{\lambda r s} +\Lambda \Jmode_{\lambda r s} \\&+\sum_{m,n}\sum_{\mu,\nu}\varepsilon_{\mu\nu}^\lambda \Pmode_{\mu (r+m)(s+n)}\Kmode_{\nu -m -n}\vert s,r\in \mathbb{Z},1\leq\lambda\leq3\Bigg>
    \end{aligned}
\end{equation*}
with respect to different central extensions for $\Lambda=0$ and $\Lambda\neq0$.

In the next step, we want to obtain modules for $A_\torus$ by constructing modules of the free corner algebras $\mathfrak{igal}(3)[z,z^{-1},w,w^{-1}]\oplus\mathbb{R}$ where the ideal of constraints acts trivially, i.e.\ modules that descend to the quotient algebra. 

\section{{Modules of the Free Corner Algebra of Non-Abelian \texorpdfstring{$BF$}{bf} Theory on the Torus}}\label{sec:nonabtorus_corner_modules}
In this section, we construct families of simple modules of the free corner algebra of the non-abelian theory in terms of differential operators on a space of polynomials. 
Obtaining representations of the corner algebras is not straightforward as $\mathfrak{igal}(3)$ is non-semisimple.
We use the induced module construction applied to a so-called modified triangular decomposition (MTD) of the Lie algebra to construct a family of representations {of the free corner algebra} in Section \ref{sec:induced_module_con}. 
The so constructed modules have nice properties such as a grading and a vacuum vector and can be viewed as analogs of highest-weight representations in this setting. 
In Section \ref{subsec:repsGhat}, we explicitly realize these modules as second-order differential operators on a space of polynomials, essentially constituting a Wakimoto module (cf.\ the discussion in Section~\ref{sec:failed_attempts}) for vanishing cosmological constant. We prove that these modules are irreducible.
Finally, we address how to modify the modules so that the ideal of constraints has a well-defined action. However, we show that the constraints cannot be imposed non-trivially.
In future works, we will try to understand how this can be remedied. In the last part, Section \ref{sec:ghat_lambda_modules}, we give a similar construction for non-zero cosmological constant that likely suffers from the same flaw.

\subsection{Induced Module Construction}\label{sec:induced_module_con}
There is a vast literature on representations of infinite-dimensional Lie algebras. However, there are few theorems that hold in full generality and often only selected examples are worked out.
In the following Section~\ref{sec:failed_attempts}, we will highlight some approaches to construct representations that seemed promising but ultimately did not prove fruitful.
The core of the problem is the fact that $\mathfrak{igal}(3)$ is not semisimple.

Fortunately, the failure of the standard highest-weight module construction can be controlled. We introduce a decomposition of the Lie algebra in Section~\ref{sec:modified_triang_dec} inspired by the standard triangular decomposition. Afterwards, in Section~\ref{sec:induced_module_for_ghat}, we apply the induced module construction similar to the one for Verma modules. This modified construction produces two families of representations parametrized by a ``vacuum representation'' instead of a vacuum vector. We can then apply these considerations to the Lie algebra describing the free corner structure of 4-dimensional non-abelian $BF$ theory on the torus. The representation space can be identified with that of polynomials in infinitely many variables. Furthermore, there is a grading induced by the degree of the polynomials and a distinct vacuum-like vector.

\subsubsection{{Alternative Approaches}}\label{sec:failed_attempts}    
One can construct representations of $\Ghat(\su)$ and $\Ghat_\Lambda(\su)$ by using representations of the Galilei Lie algebra that restrict to $\mathfrak{igal}(3)$ and then lift them to the double-loop algebra and finally extend them trivially to the full algebras. However, these representations are undesirable since the central charge is manifestly represented by zero. 
Another possibility is to try to non-trivially extend representations of the underlying double loop algebra using Lau's construction \cite{Lau_2005}. However,
one cannot choose the central extension arbitrarily, i.e.\ the cocycle is determined
by the construction itself, up to some limited choices. For most reasonable guesses, the resulting cocycle is very different from the one of interest.
Other reasonable approaches include trying to adapt the Wakimoto modules (cf.\ Theorem \ref{thm:waki_free_field_rep_of_sl2}) introduced by \cite{Wakimoto_1986} for $\widehat{\mathfrak{sl}}(2)$ and generalized by \cite{Feigin_Frenkel_1988} and others to general affine Lie algebras. These representations have even been extended to the double-loop setting in \cite{YOUNG20211} and \cite{Franzini:2024qhp}. One can also attempt to use the free field constructions, which are relevant in CFT, as they embed the complicated symmetry algebra into a much simpler algebra coming from free field modes. 
However, these constructions crucially use the (semi)simplicity of the original finite-dimensional Lie algebra, so it is not clear how to adapt them to the present case.
There is also the framework unifying finite-dimensional semisimple, affine, toroidal, and other Lie algebras called extended affine Lie algebras (EALAs) (cf.\ \cite{Allison_Azam_Berman_Gao_Pianzola_1997, Neher2004}). However, it is also not clear how to deal with the non-semisimplicity in this case.
A further possibility might be to describe representations of some larger Lie algebra and obtain the one in question through a contraction procedure. For example, there is a contraction from $\mathfrak{so}(5)$ to the Poincaré algebra and further to the Galilean algebra. Representations of the former are well under control, even for extended double-loop algebras (see EALAs above). So one could hope to extend the contraction to the infinite-dimensional case. 

To the best of our knowledge, there are no results regarding the construction of useful representations for the present case. 

\subsubsection{Modified Triangular Decomposition and Induced Modules}\label{sec:modified_triang_dec}
The standard Verma module construction, e.g.\ Chapter 9 in \cite{KacRaiBombay87}, starts with a triangular decomposition. Inspired by that decomposition, we define a modified and abstract version thereof.
\begin{definition}
    Let $\mathfrak{g}$ be a Lie algebra over $\mathbb{K}=\mathbb{R},\mathbb{C}$, potentially infinite dimensional. A \textbf{modified triangular decomposition (MTD)} of $\mathfrak{g}$ is a decomposition of the Lie algebra into a direct sum of vector subspaces 
    \begin{equation*}
        \mathfrak{g} = \mathfrak{n^-}\oplus \mathfrak{h}\oplus \mathfrak{n^+},
    \end{equation*}
    such that the following equations hold
    \begin{align*}
        &[\mathfrak{h},\mathfrak{h}] \subseteq \mathfrak{h}, \\
        &[\mathfrak{h},\mathfrak{n^\pm}] \subseteq \mathfrak{n^\pm}, \\
        &[\mathfrak{n^\pm},\mathfrak{n^\mp}] \subseteq \mathfrak{h},\\
        &[\mathfrak{n^\pm},\mathfrak{n^\pm}] =\{0\}\,.
    \end{align*}
\end{definition}
These equations are compatible with the following $\mathbb{Z}$-grading: $ \mathfrak{g}_0 \coloneqq \mathfrak{h}, \mathfrak{g}_{-1}\coloneqq\mathfrak{n}^-,\mathfrak{g}_{1}\coloneqq\mathfrak{n}^+$ and $\mathfrak{g}_\alpha \coloneqq \{0\}$ for $\alpha \neq -1,0,1$. The essential structural difference to the usual triangular decomposition is that $\mathfrak{h}$ is not required to be abelian.  Next, we mimic the definition of the Verma module.

Let $(V,\hrep)$ be a representation of $\mathfrak{h}$, not necessarily 1 dimensional. Equivalently, $V$ is a $\mathcal{U}(\mathfrak{h})$-module. Note that $\mathfrak{h}\oplus\mathfrak{n^\pm}\,$ is a subalgebra of $\mathfrak{g}$. Similarly to the usual highest-weight definition, we can extend the module structure on $V$ to a $\mathcal{U}(\mathfrak{h}\oplus\mathfrak{n^\pm})$-module by defining
\begin{alignat*}{2}
    h\cdot v&\coloneqq \hrep(h)v \quad&&\forall h \in\mathfrak{h},\forall v\in V,\\
    x\cdot v&\coloneqq 0 \quad &&\forall x \in\mathfrak{n^\pm},\forall v\in V,
\end{alignat*}
and the rest by linearity and the usual composition law.
\begin{lemma}
    This action is well defined.
\end{lemma} 
\begin{proof}
It is sufficient, in fact equivalent, to show that $V$ is a representation of the Lie algebra $\mathfrak{h}\oplus\mathfrak{n^\pm}$. For clarity's sake, we label the $\mathcal{U}(\mathfrak{h}\oplus\mathfrak{n^\pm})$-action explicitly by $\rep$. Linearity of the action follows from:
    \begin{alignat*}{2}
        &\rep(\alpha h + h') = \hrep(\alpha h + h') = \alpha \hrep(h) + \hrep(h') = \alpha \rep(h) + \rep(h') &&\quad \forall h,h' \in \mathfrak{h}, \forall\alpha \in \mathbb{K}\,, \\
        &\rep(\alpha x + x') = 0= \alpha \rep(x) + \rep(x') &&\quad \forall x,x' \in \mathfrak{\mathfrak{n}^\pm}, \forall\alpha \in \mathbb{K}\,, \\
        &\rep(\alpha h + \beta x) =  \alpha\hrep( h)+0 = \alpha \rep(h) + \beta \rep(x) &&\quad \forall h\in\mathfrak{h},\forall x \in \mathfrak{\mathfrak{n}^\pm}, \forall\alpha,\beta \in \mathbb{K}\,.
    \end{alignat*}
and the homomorphism property from:
    \begin{alignat*}{2}
        &[\rep(h),\rep(h')] = [\hrep(h),\hrep(h')] = \hrep([h,h']) = \rep([h,h']) &&\quad \forall h,h' \in \mathfrak{h}\,,  \\
        &[\rep(x),\rep(x')] = [0,0]= 0=\rep([x,x']) &&\quad \forall x,x' \in \mathfrak{\mathfrak{n}^\pm}\,, \\
        &[\rep(h),\rep(x)] = [\hrep(h),0] = 0 = \rep([h,x]) &&\quad  \forall h \in \mathfrak{h},\forall x \in \mathfrak{\mathfrak{n}^\pm}
        \,.
    \end{alignat*}
\end{proof}
This enhanced module structure allows us to define the Verma-module equivalent.
\begin{definition}\label{def:ind_mod}
    The \textbf{induced $\mathcal{U}(\mathfrak{g})$-module $ \mathrm{M}_\hrep^\pm$} is defined by
    \begin{equation*}
        \mathrm{M}_\hrep^\pm\coloneqq \mathcal{U}(\mathfrak{g})\otimes_{\mathcal{U}(\mathfrak{h}\oplus\mathfrak{n^\pm})}V\,.
    \end{equation*}
\end{definition}
Compared to the usual highest-weight module, the generalization essentially allows for an entire vacuum subspace generating the representation rather than a vacuum vector. Next, we work out a more explicit description of the module. 
\begin{proposition}[{\cite[Proposition 2.5.15.]{E19}}]\label{thm:n-module}
    As a left $\mathcal{U}(\mathfrak{n}^\mp)$-module $\mathrm{M}_\hrep^\pm \cong \mathcal{U}(\mathfrak{n}^\mp)\otimes_\mathbb{K}V$
\end{proposition}
\begin{proof} The proof in \cite[Proposition 2.5.15.]{E19} goes as follows
    \begin{align*}
        \mathrm{M}_\hrep^\pm &= \mathcal{U}(\mathfrak{g})\otimes_{\mathcal{U}(\mathfrak{h}\oplus\mathfrak{n^\pm})}V\\
        &\cong ( \mathcal{U}(\mathfrak{n}^\mp)\otimes_\mathbb{K}\mathcal{U}(\mathfrak{h}\oplus\mathfrak{n}^\pm))\otimes_{\mathcal{U}(\mathfrak{h}\oplus\mathfrak{n^\pm})}V \\
        &\cong \mathcal{U}(\mathfrak{n}^\mp)\otimes_\mathbb{K}\underbrace{(\mathcal{U}(\mathfrak{h}\oplus\mathfrak{n}^\pm)\otimes_{\mathcal{U}(\mathfrak{h}\oplus\mathfrak{n^\pm})}V)}_{{\cong V}}\\
         &\cong \mathcal{U}(\mathfrak{n}^\mp)\otimes_\mathbb{K}V\,,
    \end{align*}
    where the isomorphism on the second line is meant in terms of $\mathcal{U}(\mathfrak{n}^\mp)$-$\mathcal{U}(\mathfrak{h}\oplus\mathfrak{n^\pm})$-bimodules. 
\end{proof}
This proposition allows us to find a convenient basis for the underlying vector space. For concreteness, let us restrict to at most countably infinite-dimensional Lie algebras and representations. This is technically not necessary, but it will make the following notation less cumbersome.
Let $\{v_i\}_{i\in\mathbb{Z}}$ be a basis of the representation $V$, and let $\{\mathrm{X}^\pm_j\}_{j\in\mathbb{Z}}$ be bases of the Lie algebras $\mathfrak{n}^\pm$. By the PBW theorem, a basis of $\mathcal{U}(\mathfrak{n}^\mp)\otimes_\mathbb{K}V$ is given by the set
\begin{equation}\label{eq:pbwbassis_induced}
    \{\mathrm{X}_{j_0}^\mp \cdots \mathrm{X}_{j_n}^\mp\otimes_\mathbb{K} v_i\}_{j_0 \leq \cdots \leq j_n,\;n\in\mathbb{N},\;i\in\mathbb{Z}}\,.
\end{equation}
The basis vectors arising from $1\otimes_\mathbb{K} v_i$ in of~\eqref{eq:pbwbassis_induced} are always understood to be included as well. By the isomorphism above, the set
\begin{equation} \label{eq:basisstates}
    \{\mathrm{X}_{j_0}^\mp \cdots \mathrm{X}_{j_n}^\mp\otimes_{\mathcal{U}(\mathfrak{h}\oplus\mathfrak{n^\pm})} v_i\}_{j_0 \leq \cdots \leq j_n,\;n\in\mathbb{N},\;i\in\mathbb{Z}}\,,
\end{equation}
is a basis of the module $\mathrm{M}_\hrep^\pm$. 

We can now attempt to give more meaning to this space. The $\mathrm{X}^\pm$-operators commute among each other; therefore, they generate (bosonic) Fock-type states. This motivates the notation from physics:
\begin{notation}
    We denote the basis states (\ref{eq:basisstates}) of the induced module $\mathrm{M}_\hrep^\pm$ by 
    \begin{equation*}
        \ket{x_{j_0}^\mp \cdots x_{j_n}^\mp; v_i}_{j_0 \leq \cdots \leq j_n,\;n\in\mathbb{N},\;i\in\mathbb{Z}}\,.
    \end{equation*}
\end{notation}
The generators and their representations are denoted by the same letter and we use the words generator and operator interchangeably.
From now on, we identify vectors in $V$ with their image under the embedding, i.e.\ $v = 1 \otimes_{\mathcal{U}(\mathfrak{h}\oplus\mathfrak{n^\pm})}v$ if there are no ambiguities. To refer to $V$, we also speak of the \textbf{vacuum sector} because its elements get annihilated by the operators in $\mathfrak{n^\pm}$. These operators will be called \textbf{lowering} or \textbf{annihilation operators} and the ones in $\mathfrak{n^\mp}$ will be called \textbf{raising} or \textbf{creation operators}. The analogy to ladder operators only works when acting on states in $V$. In general, the action of the various generators only remotely resembles that of the usual operators.
\begin{remark}
    There are many questions raised by this construction. It would be interesting to explore what types of Lie algebras admit a MTD and under which conditions two induced modules are equivalent for different vacuum sectors $V$ and $V'$. In general, it would be instructive to examine how much of the theory of highest-weight representations carries over. For example, one should be able to define singular sectors as a generalization of singular vectors.
\begin{definition}
    A \textbf{singular sector} $W$ is a non-empty vector subspace of $\mathrm{M}_\hrep^\pm$ such that:
    \begin{align*}
        \mathfrak{h}\cdot W&\subset W\,,\\
        \mathfrak{n}^\pm\cdot W&= 0\,.
    \end{align*}
\end{definition}
The vacuum sector $V$ is an example of a singular sector. In analogy to Verma modules, one could define primitive singular sectors and examine if they are in some correspondence with submodules of $\mathrm{M}_\hrep^\pm$.
\end{remark}

\subsubsection{Induced Module Construction for \texorpdfstring{$\Ghat(\su)$ \& $\Ghat_\Lambda(\su)$}{Ghat}}\label{sec:induced_module_for_ghat}
The following example will be the cornerstone of the application to $BF$ theory. 
\begin{example}[Existence of MTD for \texorpdfstring{$\Ghat(\su)$ \& $\Ghat_\Lambda(\su)$}{Ghat}]\label{ex:MDTGhat}
Define new basis generators by $\mathrm{X}^\pm \coloneqq \pm \mathrm{X}_1 -i\mathrm{X}_2 $ and $ \mathrm{X}^z \coloneqq-2i\mathrm{X}_3$ for $\mathrm{X}\in\{\Jmode,\Kmode,\Pmode\}$. Then, one can show the new relations:
    \begin{alignat*}{2}
        &[\Jmode_{ kl}^+,\Jmode_{ mn}^-] &&= \Jmode_{ (k+m)(l+n)}^z\,,\\
        &[\Jmode_{ kl}^z,\Jmode_{ mn}^\pm] &&= \pm 2 \Jmode_{ (k+m)(l+n)}^\pm\,,\\
        &[\Jmode_{ kl}^+,\Kmode_{ mn}^-] &&= [\Kmode_{ kl}^+,\Jmode_{ mn}^-] =\Kmode_{ (k+m)(l+n)}^z - 2im\delta_{k,-m}\delta_{l,-n}\Ccharge\,,\\
        &[\Kmode_{ kl}^z,\Jmode_{ mn}^\pm] &&= \pm 2 \Kmode_{ (k+m)(l+n)}^\pm\,,\\
        &[\Jmode_{ kl}^z,\Kmode_{ mn}^\pm] &&= \pm 2 \Kmode_{ (k+m)(l+n)}^\pm\,,\\
        &[\Jmode_{ kl}^z,\Kmode_{ mn}^z] &&= -4im\delta_{k,-m}\delta_{l,-n}\Ccharge\,,\\
        &[\Jmode_{ kl}^+,\Pmode_{ mn}^-] &&= [\Pmode_{ kl}^+,\Jmode_{ mn}^-] =\Pmode_{ (k+m)(l+n)}^z +2in\delta_{k,-m}\delta_{l,-n}\Ccharge\,,\\
        &[\Pmode_{ kl}^z,\Jmode_{ mn}^\pm] &&= \pm 2 \Pmode_{ (k+m)(l+n)}^\pm\,,\\
        &[\Jmode_{ kl}^z,\Pmode_{ mn}^\pm] &&= \pm 2 \Pmode_{ (k+m)(l+n)}^\pm\,,\\
        &[\Jmode_{ kl}^z,\Pmode_{ mn}^z] &&= 4in\delta_{k,-m}\delta_{l,-n}\Ccharge\,,\\
        &[\Kmode_{ kl}^+,\Pmode_{ mn}^-] &&= [\Kmode_{ kl}^-,\Pmode_{ mn}^+] = -2\Lambda\delta_{k,-m}\delta_{l,-n}\Ccharge\,, \\
        &[\Kmode_{ kl}^z,\Pmode_{ mn}^z] &&= -4\Lambda\delta_{k,-m}\delta_{l,-n}\Ccharge\,. \\
    \end{alignat*}
Note that by setting $\Lambda= 0$, we obtain the relations for $\Ghat(\su)$ instead. The relations of the generators are very reminiscent of ordinary $\su$ on a structural level. This observation can be explained by the fact that $\mathfrak{igal}(3)$ is isomorphic to the Lie algebra $\su\ltimes(\mathbb{R}^3\oplus\mathbb{R}^3)$, so much of the structure of $\su$ persists. From the relations, we can infer that the Lie algebras $\Ghat(\su)$ and $\Ghat_\Lambda(\su)$ allow for a modified triangular decomposition given by
    \begin{alignat*}{2}
        &\mathfrak{h} &&\coloneqq \langle\Jmode_{kl}^z,\Kmode_{kl}^z,\Pmode_{kl}^z\,, \Ccharge\mid k,l\in\mathbb{Z}\rangle\,, \\
        &\mathfrak{n^\pm}&&\coloneqq \langle\Jmode_{kl}^\pm,\Kmode_{kl}^\pm,\Pmode_{kl}^\pm\mid k,l\in\mathbb{Z}\rangle\,.
    \end{alignat*}
\end{example}
The induced module considerations can now be applied to the Lie algebra describing the corner structure of 4-dimensional non-abelian $BF$ theory on the torus. Let $\Ghat_\Lambda(\su)\cong\mathfrak{n}^-\oplus\mathfrak{h}\oplus\mathfrak{n}^+$ be the MTD of Example~\ref{ex:MDTGhat} and similarly for $\Lambda= 0$. First, we need to determine a suitable representation of $\mathfrak{h}$. 
To find such representations, we can use the following fact.
\begin{lemma}[Embedding of the Abelian Case]\label{lem:embedding_abelian} There is an  isomorphism of Lie algebras:
    \begin{equation*}
        \mathfrak{h} \cong \Ghat_\Lambda\,,
    \end{equation*}
    where $\Ghat_\Lambda$ denotes the abelian Lie algebra from Section \ref{sec:abelianontor} with brackets \eqref{eq:abEandPhcoupling}-\eqref{eq:abPhandThcoupling}.
\end{lemma}
\begin{proof}
To see this, consider the linear map
$\iota\colon\Ghat_\Lambda\longrightarrow\mathfrak{h}$ defined on the basis by$\colon\Emode_{kl} \longmapsto \frac{i}{2}\Jmode_{kl}^z, \Phmode_{kl} \longmapsto \frac{i}{2}\Kmode_{kl}^z,\Thmode_{kl} \longmapsto \frac{i}{2}\Pmode_{kl}^z\,$ and $\Ccharge\longrightarrow\Ccharge$.
This map is an isomorphism of Lie algebras, as one can easily check. 
\end{proof}
\begin{remark}
    Consequently, the abelian case is ``nicely" embedded in the non-abelian case. This embedding is not unique, however, as there is a second choice for the isomorphism, namely $\Emode_{kl} \longmapsto \frac{-i}{2}\Jmode_{kl}^z, \Phmode_{kl} \longmapsto \frac{-i}{2}\Kmode_{kl}^z,\Thmode_{kl} \longmapsto \frac{-i}{2}\Pmode_{kl}^z\,$. If $\Lambda = 0$, there are even uncountably many choices. We will choose $\Emode_{kl} \longmapsto \frac{i}{2}\Jmode_{kl}^z, \Phmode_{kl} \longmapsto \frac{i}{2}\Kmode_{kl}^z,\Thmode_{kl} \longmapsto \frac{i}{2}\Pmode_{kl}^z\,$.
\end{remark} 
In Theorem~\ref{thm:abisheis}, we have shown that the free corner algebra of abelian $BF$ is isomorphic to the oscillator algebra together with the countably infinite-dimensional abelian Lie algebra. Therefore, the trivial representation $\mathbb{C}$ or the bosonic Fock space $V$ constitute possible representations. The former would force the central charge to act trivially, which we want to avoid.
Therefore, we choose the irreducible bosonic Fock space representation $(V,\hrep)$ of $\mathfrak{h}$ and define the induced modules in \ref{def:ind_mod},\footnote{The space of polynomials is already irreducible under the action of $\mathcal{A}$ so adding the abelian summand does not change that fact.} which we will denote $\Mrep_\hrep^\pm$ and $\Mrep_{\hrep,\Lambda}^\pm$ respectively. It turns out that these modules have some nice properties.

The module $V$ possesses a unique vacuum vector: the constant polynomial 1, which we shall denote by $\ket{0}$. Similarly, there is a distinguished vector in the induced representations.
\begin{definition}
    The \textbf{vacuum vector} or \textbf{vacuum state} of $\Mrep_{\hrep,\Lambda}^\pm$ is defined as
    \begin{equation*}
        \ket{0}_{\Mrep_{\hrep,\Lambda}^\pm} \coloneqq 1\otimes_{\mathcal{U}(\mathfrak{h}\oplus\mathfrak{n^\pm})}\ket{0},
    \end{equation*}
    where 1 is the identity element in the unital algebra $\mathcal{U}(\mathfrak{\Ghat}_\Lambda(\su))$. 
\end{definition}
Similarly, we define the vacuum vector of $\Mrep_\hrep^\pm$. We now focus on $\Mrep_\hrep^+$ and $\Mrep_{\hrep,\Lambda}^+$ for concreteness and drop the respective minus sign in the labels of the basis states to reduce the amount of clutter in the notation. 
Thus, the basis states of the induced modules look as follows:
\begin{equation*}
    \ket{j_{k_0 l_0} \cdots j_{k_a l_a}\ k_{m_0 n_0} \cdots k_{m_b n_b} \ p_{r_0 s_0}\cdots p_{r_c s_c}  ; v_i},
\end{equation*}
where $a,b,c\in\mathbb{N}$, $k_0 \leq \cdots \leq k_a, \dots, s_0 \leq \cdots \leq s_c$ and $\{v_i\}_{i\in\mathbb{Z}}$ are a basis of $V$. Later, we will also use an explicit basis of $V$ in terms of monomials.

Furthermore, there is an operator that acts diagonally in the chosen basis and behaves like a number operator. Define the \textbf{number operator} $\mathrm{N} \coloneqq -\frac{1}{2}\Jmode_{00}^z -i\charge{\widehat{\Emode}}$ for $\Mrep_\hrep^+$ and $\Mrep_{\hrep,\Lambda}^+$ respectively, where $\charge{\widehat{\Emode}}\in\mathbb{C}$ is determined by the action of $\widehat{\Emode}$ on $V$ (the Fock space is irreducible; therefore, by Dixmier's Lemma, it must act by a multiple of the identity).
\begin{lemma}
    The number operator satisfies the following properties:
\begin{enumerate}
\item $[\mathrm{N},\mathrm{X}_{kl}^\mp]=\pm \mathrm{X}_{kl}^\mp \quad\text{for } \mathrm{X} = \Jmode,\Kmode,\Pmode$ and $k,l\in\mathbb{Z}$ and zero otherwise
       \item  $\begin{aligned}\mathrm{N}\ket{j_{k_0 l_0} \cdots j_{k_a l_a}\ k_{m_0 n_0} \cdots k_{m_b n_b} \ p_{r_0 s_0} \cdots p_{r_c s_c}  ; v} = (a+b+c+3)&\vert j_{k_0 l_0} \cdots j_{k_a l_a}\ k_{m_0 n_0} \cdots k_{m_b n_b} \\&p_{r_0 s_0} \cdots p_{r_c s_c};v\rangle \quad \forall v \in V
           \end{aligned}$
\end{enumerate}
\end{lemma}
\begin{proof}
    1. This statement follows directly from Example~\ref{ex:MDTGhat}.
    2. Follows immediately from the first property and $\mathrm{N}\, V = 0$. The action of the generators on the vacuum is worked out explicitly in Section \ref{sec:action_on_Mrep}.
\end{proof}
As a consequence, $\Mrep_{\hrep,\Lambda}^+$ decomposes
\begin{equation*}
    \Mrep_{\hrep,\Lambda}^+ = \bigoplus_{k\in\mathbb{N}}(\Mrep_{\hrep,\Lambda}^+)_k\,,
\end{equation*}
where $(\Mrep_{\hrep,\Lambda}^+)_k$ is the eigenspace of $\mathrm{N}$ with eigenvalue $k\in\mathbb{N}$. In particular, $(\Mrep_{\hrep,\Lambda}^+)_0 = V$. Each eigenspace is infinite dimensional, since the number operator just counts the total degree of the monomial. Analogously, we have an $\mathbb{N}$-grading of $\Mrep_\hrep^+$. The operators turn into graded maps accordingly. The $(\pm)$-operators are maps of degree $\mp1$ respectively and the $z$-operators are maps of degree 0. The degree of a homogeneous element $p\in\Mrep_\hrep^+$ will be denoted $\vert p\vert$. In Section~\ref{sec:irred}, we will use the degree of the operators to infer irreducibility for a fixed $\hrep$. 

In the coming sections, we will treat the cases of zero and non-zero cosmological constants separately, as the explicit constructions are different.

\subsection{Modules of \texorpdfstring{$\Ghat(\su)$}{Ghat}}\label{subsec:repsGhat}
In this section, we work out a family of explicit bosonic Fock-type modules $\Mrep_\hrep^+$ of the {free corner algebra} $\mathcal{U}(\Ghat(\su))$ of non-abelian $BF$ theory on the torus in terms of differential operators on a space of polynomials. We show that the modules in the family are irreducible. Finally, we explore the consequences of imposing the constraints. To ensure that the action is well defined, we need to choose certain polarizations of $V$. It turns out that the constraints have a non-zero action in the module for all possible choices of charges. We thus conclude that the family of modules induce trivial modules for the physical corner algebra $A_\torus$. 

\subsubsection{Action of the Generators on \texorpdfstring{$\Mrep_{\hrep}^+$}{text}}\label{sec:action_on_Mrep}
It is reasonable to expect that the generators can be realized as differential operators because they are linear maps on a space of polynomials. One can obtain the individual contributions by acting on monomials. 
To this end, we define the $\mathfrak{h}$-action $\hrep$ on $V$ explicitly using the notation from the abelian case and the identification from Lemma \ref{lem:embedding_abelian}. Recall the ladder operators from the proof of Theorem~\ref{thm:abisheis}. The representation of the ladder operators on the bosonic Fock space $V=\mathbb{C}[\{v_{kl}\}_{(k,l)\in\mathbb{Z}^2\setminus \{(0,0)\}}]$ is defined by:
\begin{alignat*}{3}
    &\cmode^\dag_{kl} &&\longmapsto v_{kl} \quad && k\neq0,l\neq0, \\
    &\cmode_{kl} &&\longmapsto i\charge{\Ccharge}\pdv{}{v_{kl}} \quad&&k\neq0,l\neq0, \\
    &\amode^\dag_{l} &&\longmapsto v_{0l} \quad && l\neq0, \\
    &\amode_{l} &&\longmapsto i\charge{\Ccharge}\pdv{}{v_{0l}} \quad&&l\neq0, \\
     &\bmode^\dag_{k} &&\longmapsto v_{k0} \quad && k\neq0, \\
    &\bmode_{k} &&\longmapsto i\charge{\Ccharge}\pdv{}{v_{k0}} \quad&&k\neq0,\\
    &\Ccharge &&\longmapsto \charge{\Ccharge},\\
    &\widehat{\Phmode}_l &&\longmapsto \charge{\widehat{\Phmode}_l} \quad&&l\in\mathbb{Z},\\
    &\widehat{\Thmode}_k &&\longmapsto \charge{\widehat{\Thmode}_k} \quad&&k\in\mathbb{Z},\\
    &\widehat{\Fmode}^-_{kl} &&\longmapsto \charge{\widehat{\Fmode}^-_{kl}} \quad&&k,l\neq0,\\
    &\widehat{\Emode} &&\longmapsto \charge{\widehat{\Emode}}\,,
\end{alignat*}
and for convenience $ v_{00} \equiv 1, \pdv{}{v_{00}}\equiv 1$. By the same argument as in defining the number operator, the central elements are proportional to the identity and thus determined by a complex number. Note that we could choose a different assignment of multiplication and differentiation. We call this assignment a \textbf{choice of polarization in $V$}. It turns out that to impose the constraints, we have to change to a different polarization.

We also sometimes denote the states by their respective polynomials.
\begin{align*}
    1 \equiv\ket{0}, \\
    j_{k_0 l_0} \cdots j_{k_a l_a}&\ k_{m_0 n_0} \cdots k_{m_b n_b} \ p_{r_0 s_0} \cdots p_{r_c s_c}\ v_{p_0 q_0}\cdots v_{p_d q_d}\equiv\\& \vert j_{k_0 l_0} \cdots j_{k_a l_a}\ k_{m_0 n_0} \cdots k_{m_b n_b}
     p_{r_0 s_0} \cdots p_{r_c s_c}  ; v_{p_0 q_0}\cdots v_{p_d q_d}\rangle\,,
\end{align*}
where we now also express the basis of $V$ using the ladder operators defined above acting on the vacuum.

By systematically acting on monomials, we can deduce the necessary coefficients in front of the differential operators. We just need to commute the operators until they hit the vacuum sector and then apply the explicit formulas from above. 

The central charge:
\begin{align*}
    \Ccharge \ket{v} &= 1 \otimes_{\mathcal{U}(\mathfrak{h}\oplus\mathfrak{n^\pm})}\Ccharge\ket{v} \\
    &\coloneqq \charge{\Ccharge} \ket{v}.
\end{align*}
Notice that the action of $\Ccharge$ is proportional to the identity on all of $\Mrep_\hrep^+$ and is determined by its action on the vacuum sector $V$. From now on, we set $\charge{\Ccharge}$=1 and suppress the tensor product notation.

The creation operators:
\begin{align*}
    &\Jmode_{kl}^-\ket{v} =\ket{j_{kl};v}, \\
    &\Kmode_{kl}^-\ket{v} =\ket{k_{kl};v}, \\
    &\Pmode_{kl}^-\ket{v}=\ket{p_{kl};v}.
\end{align*}

The annihilation operators:
\begin{align*}
    &\Jmode_{kl}^+\ket{v} =0, \\
    &\Kmode_{kl}^+\ket{v} =0, \\
    &\Pmode_{kl}^+\ket{v}= 0.
\end{align*}

The $z$-operators:
\begin{align*}
    \Jmode_{kl}^z\ket{v} =-2i\Emode_{kl}\ket{v} 
    &= \begin{cases}
       -2i\cmode_{-k-l}\ket{v} & \quad\text{for } k \neq 0, l \neq 0 \\
         -2i\bmode_{-k}\ket{v} & \quad\text{for } k \neq 0, l = 0 \\
         -2i\amode_{-l}\ket{v} & \quad\text{for } k = 0, l\neq 0 \\
         -2i\widehat{\Emode}\ket{v} & \quad\text{for } k = l=0
  \end{cases}\\
  &= \begin{cases}
         2\pdv{}{v_{-k-l}}\ket{v} & \quad\text{for } k \neq 0, l \neq 0 \\
         2\pdv{}{v_{-k0}}\ket{v} & \quad\text{for } k \neq 0, l = 0 \\
         2\pdv{}{v_{0-l}}\ket{v} &\quad \text{for } k = 0, l\neq 0 \\
         -2i\widehat{\Emode}\ket{v} & \quad\text{for } k = l=0
  \end{cases}\\
  &= (2 + (-2 -2i\charge{\widehat{\Emode}}) \delta_{k,0} \delta_{l,0}) \pdv{}{v_{-k-l}}\ket{v}, \\
    \Kmode_{kl}^z\ket{v} =-2i\Phmode_{kl}\ket{v}
    &= \begin{cases}
        -2ik\left(\cmode^\dag_{kl}-\Fmode^-_{kl}\right)\ket{v} & \quad\text{for } k \neq 0, l \neq 0 \\
        -2ik\bmode_k^\dag\ket{v} & \quad\text{for } k \neq 0, l = 0 \\
        -2i\widehat{\Phmode}_l\ket{v} & \quad\text{for } k = 0, l\in \mathbb{Z}
  \end{cases}\\
  &= \begin{cases}
        -2ik\ket{v_{kl}v}+2ik\charge{\Fmode^-_{kl}}\ket{v} & \quad\text{for } k \neq 0, l \neq 0 \\
        -2ik\ket{v_{k0}v} & \quad\text{for } k \neq 0, l = 0 \\
        -2i\charge{\widehat{\Phmode}_l}\ket{v} & \quad\text{for } k = 0, l\in \mathbb{Z} 
  \end{cases}\\
    &= (-2ik v_{kl} - 2i\charge{\widehat{\Phmode}_l} \delta_{k,0}+2ik\charge{\Fmode_{kl}^-}(1-\delta_{l,0}))\ket{v}, \\
    \Pmode_{kl}^z\ket{v} =-2i\Thmode_{kl}\ket{v} 
    &= \begin{cases}
        2il\left(\cmode^\dag_{kl}+\Fmode^-_{kl}\right)\ket{v} & \quad\text{for } k \neq 0, l \neq 0 \\
        2il\amode_l^\dag\ket{v} & \quad\text{for } k = 0, l \neq 0 \\
        -2i\widehat{\Thmode}_k\ket{v} & \quad\text{for } k\in \mathbb{Z}, l=0
  \end{cases}\\
  &= \begin{cases}
        2il\ket{v_{kl}v}+2il\charge{\Fmode^-_{kl}}\ket{v}  & \quad\text{for } k \neq 0, l \neq 0 \\
        2il\ket{v_{0l}v} & \quad\text{for } k = 0, l \neq 0 \\
        -2i\charge{\widehat{\Thmode}_k}\ket{v} & \quad\text{for } k\in \mathbb{Z}, l= 0
  \end{cases}\\
   &= (2il v_{kl} - 2i\charge{\widehat{\Thmode}_k} \delta_{l,0}+2il\charge{\Fmode_{kl}^-}(1-\delta_{k,0}))\ket{v}.
\end{align*}
One can continue with determining the action on ``first excited states'', i.e.\ monomials of degree 1, and so on with higher degree monomials. The following theorem gives an explicit formula for the induced module of $\Ghat(\su)$. 

To make the formulas more compact, we introduce the following notation:
\begin{alignat*}{2}
    &\Diag_{kl}^{\widehat{\Emode}} &&\coloneqq 2 + (-2-2i \charge{\widehat{\Emode}}) \delta_{k,0} \delta_{l,0}\,,\\
    &\Diag_{kl}^{\widehat{\Phmode}}&&\coloneqq   - 2i\charge{\widehat{\Phmode}_l} \delta_{k,0}+2ik\charge{\Fmode_{kl}^-}(1-\delta_{l,0})\,, \\
    &\Diag_{kl}^{\widehat{\Thmode}}&&\coloneqq  - 2i\charge{\widehat{\Thmode}_k} \delta_{l,0}+2il\charge{\Fmode_{kl}^-}(1-\delta_{k,0})\,,
\end{alignat*}
and
\begin{alignat*}{2}
     &\Ewaki{\mathrm{X},\mathrm{Y}}(k,l) &&\coloneqq \sum_{m,n} (-2x_{(k+m)(l+n)}) \pdv{}{y_{mn}}
    \quad \text{for } \mathrm{X},\mathrm{Y} \in \{\Jmode,\Kmode,\Pmode\}\,,\\
    &\Ewaki{\mathrm{X},\mathrm{Y}\mathrm{Z}}(k,l) &&\coloneqq \sum_{m,n}\sum_{r,s} (-2x_{(k+m+r)(l+n+s)}) \pdv{}{y_{mn}} \pdv{}{z_{rs}}
    \quad \text{for } \mathrm{X},\mathrm{Y},\mathrm{Z} \in \{\Jmode,\Kmode,\Pmode\} \,.
\end{alignat*}
\begin{theorem}\label{thm:freefieldtor}
    The assignment
    \begin{align*}
        \Ccharge \longrightarrow& 1\\
        \Jmode_{kl}^- \longrightarrow& j_{kl}\\
        \Kmode_{kl}^- \longrightarrow& k_{kl}\\
        \Pmode_{kl}^- \longrightarrow& p_{kl}\\
        \Jmode_{kl}^z \longrightarrow&
         \Diag_{kl}^{\widehat{\Emode}}\pdv{}{v_{-k-l}} +
         \Ewaki{\Jmode,\Jmode}(k,l)+ \Ewaki{\Kmode,\Kmode}(k,l)+ \Ewaki{\Pmode,\Pmode}(k,l)\\
        \Kmode_{kl}^z \longrightarrow& -2ik v_{kl} - \Diag_{kl}^{\widehat{\Phmode}}
        + \Ewaki{\Kmode,\Jmode}(k,l)\\
        \Pmode_{kl}^z \longrightarrow& 2il v_{kl} - \Diag_{kl}^{\widehat{\Thmode}}
        + \Ewaki{\Pmode,\Jmode}(k,l)\\
        \Jmode_{kl}^+ \longrightarrow&  \sum_{m,n} (\Diag_{(k+m)(l+n)}^{\widehat{\Emode}}\pdv{}{v_{-(k+m)-(l+n)}} \pdv{}{j_{mn}} \\
        &+ \sum_{m,n} (-2i(k+m)v_{(k+m)(l+n)}+ \Diag_{(k+m)(l+n)}^{\widehat{\Phmode}} -2im \delta_{k+m,0} \delta_{l+n,0}) \pdv{}{k_{mn}} \\
        &+ \sum_{m,n} (2i(l+n)v_{(k+m)(l+n)} +\Diag_{(k+m)(l+n)}^{\widehat{\Thmode}} +2in \delta_{k+m,0} \delta_{l+n,0}) \pdv{}{p_{mn}}\\&+ \frac{1}{2}\Ewaki{\Jmode,\Jmode\Jmode}(k,l)+ \Ewaki{\Kmode,\Kmode\Jmode}(k,l)+ \Ewaki{\Pmode,\Pmode\Jmode}(k,l)\\
        \Kmode_{kl}^+ \longrightarrow& \sum_{m,n} (-2i(k+m)v_{(k+m)(l+n)} +\Diag_{(k+m)(l+n)}^{\widehat{\Phmode}} -2im \delta_{k+m,0} \delta_{l+n,0}) \pdv{}{j_{mn}} + \frac{1}{2}\Ewaki{\Kmode,\Jmode\Jmode}(k,l)\\
        \Pmode_{kl}^+ \longrightarrow& \sum_{m,n} (2i(l+n)v_{(k+m)(l+n)} +\Diag_{(k+m)(l+n)}^{\widehat{\Thmode}} +2in \delta_{k+m,0} \delta_{l+n,0} ) \pdv{}{j_{mn}} + \frac{1}{2}\Ewaki{\Pmode,\Jmode\Jmode}(k,l)
    \end{align*}
    constitutes a family of representations of the free corner algebra $\Ghat(\su)$ on $\Mrep_\hrep^+$, i.e.\ the space of polynomials $\mathbb{C}[\{j_{kl},k_{kl},p_{kl},v_{mn}\}_{k,l,m,n\in\mathbb{Z},(m,n)\neq(0,0)}]$, parametrized by complex numbers $\charge{\widehat{E}},\charge{\Fmode_{kl}^-},\charge{\widehat{\Phmode}_n},$ $\charge{\widehat{\Thmode}_m}\in \mathbb{C}$ for $k,l\in \mathbb{Z}\setminus{\{0\}}$ and $ m,n \in \mathbb{Z}$.
\end{theorem}

\begin{proof}
    The proof {involves a direct verification of the commutation relations and can be found in \cite[Appendix B]{Leupp2025}}, modulo setting most of the charges to zero for convenience and in light of the constraints.
\end{proof}
\begin{remark}
    These modules resemble the free field realizations of affine algebras
    and also the polynomial representations of \cite{MOROZOV2022137193}. One can also compare directly with Theorem \ref{thm:waki_free_field_rep_of_sl2}.\footnote{The notation above is inspired by that paper.} By applying the construction at hand to the usual triangular decomposition of $\widehat{\mathfrak{sl}}(2)$, one obtains a very similar representation to Wakimoto's, but without the more involved polarization and thus without normal ordering. As a consequence, the central charge necessarily acts trivially. The representation can be recovered exactly by changing the polarization and choosing a normal ordering. 
    In Section \ref{sec:regularizeconstr} when trying to impose the constraints, we are also led to change the polarization; however, no normal ordering will be required. It would be interesting to explore the precise connection between the module construction herein and the polynomial representations or free field realizations.
\end{remark}

\subsubsection{Irreducibility of \texorpdfstring{$\Mrep_{\hrep}^+$}{text}} \label{sec:irred}
We want to understand whether the representations of $\Ghat(\su)$ contain non-trivial subrepresentations because they are connected to imposing the constraints, as will be discussed in Section \ref{sec:regularizeconstr}. 
\begin{proposition}\label{clm:H_is_irred}
    The modules in the family of representations $\Mrep_\hrep^+$ of $\Ghat(\su)$ from Theorem~\ref{thm:freefieldtor} are irreducible.
\end{proposition}
\begin{proof}
    This is the proof of Proposition 3.21 in \cite{Leupp2025}.
\end{proof}
It would be interesting to investigate when two representations with different choices of charges $\charge{\widehat{\Thmode}_k},\charge{\widehat{\Phmode}_l}$ and $ \charge{\Fmode_{kl}^-}$ (or more generally, different choices of $(V,\hrep)$) are equivalent.

\subsubsection{Action of the Constraint Ideal \texorpdfstring{$\mathcal{I}_{F_A}$}{text} on \texorpdfstring{$\Mrep_{\hrep}^+$}{t}}\label{sec:regularizeconstr}
As discussed in Section~\ref{sec:corner algebra of BF}, we would like to obtain representations of the {free corner algebra} that descend to the physical corner algebra $A_\torus \cong \mathcal{U}(\Ghat(\su)) \vert_{F_A = 0}$, i.e.\ representations where the ideal of constraints acts trivially or generates a proper submodule. Ideally, choosing appropriate values for the undetermined charges suffices to ensure that the constraints act by zero or generate a proper submodule.

To examine the action of the constraints, it is convenient to transform them into the ladder basis as well.
\begin{alignat*}{2}
    &\widehat{f}_{rs}^\pm &&= \pm\widehat{f}_{1 rs}-i\widehat{f}_{2rs}
    = -is\Kmode_{rs}^\pm-ir\Pmode_{rs}^\pm + \frac{1}{2}\sum_{m,n}\left( \mp\Pmode_{(r+m)(s+n)}^\pm\Kmode_{-m-n}^z\pm\Pmode_{(r+m)(s+n)}^z\Kmode_{-m-n}^\pm \right), \\
    &\widehat{f}_{rs}^z &&=-2i\widehat{f}_{3rs}=-is\Kmode_{rs}^z-ir\Pmode_{rs}^z +\sum_{m,n}\left(\Pmode_{(r+m)(s+n)}^+\Kmode_{-m-n}^- -\Pmode_{(r+m)(s+n)}^-\Kmode_{-m-n}^+ \right).
\end{alignat*}
The new bracket relations are given by:
\begin{alignat*}{2}
    &[\widehat{f}_{rs}^z,\Jmode_{mn}^\pm] &&= \pm2 \widehat{f}_{(r+m)(s+n)}^\pm\,,\\
    &[\widehat{f}_{rs}^\pm,\Jmode_{mn}^z] &&= \mp 2\widehat{f}_{rs}^\pm\,,\\
    &[\widehat{f}_{rs}^\pm,\Jmode_{mn}^\mp] &&= \pm\widehat{f}_{(r+m)(s+n)}^z\, ,
\end{alignat*}
and zero otherwise. Now, we would like to examine the action of the constraints on $\Mrep_\hrep^+$ from Theorem~\ref{thm:freefieldtor}.
However, there is a problem with the representation, because the action of the $\widehat{f}_{rs}^-$-constraints are not well defined. To see this, consider 
$\widehat{f}^-_{kl}$ whose action contains terms proportional to
\begin{align}
    \sum_{m,n}\left( m p_{(r+m)(s+n)}v_{-m-n}- lv_{(r+m)(s+n)}k_{-m-n} \right). \label{eq:divergentconstr}
\end{align}

The action leads to an infinite sum over the linearly independent vectors. 
Fortunately, there is some freedom in choosing the polarization of $V$. Instead of the representation $(\hrep,V)$, we can choose another representation $(\bar{\hrep},V)$:
\begin{equation}\label{eq:new_v_polarization}
\begin{aligned}
    &\cmode^\dag_{kl} &&\longmapsto -\pdv{}{v_{kl}} \quad && k\neq0,l\neq0, \\
    &\cmode_{kl} &&\longmapsto iv_{kl} \quad&&k\neq0,l\neq0,  \\
    &\amode^\dag_{l} &&\longmapsto -\pdv{}{v_{0l}} \quad && l\neq0, \\
    &\amode_{l} &&\longmapsto iv_{0l} \quad&&l\neq0,  \\
     &\bmode^\dag_{l} &&\longmapsto -\pdv{}{v_{k0}} \quad && k\neq0,  \\
    &\bmode_{l} &&\longmapsto iv_{k0}  \quad&&k\neq0, 
\end{aligned}
\end{equation}
and $ v_{00} \equiv 1, \pdv{}{v_{00}}\equiv1$.
The terms in the infinite sum (\ref{eq:divergentconstr}) now contain derivatives instead of multiplication with a $v$ variable and act in a well-defined manner on any state in $\Mrep_{\bar{\hrep}}^+$. 
\begin{remark}
    We might expect a similar freedom in the other variables $j,k$ and $p$. However, it is much more difficult to change the polarization of these variables, because they appear in infinite sums and can become ill-defined if one does not take proper care of ordering them suitably.
\end{remark}
\begin{claim}\label{cl:irred}
    We conjecture that the family of representations $\Mrep_{\bar{\hrep}}^+$ is irreducible.
\end{claim}
We expect this claim because the representation is very similar on a structural level. However, we have not proven it. 
Since the family $\Mrep_{\bar{\hrep}}^+$ is irreducible, we must have that $\mathcal{I}_{F_A}$ acts by zero on the nose. 

Now, let us examine the set of conditions that the constraints impose on the undetermined charges. It should be sufficient to check the action of the constraints on the vacuum sector $V$. This follows directly from the ideal property and that we can systematically move the constraints past the monomials until they hit the vacuum states. 
We then have the following proposition.
\begin{proposition} The action of the constraint ideal $\mathcal{I}_{F_A}$ on the vacuum sector $V$ of $\Mrep_{\bar{\hrep}}^+$ imposes the following restrictions on the undetermined charges:
    \begin{alignat*}{2}
        &\widehat{f}^+_{rs}V\overset{!}{=} 0  \quad\forall r,s\in\mathbb{Z} \implies &&\text{no restriction},  \\
        &\widehat{f}^z_{rs}V \overset{!}{=} 0  \quad\forall r,s\in\mathbb{Z}\implies &&\charge{\Fmode_{kl}^-}=\charge{\widehat{\Phmode}_l}=\charge{\widehat{\Thmode}_k} = 0 \quad\text{for } k \neq 0, l \neq 0, \\
        &\widehat{f}^-_{rs}V \overset{!}{=} 0 \quad\forall r,s \in\mathbb{Z}\hphantom{\implies}&&\text{not possible}
        \,.
    \end{alignat*}
\end{proposition}
\begin{proof}
     
     \textbf{$(+)$-constraints:} The $(+)$-constraints trivially vanish when acting on the vacuum sector, as there is always a derivative with respect to the $j$, $k$ or $p$ variables, because the $\Kmode$- and $\Pmode$-operators commute. Therefore, they impose no restriction on the charges.

     \textbf{$(z)$-constraints:} The $(z)$-operators act as follows:
     \begin{align*}
        \widehat{f}^z_{00}\vert_V&=0\,,\\
        \widehat{f}^z_{r0}\vert_V&=-2r\charge{\widehat{\Thmode}_r}\,, \\
        \widehat{f}^z_{0s}\vert_V&=-2s\charge{\widehat{\Phmode}_s} \,.
     \end{align*}
     For any physically admissible representations, these terms must vanish. Thus, we have that $\charge{\widehat{\Phmode}_s}=\charge{\widehat{\Thmode}_r} = 0 \text{ for } r \neq 0, s \neq 0\,$.
     Assuming these conditions have been imposed, the leftover $(z)$-constraints act as:
    \begin{align*}
        \widehat{f}^z_{rs}\vert_V&=2(r+s)\charge{\Fmode_{rs}^-}\,.
    \end{align*}
    Therefore, the only possibility is to set these charges to zero as well.

    \textbf{$(-)$-constraints:} The $(-)$-operators act as follows:
    \begin{align*}
        \widehat{f}_{00}^-\vert_V&=\frac{1}{2}\sum_{m,n}\left(-2imp_{mn}\pdv{}{v_{-m-n}}+2ink_{-m-n}\pdv{}{v_{mn}}\right)
        =-\sum_{m,n}\left(imp_{mn}+ink_{mn}\right)\pdv{}{v_{-m-n}}\,.
    \end{align*}
    Since the representation is irreducible, it is not possible to restrict to a non-trivial subrepresentation by a suitable choice of charges where this operator could vanish. 
\end{proof}
The $(z)$-part of the restriction, coming from the linear terms in the constraints, nicely coincides with the constraints from the abelian case. 
However, the non-linear terms act non-zero and thus:
\begin{corollary}
    The family $\Mrep_{\bar{\hrep}}^+$ of $\mathcal{U}(\Ghat(\su))$-modules induces a trivial $A_\torus$-module for any choice of charges. 
\end{corollary}

\subsubsection{{Obstruction to} Imposing the Constraints}
To understand why the  modules of the free corner algebra descend to trivial modules of the physical corner algebra, let us recall that:
the classical constraints are elements of the vector space of symmetric multilinear forms on $\Ghat^*_\Lambda(\su)$ (see Section \ref{sec:constr}). As such, the constraints are not elements of the symmetric tensors over $\Ghat_\Lambda(\su)$. However, we still treated them as such, leading to infinite sums and potentially too many constraint conditions.
It is not yet clear how this can be remedied. 

\subsection{Modules of \texorpdfstring{$\Ghat_\Lambda(\su)$}{Ghat}}\label{sec:ghat_lambda_modules}
In this section, we work out a family of explicit bosonic Fock-type modules of the free corner algebra $\mathcal{U}(\Ghat_\Lambda(\su))$ for non-abelian $BF$ theory with non-zero $\Lambda$ in terms of differential operators on a space of polynomials. The analysis is very similar to the previous section with the additional difficulty of the non-trivial coupling of $\Kmode$- and $\Pmode$-generators.
In this case, the constraints act in an ill-defined way, even in the new polarization, but the representation can be altered accordingly. The final implications of the constraints have not been worked out, but they are expected to behave similarly to the case of $\Lambda = 0$, which means that they do not descend to the physical corner algebra $A_\torus$.

\subsubsection{Action of the Generators on \texorpdfstring{$\Mrep_{\hrep,\Lambda}^+$}{text}}
We proceed analogously to Section~\ref{subsec:repsGhat} to derive an explicit representation. For non-zero cosmological constant, the representation of the ladder operators is defined by: 
\begin{alignat*}{3}
    &\wmode^\dag_{kl} &&\longmapsto x_{kl} \quad && k\neq0,l\neq0, \\
    &\wmode_{kl} &&\longmapsto i\pdv{}{x_{kl}} \quad&&k\neq0,l\neq0, \\
    &\umode^\dag_{l} &&\longmapsto x_{0l} \quad && l\neq0, \\
    &\umode_{l} &&\longmapsto i\pdv{}{x_{0l}}\quad&&l\neq0, \\
     &\vmode^\dag_{k} &&\longmapsto x_{k0}\quad && k\neq0, \\
    &\vmode_{k} &&\longmapsto i\pdv{}{x_{k0}} \quad&&k\neq0,\\
    &\bar{\Phmode} &&\longmapsto x_{00},\\
    &\bar{\Thmode} &&\longmapsto i\pdv{}{x_{00}},\\
    &\Ccharge &&\longmapsto 1, \\
    &\widehat{\umode}_l &&\longmapsto 0 \quad&&l\neq 0,\\
    &\widehat{\vmode}_k &&\longmapsto 0 \quad&&k\neq0,\\
    &\widehat{\wmode}_{kl} &&\longmapsto 0 \quad&&k,l\neq0,\\
    &\widehat{\Emode} &&\longmapsto 0\,.
\end{alignat*}
In contrast to the $\Lambda = 0 $ case, the expressions $x_{00}, \pdv{}{x_{00}}$ are already assigned and are not set to 1. Additionally, we have set the central charge $\Ccharge$ to 1 and, in anticipation of the constraints, the charges $\charge{\widehat{\umode}_l},\charge{\widehat{\vmode}_k},\charge{\widehat{\wmode}_{kl}}$ and $\charge{\widehat{\Emode}}$ to 0.
We can systematically determine the differential operators associated with the generators by acting on monomials (cf.\ \cite[Section 3.5.1]{Leupp2025}).

The following theorem gives an explicit formula for the representation of $\Ghat_\Lambda(\su)$.
\begin{theorem}
    The assignment
    \begin{align*}
        \Ccharge \longrightarrow& 1\\
        \Jmode_{kl}^- \longrightarrow& j_{kl}\\
        \Kmode_{kl}^- \longrightarrow& k_{kl}\\
        \Pmode_{kl}^- \longrightarrow& p_{kl}\\
        \Jmode_{kl}^z \longrightarrow&
        (2  -2\delta_{k,0} \delta_{l,0}) \pdv{}{x_{-k-l}} 
        + \Ewaki{\Jmode,\Jmode}(k,l)+ \Ewaki{\Kmode,\Kmode}(k,l)+ \Ewaki{\Pmode,\Pmode}(k,l)\\
        \Kmode_{kl}^z \longrightarrow& -2ikx_{kl}+2\Lambda x_{00}\delta_{k,0}\delta_{l,0}+\Lambda \ADiag_{kl}^\Kmode\pdv{}{x_{-k-l}}+\Ewaki{\Kmode,\Jmode}(k,l)\\
        \Pmode_{kl}^z \longrightarrow& 2ilx_{kl}+2 \pdv{}{x_{00}}\delta_{k,0}\delta_{l,0}+\Lambda \BDiag_{kl}^\Pmode\pdv{}{x_{-k-l}}+\Ewaki{\Pmode,\Jmode}(k,l)\\
        \Jmode_{kl}^+ \longrightarrow&
        \sum_{m,n}(2  -2\delta_{k+m,0} \delta_{l+n,0}) \pdv{}{x_{(-k-m)(-l-n)}} \pdv{}{j_{mn}}\\
        &+\sum_{m,n }\biggl(-2i(k+m)x_{(k+m)(l+n)}+2\Lambda x_{00}\delta_{k+m,0}\delta_{l+n,0}\\& \quad +\Lambda \ADiag_{(k+m)(l+n)}^\Kmode\pdv{}{x_{(-k-m)(-l-n)}} -2im\delta_{k+m,0}\delta_{l+n,0} \biggr)\pdv{}{k_{mn}}\\
        &+\sum_{m,n }\biggl(2i(l+n)x_{(k+m)(l+n)}+2 \pdv{}{x_{00}}\delta_{k+m,0}\delta_{l+n,0}\\& \quad +\Lambda \BDiag_{(k+m)(l+n)}^\Pmode\pdv{}{x_{(-k-m)(-l-n)}} +2in\delta_{k+m,0}\delta_{l+n,0} \biggr)\pdv{}{p_{mn}}\\
        &+ \frac{1}{2}\Ewaki{\Jmode,\Jmode\Jmode}(k,l)+ \Ewaki{\Kmode,\Kmode\Jmode}(k,l)+ \Ewaki{\Pmode,\Pmode\Jmode}(k,l)\\
        \Kmode_{kl}^+ \longrightarrow&-2\Lambda\pdv{}{p_{-k-l}}+
        \sum_{m,n }\biggl(-2i(k+m)x_{(k+m)(l+n)}+2\Lambda x_{00}\delta_{k+m,0}\delta_{l+n,0}\\& \quad +\Lambda \ADiag_{(k+m)(l+n)}^\Kmode\pdv{}{x_{(-k-m)(-l-n)}} -2im\delta_{k+m,0}\delta_{l+n,0} \biggr)\pdv{}{j_{mn}}\\&+\frac{1}{2}\Ewaki{\Kmode,\Jmode\Jmode}(k,l)
        \\
        \Pmode_{kl}^+ \longrightarrow& 2\Lambda\pdv{}{k_{-k-l}}+
        \sum_{m,n }\biggl(2i(l+n)x_{(k+m)(l+n)}+2 \pdv{}{x_{00}}\delta_{k+m,0}\delta_{l+n,0}\\& \quad +\Lambda \BDiag_{(k+m)(l+n)}^\Pmode\pdv{}{x_{(-k-m)(-l-n)}} +2in\delta_{k+m,0}\delta_{l+n,0} \biggr)\pdv{}{j_{mn}}\\&+\frac{1}{2}\Ewaki{\Pmode,\Jmode\Jmode}(k,l)\, ,
    \end{align*}
\end{theorem}
where 
\begin{align*}
    \ADiag_{kl}^\Kmode &\coloneqq -\frac{i}{l}(1-\delta_{k,0})(1-\delta_{l,0})-\frac{2i}{l}\delta_{k,0}(1-\delta_{l,0}), \quad \forall k,l\in\mathbb{Z},\\
    \BDiag_{kl}^\Pmode &\coloneqq -\frac{i}{k}(1-\delta_{k,0})(1-\delta_{l,0})-\frac{2i}{k}(1-\delta_{k,0})\delta_{l,0}, \quad \forall k,l\in\mathbb{Z},
\end{align*}
constitutes a representation of {the free corner algebra} $\Ghat_\Lambda(\su)$ on $\Mrep_{\hrep,\Lambda}^+$, i.e.\ the space of polynomials $\mathbb{C}[\{j_{kl},k_{kl},p_{kl},x_{mn}\}_{k,l,m,n\in\mathbb{Z}}]$.
\begin{proof}
    The proof is analogous to the proof of Theorem~\ref{thm:freefieldtor}.
\end{proof}
\begin{claim}
    We conjecture that the representation is irreducible.
\end{claim}
This can be expected since the overall structure of the shift operators is still the same and so the main ideas of Proposition \ref{clm:H_is_irred} should be adaptable.

\subsubsection{Action of the Constraint Ideal \texorpdfstring{$\mathcal{I}_{F_A +\Lambda B}$}{text} on \texorpdfstring{$\Mrep_{\hrep,\Lambda}^+$}{text}}
In Section~\ref{sec:regularizeconstr}, we described a way to ``regularize" the constraints by passing to a different polarization of the underlying $\mathfrak{h}$-representation $(V,\hrep)$. However, the new polarization is not sufficient when $\Lambda\neq0$. There is an additional ordering ambiguity in the $\widehat{f}_{00}^z$-constraint, because the $\Pmode^+$-operator does not commute with the $\Kmode^+$-operator anymore. Potentially, one could define a normal ordering convention to take care of this issue. Another possibility, the one we will pursue, is to change the polarization of the $k$ variables analogously to that of the $v$ variables in Equation~\eqref{eq:new_v_polarization}. This change is not straightforward, as the $k$ variables appear as multiplicative factors and derivatives in $\Ewaki{\Kmode,\Kmode}(k,l)$, for example, and the change thus leads to new divergences. However, one can get around this issue by replacing the problematic operators with non-divergent ones that have the same algebraic properties. The essential algebraic properties can be inferred from the proof of Theorem~\ref{thm:freefieldtor} and mainly involve the degree shifting of the polynomial operators. The resulting representation behaves quite differently than the original family (see Remark \ref{rmk:modified_rep}).

To summarize, the assignment
    \begin{align*}
        \Ccharge \longrightarrow& 1\\
        \Jmode_{kl}^- \longrightarrow& j_{kl}\\
        \Kmode_{kl}^- \longrightarrow&- \pdv{}{k_{-k-l}}\\
        \Pmode_{kl}^- \longrightarrow& p_{kl}\\
        \Jmode_{kl}^z \longrightarrow&
        (2  -2\delta_{k,0} \delta_{l,0}) {x_{-k-l}} 
        + \Ewaki{\Jmode,\Jmode}(k,l)- \Ewaki{\Kmode,\Kmode}(k,l)+ \Ewaki{\Pmode,\Pmode}(k,l)\\
        \Kmode_{kl}^z \longrightarrow& 2ik\pdv{}{x_{kl}}+2\Lambda x_{00}\delta_{k,0}\delta_{l,0}+\Lambda \ADiag_{kl}^\Kmode{x_{-k-l}}+\sum_{m,n} 2\pdv{}{k_{(-k-m)(-l-n)}}\pdv{}{j_{mn}} \\
        \Pmode_{kl}^z \longrightarrow& -2il\pdv{}{x_{kl}}+2 \pdv{}{x_{00}}\delta_{k,0}\delta_{l,0}+\Lambda \BDiag_{kl}^\Pmode{x_{-k-l}}+\Ewaki{\Pmode,\Jmode}(k,l)\\
        \Jmode_{kl}^+ \longrightarrow&
        \sum_{m,n}(2  -2\delta_{k+m,0} \delta_{l+n,0}) {x_{(-k-m)(-l-n)}} \pdv{}{j_{mn}}\\
        &+\sum_{m,n }\biggl(2i(k+m)\pdv{}{x_{(k+m)(l+n)}}+2\Lambda x_{00}\delta_{k+m,0}\delta_{l+n,0}\\& \quad +\Lambda \ADiag_{(k+m)(l+n)}^\Kmode{x_{(-k-m)(-l-n)}} -2im\delta_{k+m,0}\delta_{l+n,0} \biggr)k_{-m-n}\\
        &+\sum_{m,n }\biggl(-2i(l+n)\pdv{}{x_{(k+m)(l+n)}}+2 \pdv{}{x_{00}}\delta_{k+m,0}\delta_{l+n,0}\\& \quad +\Lambda \BDiag_{(k+m)(l+n)}^\Pmode{x_{(-k-m)(-l-n)}} +2in\delta_{k+m,0}\delta_{l+n,0} \biggr)\pdv{}{p_{mn}}\\
        &+ \frac{1}{2}\Ewaki{\Jmode,\Jmode\Jmode}(k,l)- \Ewaki{\Kmode,\Kmode\Jmode}(k,l)+ \Ewaki{\Pmode,\Pmode\Jmode}(k,l)\\
        \Kmode_{kl}^+ \longrightarrow&-2\Lambda\pdv{}{p_{-k-l}}+
        \sum_{m,n }\biggl(2i(k+m)\pdv{}{x_{(k+m)(l+n)}}+2\Lambda x_{00}\delta_{k+m,0}\delta_{l+n,0}\\& \quad +\Lambda \ADiag_{(k+m)(l+n)}^\Kmode{x_{(-k-m)(-l-n)}} -2im\delta_{k+m,0}\delta_{l+n,0} \biggr)\pdv{}{j_{mn}}\\&+\frac{1}{2}\sum_{m,n}\sum_{r,s}2\pdv{}{k_{(-k-r-m)(-l-s-n)}}\pdv{}{j_{rs}}\pdv{}{j_{mn}}
        \\
        \Pmode_{kl}^+ \longrightarrow& 2\Lambda{k_{kl}}+
        \sum_{m,n }\biggl(-2i(l+n)\pdv{}{x_{(k+m)(l+n)}}+2 \pdv{}{x_{00}}\delta_{k+m,0}\delta_{l+n,0}\\& \quad +\Lambda \BDiag_{(k+m)(l+n)}^\Pmode{x_{(-k-m)(-l-n)}} +2in\delta_{k+m,0}\delta_{l+n,0} \biggr)\pdv{}{j_{mn}}\\&+\frac{1}{2}\Ewaki{\Pmode,\Jmode\Jmode}(k,l)\, .
    \end{align*}
constitutes a representation of $\Ghat_\Lambda(\su)$ on $\Mrep_{\hrep,\Lambda}^+$, i.e.\ the space of polynomials $\mathbb{C}[\{j_{kl},k_{kl},p_{kl},x_{mn}\}_{k,l,m,n\in\mathbb{Z}}]$, where the action of constraints is well defined.\footnote{The corresponding bracket relations were verified by using a SymPy Python script and by modifying the $\Lambda=0$ representation with the same replacement instead. }
\begin{remark}\label{rmk:modified_rep}
    Note that this representation is quite different from the original construction, as the $\Jmode^+$-operators do not annihilate $V$ anymore. Therefore, one would have to carefully examine its properties again and see how the constraints act. 
\end{remark}

\newpage
\pagestyle{fancy}
\thispagestyle{plain}
\appendix
\renewcommand{\sectionmark}[1]{
  \markright{\thesection\ #1}
} 
\newpage
\thispagestyle{plain}
\addcontentsline{toc}{section}{References}
\printbibliography

\end{document}